\newif\ifanon\anontrue
\newif\iflong\longtrue
\newif\ifproof\prooffalse
\newif\ificml\icmlfalse
\newif\ifspaa\spaatrue
\newif\ifready%
\documentclass[sigconf,nonacm]{acmart}

\usepackage{subfigure}
\usepackage{graphics}
\usepackage{graphicx}
\usepackage{booktabs}

\usepackage{amsthm}
\usepackage{amsmath}

\usepackage{relsize}
\usepackage{algorithm}
\usepackage{algorithmicx}
\usepackage[noend]{algpseudocode}
\usepackage{multicol}
\usepackage[most]{tcolorbox}
\usepackage{color}
\usepackage{pgfplots}
\usetikzlibrary{patterns}
\definecolor{mypurple}{HTML}{f875aa}
\definecolor{mypink}{HTML}{c6ffc1}
\definecolor{myred}{HTML}{3edbf0}
\definecolor{myorange}{HTML}{77acf1}
\definecolor{myblue}{HTML}{04009a}
\definecolor{mydarkpurple}{HTML}{E75480}
\definecolor{mydarkpink}{HTML}{22bc22}
\definecolor{mydarkred}{HTML}{8B0000}
\definecolor{mydarkorange}{HTML}{FF8C00}
\definecolor{mydarkblue}{HTML}{04009a}
\definecolor{mybluegreen}{HTML}{30BFBF}

\usepackage{caption}
\usepackage{placeins}
\pgfplotsset{
    compat=1.3,
    legend image code/.code={
        \draw [#1] (0cm,-0.1cm) rectangle (0.6cm,0.1cm);
    },
}

\usepackage{thm-restate}
\usepackage{enumerate}
\usepackage{enumitem}
\setlist{noitemsep,topsep=0pt,parsep=0pt,partopsep=0pt}
\usepackage{mathtools}
\usepackage{xspace}
\usepackage{booktabs}
\usepackage{pifont}
\makeatletter
\def\thm@space@setup{\thm@preskip=2pt
\thm@postskip=2pt}
\makeatother

\theoremstyle{plain}
\newtheorem{theorem}{Theorem}[section]

\newtheorem{invariant}{Invariant}
\newtheorem{lemma}[theorem]{Lemma}
\newtheorem{corollary}[theorem]{Corollary}
\newtheorem{definition}[theorem]{Definition}

\newtheorem{observation}[theorem]{Observation}

\newcommand{\defn}[1]{\textbf{\textit{#1}}}

\algrenewcommand\algorithmicindent{1em}%

\newcommand{\tO}{\widetilde{O}}

\newcommand{\kc}{$k$-core\xspace}

\DeclarePairedDelimiter\ceil{\lceil}{\rceil}
\DeclarePairedDelimiter\floor{\lfloor}{\rfloor}

\DeclareMathOperator{\poly}{poly}

\definecolor{mygreen}{RGB}{20,140,80}
\definecolor{linkcolor}{RGB}{0,0,230}
\definecolor{mylightgray}{RGB}{230,230,230}
\definecolor{verylightgray}{RGB}{245,245,245}

\newcommand{\etal}[0]{et al.\xspace}

\newcounter{myalgctr}

\newtcolorbox{OuterBox}[1][]{%
    breakable,
    enhanced,
    frame hidden,
    interior hidden,
    left=-5pt,
    right=-5pt,
    top=-5pt,
    float=p,
    boxsep=0pt,
    arc=0pt
#1}%

\newtcolorbox{InnerBox}[1][]{%
    enforce breakable,
    enhanced,
    colback=gray,
    colframe=white,
#1}%

\newenvironment{tbox}{
\vspace{0.2cm}
\begin{tcolorbox}[width=\columnwidth,
                  enhanced,
                  boxsep=2pt,
                  left=1pt,
                  right=1pt,
                  top=4pt,
                  boxrule=1pt,
                  arc=0pt,
                  colback=white,
                  colframe=black,
	              breakable
                  ]%
}{
\end{tcolorbox}
}

\newcommand{\tboxhrule}[0]{\vspace{0.1cm} {\color{black} \hrule} \vspace{0.2cm}}

\newenvironment{titledtbox}[1]{\begin{tbox}#1 \tboxhrule}{\end{tbox}}

\newcommand{\batch}{\mathcal{B}}

\newcommand{\tbl}{L}

\newcommand{\core}{k}
\newcommand{\greater}{U}
\newcommand{\upl}{up*}
\newcommand{\hd}{up-degree\xspace}
\newcommand{\ld}{\upl-degree\xspace}
\newcommand{\dl}{\mathsf{dl}\xspace}

\newcommand{\level}{\ell\xspace}

\newcommand{\eps}{\varepsilon}
\newcommand{\up}{\mathsf{up}}
\newcommand{\down}{\mathsf{up^*}}
\newcommand{\pred}{<}
\newcommand{\topmost}{T}
\newcommand{\group}{g}
\newcommand{\grp}{\group}
\newcommand{\gn}{\group n}
\newcommand{\coeff}{\left(2 + \add/\lambda\right)}

\newcommand{\prevind}{p}
\newcommand{\ind}{i}
\newcommand{\add}{3}
\newcommand{\kest}{\hat{\core}}
\newcommand{\neighbor}{N}

\newcommand{\arb}{\alpha}
\newcommand{\whp}{w.h.p.\xspace}
\newcommand{\mix}{\textbf{Mix}\xspace}
\newcommand{\expins}{\textbf{Ins}\xspace}
\newcommand{\expdel}{\textbf{Del}\xspace}
\newcommand{\hua}{Hua\xspace}
\newcommand{\sun}{Sun\xspace}
\newcommand{\zhang}{Zhang\xspace}
\newcommand{\ekcore}{ExactKCore\xspace}
\newcommand{\akcore}{ApproxKCore\xspace}
\newcommand{\plds}{PLDS\xspace}
\newcommand{\pldsopt}{PLDSOpt\xspace}
\newcommand{\lds}{LDS\xspace}
\newcommand{\dblp}{\textit{dblp}\xspace}
\newcommand{\lj}{\textit{livejournal}\xspace}

\newcommand{\wiki}{\textit{wiki}\xspace}

\newcommand{\tweet}{\textit{twitter}\xspace}
\newcommand{\frie}{\textit{friendster}\xspace}
\newcommand{\ctr}{\textit{ctr}\xspace}
\newcommand{\usa}{\textit{usa}\xspace}
\newcommand{\ytb}{\textit{youtube}\xspace}
\newcommand{\outnbr}{N^{+}}

\newcommand{\batchsize}{|\batch|}
\newcommand{\outdeg}{\alpha}

\newcommand{\incoming}{I}
\newcommand{\out}{X}

\newcommand{\static}{\mathsf{StaticMaximalMatching}\xspace}
\mathchardef\hyph="2D

\newcommand{\inn}{in-neighbors\xspace}
\newcommand{\outn}{out-neighbors\xspace}

\newcommand{\batchinsert}{\mathsf{BatchInsert}}
\newcommand{\batchdelete}{\mathsf{BatchDelete}}
\newcommand{\lowoutdegorient}{\mathsf{LowOutdegreeOrient}}
\newcommand{\delworkalpha}{W_{\mathit{del}}(\alpha)}
\newcommand{\deldepth}{D_{\mathit{del}}}
\newcommand{\insworkalpha}{W_{\mathit{ins}}(\alpha)}
\newcommand{\insdepth}{D_{\mathit{ins}}}
\newcommand{\batchdeletes}{\batch_{\mathit{del}}}
\newcommand{\batchinserts}{\batch_{\mathit{ins}}}
\newcommand{\batchedgeflip}{\mathsf{BatchFlips}}
\newcommand{\flipswork}{W_{\mathit{flips}}(\alpha)}
\newcommand{\flipsdepth}{D_{\mathit{flips}}}
\newcommand{\kcount}{\texttt{count}}

\ifready
\newcommand{\julian}[1]{{}}
\newcommand{\laxman}[1]{{}}
\newcommand{\qq}[1]{{}}
\newcommand{\shangdi}[1]{{}}
\newcommand{\jessica}[1]{{}}
\else
\newcommand{\julian}[1]{{\color{cyan} Julian: #1}}
\newcommand{\laxman}[1]{{\color{brown} Laxman: #1}}
\newcommand{\qq}[1]{{\color{magenta} Quanquan: #1}}
\newcommand{\shangdi}[1]{{\color{orange} Shangdi: #1}}
\newcommand{\jessica}[1]{{\color{blue} Jessica: #1}}
\fi

\newcommand{\myparagraph}[1]{\vspace{1pt}\noindent {\bf #1.}}

\usepackage[font=small,aboveskip=0pt,belowskip=0pt,tableposition=above]{caption}

\usepackage{footnote}
\makesavenoteenv{tabular}
\makesavenoteenv{table}
\usepackage{tablefootnote}

\usepackage[capitalize,nameinlink]{cleveref}
\crefname{theorem}{Theorem}{Theorems}
\Crefname{lemma}{Lemma}{Lemmas}
\Crefname{claim}{Claim}{Claims}
\Crefname{observation}{Observation}{Observations}
\Crefname{algorithm}{Algorithm}{Algorithms}
\Crefname{myalgctr}{Algorithm}{Algorithms}
\Crefname{invariant}{Invariant}{Invariant}
\Crefname{challenge}{Challenge}{Challenges}

\makeatletter
\renewcommand\theHALG@line{\thealgorithm.\arabic{ALG@line}}
\makeatother

\algblock{ParFor}{EndParFor}
\algblock{Input}{EndInput}
\algblock{Output}{EndOutput}
\algblock{ReduceAdd}{EndReduceAdd}

\algnewcommand\algorithmicparfor{\textbf{parfor}}
\algnewcommand\algorithmicinput{\textbf{Input:}}
\algnewcommand\algorithmicoutput{\textbf{Output:}}
\algnewcommand\algorithmicreduceadd{\textbf{ReduceAdd}}
\algnewcommand\algorithmicpardo{\textbf{do}}
\algnewcommand\algorithmicendparfor{\textbf{end\ input}}
\algrenewtext{ParFor}[1]{\algorithmicparfor\ #1\ \algorithmicpardo}
\algrenewtext{Input}[1]{\algorithmicinput\ #1}
\algrenewtext{Output}[1]{\algorithmicoutput\ #1}
\algrenewtext{ReduceAdd}[2]{#1 $\leftarrow$ \algorithmicreduceadd(#2)}
\algtext*{EndInput}
\algtext*{EndOutput}
\algtext*{EndIf}
\algtext*{EndFor}
\algtext*{EndWhile}
\algtext*{EndParFor}
\algtext*{EndReduceAdd}

\copyrightyear{2022}
\acmYear{2022}
\setcopyright{rightsretained}
\acmConference[SPAA '22]{Proceedings of the 34th ACM Symposium on Parallelism in Algorithms and Architectures}{July 11--14, 2022}{Philadelphia, PA, USA}
\acmBooktitle{Proceedings of the 34th ACM Symposium on Parallelism in Algorithms and Architectures (SPAA '22), July 11--14, 2022, Philadelphia, PA, USA}
\acmDOI{10.1145/3490148.3538569}
\acmISBN{978-1-4503-9146-7/22/07}

\usepackage{etoolbox}
\makeatletter
\patchcmd{\maketitle}{\@copyrightpermission}{
   \begin{minipage}{0.3\columnwidth}
     \href{https://creativecommons.org/licenses/by-nc-nd/4.0/}{\includegraphics[width=0.90\textwidth]{figures/cc-by-nc-nd4acm.png}}
   \end{minipage}\hfill
   \begin{minipage}{0.7\columnwidth}
     \href{https://creativecommons.org/licenses/by-nc-nd/4.0/}{This work is licensed under a Creative Commons Attribution-NonCommercial-NoDerivs International 4.0 License.}
   \end{minipage}

   \vspace{5pt}
}{}{}

\makeatother

\begin{document}

\title{Parallel Batch-Dynamic Algorithms for $k$-Core Decomposition and Related Graph Problems}
\thanks{$^\dag$This work was done while the authors were at MIT CSAIL}

\author{Quanquan C. Liu$^\dag$}
\affiliation{\institution{Northwestern University} \country{USA}}
\email{quanquan@northwestern.edu}
\author{Jessica Shi}
\affiliation{\institution{MIT CSAIL}\country{USA}}
\email{jeshi@mit.edu}
\author{Shangdi Yu}
\affiliation{\institution{MIT CSAIL}\country{USA}}
\email{shangdiy@mit.edu}
\author{Laxman Dhulipala$^\dag$}
\affiliation{\institution{University of Maryland}\country{USA}}
\email{laxman@umd.edu}
\author{Julian Shun}
\affiliation{\institution{MIT CSAIL}\country{USA}}
\email{jshun@mit.edu}

\renewcommand{\shortauthors}{Quanquan C. Liu, Jessica Shi, Shangdi Yu, Laxman Dhulipala, Julian Shun}

\begin{abstract}
  Maintaining a $k$-core decomposition quickly in a dynamic graph has
  important applications in network analysis.
  The main challenge for
  designing efficient \emph{exact} algorithms is that a
  single update to the graph can cause significant global changes. 
  Our paper focuses on \emph{approximation} 
  algorithms with small approximation factors that are much more efficient
  than what exact algorithms can obtain.

    We present the first parallel,
    batch-dynamic algorithm for approximate $k$-core
    decomposition that is efficient in both theory and practice. Our algorithm is based on our novel
    \emph{parallel level data structure}, inspired by the sequential level data structures of Bhattacharya \etal [STOC '15] and 
    Henzinger \etal [2020].
    Given a graph with $n$ vertices and a batch of updates $\batch$, our algorithm provably maintains a $(2 + \varepsilon)$-approximation of
    the coreness values of all vertices (for any constant $\varepsilon > 0$)
    in $O(|\mathcal{B}|\log^2 n)$ amortized work and $O(\log^2 n \log\log n)$
    depth (parallel time) with high probability. 
    
    As a by-product, our $k$-core decomposition algorithm also gives a batch-dynamic algorithm for maintaining an $O(\alpha)$ out-degree 
    orientation, where $\alpha$ is the \emph{current} arboricity 
    of the graph. We demonstrate the usefulness of our low out-degree orientation
    algorithm by presenting a new framework to formally study 
    batch-dynamic  algorithms in bounded-arboricity graphs.
    Our framework obtains new provably-efficient
    parallel batch-dynamic algorithms 
    for 
    maximal matching, clique counting, and vertex coloring.

    We implemented and experimentally evaluated our $k$-core decomposition 
    algorithm on a 30-core machine with two-way hyper-threading on $11$ graphs
    of varying densities and sizes. Compared to the
    state-of-the-art algorithms, our algorithm achieves
    up to a $114.52\times$ speedup against the best parallel implementation,
    up to a $544.22\times$ speedup against the best approximate
    sequential algorithm, and up to a $723.72\times$
    speedup against the best exact sequential algorithm. 
    We also obtain results for our algorithms
    on graphs that are orders-of-magnitude larger than those used in
    previous studies.
\end{abstract}

\settopmatter{printfolios=true}
\maketitle

\section{Introduction}
Discovering the structure of large-scale networks is a fundamental
problem for many areas of computing. One of the key challenges is to detect communities in which
individuals (or vertices) have close ties with one another, and to
understand how well-connected a particular individual is to the
community. 
The well-connectedness of a vertex or a
group of vertices is naturally captured by the concept of a \kc or, more
generally, the \kc decomposition; hence, this particular problem and
its variants have been widely studied in the machine
learning~\cite{Alvarez2005,Esfandiari2018,Ghaffari2019},
database~\cite{Chu20, LiZZ19,ESTW19,BGKV14,MMSS20}, social network analysis and graph analytics~\cite{KBS15, Kabir2017, dhulipala2017julienne,
dhulipala2018theoretically}, computational biology~\cite{ciaperoni2020,kitsak2010,LiuTang15,malliaros2016},
and other communities~\cite{GBGL20,
KBS15,Luo19,SGJ13}.

Given an undirected graph $G$, with $n$ vertices and $m$ edges, the
\kc of the graph is the maximal subgraph $H \subseteq G$ such that
the induced degree of every vertex in $H$ is at least $k$. The \kc
decomposition of the graph is defined as a partition of the graph
into layers such that a vertex $v$ is in layer $k$ if it belongs to
a \kc but not a $(k+1)$-core; this value is known as the \emph{coreness} of the vertex,
and the coreness values induce a natural hierarchical clustering.
Classic algorithms for \kc decomposition are inherently
   sequential.  A well-known algorithm for finding the decomposition
   is to iteratively select and remove all vertices $v$ with smallest
   degree from the graph until the graph is empty~\cite{Matula83}.
   Unfortunately, the
   length of the sequential dependencies, or the \emph{depth}, of such a
   process can be $\Omega(n)$ given a graph with $n$ vertices. 
   As \kc decomposition
   is a $\mathsf{P}$-complete problem~\cite{anderson84pcomplete}, it
   is unlikely to have a parallel algorithm with polylogarithmic depth.
   To obtain parallel methods with $\poly(\log n)$ depth,
   we relax the condition
   of obtaining an \emph{exact} decomposition to one of obtaining
   a close \emph{approximate} decomposition.

   Previous works studied approximate $k$-core decompositions as a
   way for obtaining faster and more scalable algorithms in larger graphs
   than in exact settings~\cite{chan2021,Ghaffari2019,SCS20,Esfandiari2018,Chu20}. 
   Approximate coreness values are useful for applications where
   existing methods are already approximate, such as
   diffusion protocols in epidemiological studies
   \cite{ciaperoni2020,kitsak2010,LiuTang15,malliaros2016},
   community detection and network centrality
   measures~\cite{dourisboure09,Fang17,Healy07,Mitzenmacher2015,WangCao18,ZhangYing17},
   network visualization and
   modeling~\cite{Alvarez2005,carmi2007,yangdefining2015,ZhangZhao10},
   protein interactions~\cite{Amin2006,baderautomated2003}, and
   clustering~\cite{GiatsidisMTV14,lee2010}.
   Furthermore, due to the rapidly changing nature of today's large
   networks, many recent studies have focused on the \emph{dynamic}
   setting, where edges and vertices can be inserted and deleted, and the $k$-core
   decomposition is computed in real time. 
   There has been significant interest in obtaining fast and practical
dynamic, approximate and exact \kc algorithms.
Dynamic algorithms have been developed for
both the sequential~\cite{Li2014,
Sariyuce2016,Zhang2017,Wen2019,Li2014,SCS20, lin2021dynamickcore} and
parallel~\cite{Hua2020,Jin2018,Aridhi2016} settings.
There has also been interest in the
closely-related dynamic $k$-truss problem~\cite{Huang2014, Akbas2017,
Zhang2019a, Luo2021}. However, to the best of
our knowledge, \emph{there are no existing parallel batch-dynamic \kc
algorithms with provable polylogarithmic depth}, which our algorithm achieves.

   Our paper focuses on
   the \emph{batch-dynamic} setting where
   updates are performed over a batch of
   \emph{multiple} edge updates applied simultaneously. Such a
   setting is conducive to parallelization, which we leverage
   to obtain scalable algorithms.
   We provide a work-efficient batch-dynamic approximate $k$-core decomposition algorithm based on a  parallel level data structure that we design. We implement our algorithm and show experimentally that it performs favorably compared to the state-of-the-art. Furthermore, we show that our parallel level data structure can be used to obtain work-efficient parallel batch-dynamic algorithms for several other problems, specifically, low out-degree orientation, maximal matching, $k$-clique counting, and vertex coloring.
   
   We introduce the necessary definitions in~\cref{sec:prelims} before giving a technical overview of our results in~\cref{sec:tech}.~\cref{sec:bounded-depth-arboricity} presents our parallel level data structure and $k$-core decomposition algorithm in more detail.~\cref{sec:eval} presents experimental results.
   \cref{sec:static-kcore} gives our parallel, static, approximate algorithm for \kc decomposition.
   Finally,~\cref{sec:framework} gives our low out-degree framework for our maximal matching (\cref{sec:matching}), 
   $k$-clique counting (\cref{sec:clique}), and coloring (\cref{sec:coloring}) results.  
   
   \section{Preliminaries}\label{sec:prelims}
This paper studies undirected, unweighted graphs, and we use $n$ to denote the number of vertices and $m$ to denote the number of edges in a graph.
\cref{def:approx-k-core} defines 
approximate $k$-core decomposition. The definition requires the
definition of a $k$-core, which we define first.

\begin{definition}[$k$-Core]\label{def:k-core}\label{def:k-shell}
  For a graph $G$ and positive integer $k$, the \defn{$k$-core} 
  of $G$ is the maximal subgraph of $G$ with minimum induced degree $k$.
\end{definition}

\begin{definition}[\unboldmath{$k$}-Core Decomposition]\label{def:k-core-decomp}
A \defn{\kc decomposition} is a partition of vertices into layers
such that a vertex $v$ is in layer $k$ if it
belongs to a $k$-core but not to a $(k + 1)$-core. $k(v)$
denotes the layer that vertex $v$ is in, and is called the
\defn{coreness} of $v$.
\end{definition}

Definition~\ref{def:k-core-decomp} defines an \emph{exact} $k$-core
decomposition.
A \emph{$c$-approximate} $k$-core decomposition is defined as follows.

\begin{definition}[\unboldmath{$c$}-Approximate \unboldmath{$k$}-Core
    Decomposition]\label{def:approx-k-core}
    A \defn{$c$-approximate} \defn{\kc decomposition} 
    is a partition of
    vertices into layers such that a vertex $v$ is 
    in layer $k'$ only if
    $\frac{k(v)}{c} \leq k' \leq ck(v)$, where $k(v)$ is the coreness of $v$.
\end{definition}

We let $\kest(v)$ denote the \emph{estimate} of $v$'s coreness.
\cref{fig:approx-def} shows an example of a $k$-core decomposition and a
$(3/2)$-approximate $k$-core decomposition.

\begin{figure}[!t]
    \centering
    \includegraphics[width=0.7\columnwidth]{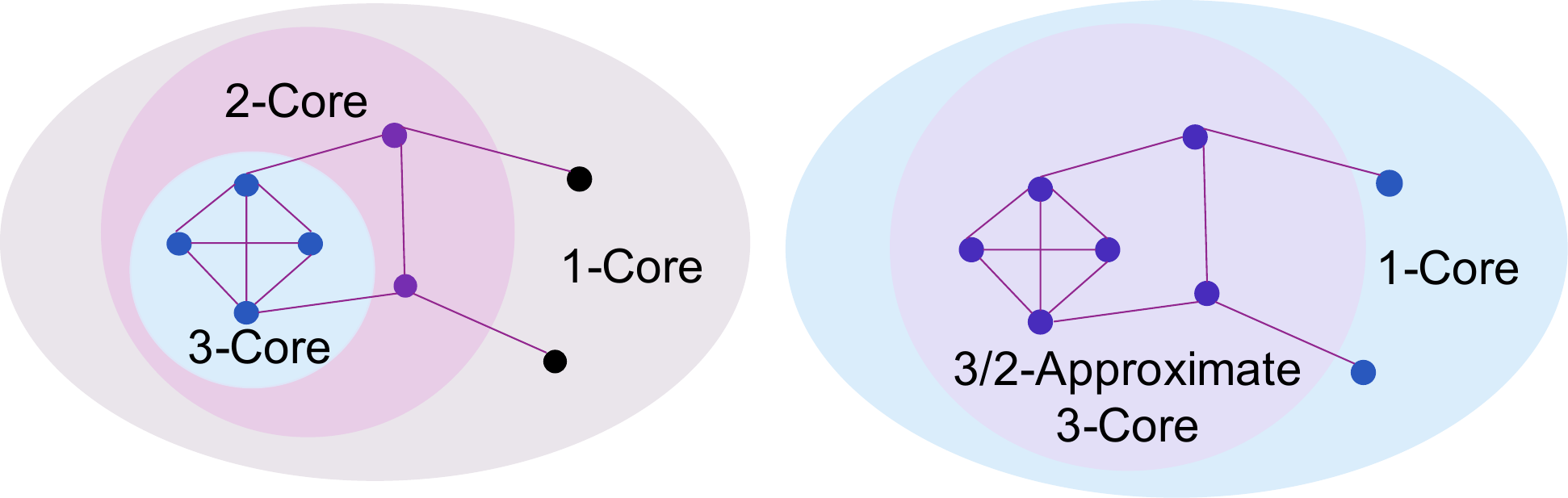}
    \caption{Exact $k$-core decomposition (left) and $(3/2)$-approximate $k$-core
    decomposition (right).}\label{fig:approx-def}
\end{figure}

\begin{table}[t]
\begin{center}
    \footnotesize
    \begin{tabular}{c c}
    \toprule
    Symbol & Meaning \\
    \midrule
    $G = (V, E)$ & undirected/unweighted graph\\
    \hline
    $n, m$ & number of vertices, edges, resp.\ \\
    \hline
    $\arb$ & current arboricity of graph \\
    \hline
    $\Delta$ & current maximum degree of graph\\
    \hline
    $K$ & number of levels in PLDS\\
    \hline
    $\neighbor(v)$ (resp.\ $\neighbor(S)$)
    & set of neighbors of vertex $v$ (resp.\ vertices $S$)\\
    \hline
    $\dl(v)$ & \emph{desire-level} of vertex $v$\\
    \hline
    $\ell$, $\ell(v)$ & a level (starting with level $\ell = 0$), current level of vertex $v$, resp.\ \\
    \hline
    $V_{\ell}$, $Z_{\ell}$ & set of vertices in level $\ell$, set of vertices in levels $\geq \ell$, resp.\ \\
    \hline
    $\group_i$ & set of levels in group $i$ (starting with $g_0$)\\
    \hline
    $\group(v), \gn(\level)$ & \emph{group number} of vertex $v$, index $i$ where level $\ell \in \group_i$, resp.\\
    \hline
    $\core(v)$, $\kest(v)$ & coreness of $v$, 
    estimate of the coreness of $v$, resp.\ \\
    \hline
    $\up(v)$,$\down(v)$ & \emph{\hd} of $v$, \emph{\ld} of $v$, resp.\ \\
    \hline
\iflong
    $\eps, \lambda, \delta$ & constants where $\eps, \lambda,\delta > 0$ \\ 
    \fi
    \bottomrule
    \end{tabular}
\end{center}
\caption{Table of notations used in this paper.}\label{tbl:notation}
\end{table}

\myparagraph{Model Definitions}
We analyze the theoretical efficiency of our parallel algorithms in
the \emph{work-depth model}.
The model is
defined in terms of two complexity measures,
\defn{work} and \defn{depth}~\cite{JaJa92,CLRS}.
The \defn{work} is the total number of operations
executed by the algorithm. The \defn{depth} is the longest
chain of sequential dependencies. We assume that concurrent reads and
writes are supported in $O(1)$ work/depth. A \defn{work-efficient} parallel algorithm is one with
work that asymptotically matches the best-known sequential time complexity for
the problem.
We say that a bound holds \defn{with high probability (\whp)} if it holds
    with probability at least $1 - 1/n^c$ for any
    $c \geq 1$.

We use parallel primitives in our algorithms, which take as input a sequence $A$ of length $n$, including: parallel \defn{reduce-add}, which
returns the sum of the entries in $A$, and parallel \defn{filter}, which also takes as input a predicate function
$f$, and returns the sequence $B$ containing each element $a \in A$ where
$f(a)$ is true, while preserving the same relative order as the
order of elements in $A$. These primitives take $O(n)$ work and $O(\log n)$ depth~\cite{JaJa92}.
We also use \defn{parallel hash tables} that support insertions, deletions, and membership queries; they can perform $n$ insertions or deletions in $O(n)$ work and $O(\log^* n)$ depth \whp, and $n$ membership queries in $O(n)$ work and $O(1)$ depth \whp~\cite{Gil91a}.
Provided an input sequence $A$, a parallel \defn{prefix-sum} takes as input an 
identity $x$ and an associative binary
operator $\oplus$, and returns the sequence $B$ of length $n$ where
$B[i] = \bigoplus_{j < i} A[j] \oplus x$. This primitive
takes $O(n)$ work and $O(\log n)$ depth~\cite{JaJa92}.

Our parallel
algorithms operate in the batch-dynamic setting. A
\defn{batch-dynamic} algorithm processes updates (vertex or edge
insertions/deletions) in batches $\batch$ of size
$|\batch|$. For simplicity, since we can reprocess the graph using an
efficient parallel static algorithm when $|\batch| \geq m$, we consider
$1 \leq |\batch| < m$ for our bounds.

Given a graph $G = (V, E)$ and a sequence of
batches of edge insertions and deletions, $\batch_1, \dots, \batch_N$,
where $\batch_i = (E^i_{\textit{delete}}, E^i_{\textit{insert}})$,
the goal is to efficiently
maintain a $(2+\eps)$-approximate $k$-core decomposition (for any
    constant $\eps > 0$) after applying each batch $\batch_i$ (in order)
    on $G$. In other words, let $G_i = (V, E_i)$ be the graph after
    applying batches $\batch_1, \dots, \batch_i$ and suppose that we have
    a $(2+\eps)$-approximate \kc decomposition on $G_i$;
    then, for $\batch_{i+1}$, our goal is to efficiently
    find a $(2+\eps)$-approximate $k$-core
    decomposition of $G_{i+1} = (V, (E_i \cup E^{i+1}_{\textit{insert}}) \setminus
    E^{i+1}_{\textit{delete}})$.
\\\\
All notations used are summarized in~\cref{tbl:notation}.
Our data structure also maintains a \emph{low 
out-degree orientation}, which may be parameterized by a graph property
known as the \emph{arboricity}.

\begin{definition}[Arboricity]\label{def:arboricity}
The arboricity ($\alpha$) of a graph is the minimum number
of spanning forests needed to cover the graph.
\end{definition}

\begin{definition}[$c$-Approximate Low Out-Degree Orientation]\label{def:loo}
    Given an undirected graph $G = (V, E)$, 
    a $c$-approximate low out-degree orientation is an acyclic orientation
    of all edges in $G$ such that the maximum out-degree of any vertex, $d^{+}_{max}$, is within a $c$-factor of the minimum possible maximum 
    out-degree, $d^{+}_{opt}$ of any acyclic orientation:\footnote{$d^{+}_{opt}$, is equal to the \emph{degeneracy}, $d$, of $G$, and is closely related to $\alpha$: $d/2 \leq \alpha \leq d$.}
        $d^{+}_{opt}/c \leq d^{+}_{max} \leq c \cdot d^{+}_{opt}$.
\end{definition}

We define an \defn{$O(\alpha)$ out-degree orientation} to be
an acyclic orientation where all out-degrees are $O(\alpha)$.
For an oriented graph, we call neighbors of vertex $v$ connected by outgoing edges the \defn{\outn} of $v$ and neighbors of $v$ connected by incoming 
edges the \defn{\inn} of $v$.
\ifspaa
Definitions of the other problems we consider are given at the top of their 
respective sections (\cref{sec:matching,sec:coloring,sec:clique}).
\fi
   
   \section{Technical Overview}\label{sec:tech}
   In this paper, we provide a number of parallel work-efficient 
   algorithms for various problems. This section gives an overview of our algorithms and how they compare with prior work.
   \cref{table:results} summarizes our algorithmic results.

We first discuss $k$-core decomposition.
A number of previous works~\cite{LiYu14,Luo21hypergraph,Sariyuce2016,Zhang17,Zhang2019a} 
provided methods for maintaining the \emph{exact}
\kc decomposition 
under single edge updates in the sequential setting. 
Unfortunately, none 
of these works provide algorithms with provable 
polylogarithmic update
time. The main bottleneck for obtaining 
\emph{provably-efficient} methods is that a single edge
update can cause \emph{all} coreness values to change: consider a cycle
with one edge removed as a simple example. Removing and adding the edge
into this cycle, repeatedly in succession, causes the coreness of all 
vertices to change by one with each update. 
In the parallel setting, a number of previous works~\cite{Hua2020,Aridhi2016,WangYu17,Jin2018,Gabert21}
investigated batch-dynamic algorithms for exact \kc decomposition. 
Unfortunately, none of these works have $\poly(\log n)$ depth
and some even have $\Omega(n)$ depth. 

This paper shows that we can surprisingly obtain a parallel batch-dynamic $k$-core decomposition algorithm with amortized time bounds that are
independent of the number of vertices that \emph{changed coreness} 
for \emph{approximate} coreness. Such
provable time bounds can be obtained by cleverly avoiding updating
coreness values until enough error has accumulated; once such error
has accumulated, we can charge the amount of time required to update
the coreness to the number of updates that occurred.  
Doing so carefully allows
a provable $O(\log^2 n)$ amortized work per update that is
independent of the number of changed coreness values. 
A recent paper by Sun \etal~\cite{SCS20} provides a \emph{sequential}
dynamic approximate $k$-core decomposition algorithm that takes
$O(\log^2n)$ amortized time per update.
Their algorithm is a threshold peeling/elimination procedure 
that gives a $(2+\eps)$-approximation bound.
They also provide another sequential algorithm, which they call \emph{round-indexing}, that performs faster in practice.\footnote{Our experiments compare against the round-indexing algorithm since it is faster than their thresholding peeling algorithm in practice.} However, they do not provide 
formal runtime proofs for this algorithm.
Their threshold peeling algorithm is inherently sequential since a vertex
that changes thresholds can cause another to change their threshold 
(and coreness estimate),
resulting in a long chain of sequential dependencies; such a situation
results in polylogarithmic \emph{amortized} depth, 
whereas efficient parallel algorithms require
polylogarithmic depth \whp in the \emph{worst case}, which we obtain.

To design our $k$-core decomposition algorithm, we formulate a \emph{parallel
   level data structure (PLDS)} inspired by the sequential level data structures (LDS) of
   Bhattacharya \etal~\cite{BHNT15} and Henzinger \etal~\cite{HNW20} to maintain a
   partition of the vertices satisfying specific degree properties in certain induced
   subgraphs. 
In the LDS, vertices are updated one
at a time. One of our main technical insights is that we can update many vertices 
\emph{simultaneously}, leading to high parallelism. 
Our \kc decomposition algorithm is work-efficient, 
   and matches the 
   approximation factor of the best-known sequential dynamic
   approximate \kc decomposition algorithm of Sun~\etal~\cite{SCS20}, while
   achieving polylogarithmic depth \whp
   
Dynamic problems related to \kc decompositions have been recently
studied in the theory community, such as densest subgraph~\cite{BHNT15,SWSoda20} and low
out-degree orientations~\cite{BB20,HeTZ14,brodal99dynamicrepresentations,Kowalik07,KKPS14,KS18,HNW20,SW20}; some of these works use
the LDS. However, none of these
previous works proved guarantees regarding the $k$-core
decomposition that can be maintained via a LDS.
Notably, we show via a new, intuitive
proof that one can use the level of a
vertex to estimate its coreness in the
LDS of~\cite{HNW20}. Unlike the proof 
in~\cite{SCS20} for their dynamic algorithm, 
our proof does not require densest 
subgraphs nor any additional information besides the two invariants 
maintained by the structure. 

Our main theoretical and practical technical contributions 
for \kc decomposition are three-fold: (1) we present a simple 
modification and a new $(2+\eps)$-approximate coreness proof for the 
sequential level data structure of~\cite{BHNT15,HNW20} (which were not previously used for coreness values)
using only the levels of the vertices---no such 
modification was known prior to this work since~\cite{SCS20} 
requires an additional elimination/peeling/round-indexing procedure; (2) we provide the first parallel work-efficient
batch-dynamic level data structure that takes $O(\log^2
n\log\log n)$  depth \whp, which we use to obtain a $(2+\epsilon)$-approximate batch-dynamic $k$-core decomposition algorithm;
and (3) we provide multicore implementations of
our new algorithm
and demonstrate its
practicality through extensive experimentation with state-of-the-art parallel and sequential algorithms. 

\begin{table}[t!]
    \centering
    \caption{Work and depth bounds of algorithms in this paper.\protect\footnotemark}\label{table:results}
    \setlength\tabcolsep{3pt}
    \footnotesize
    \begin{tabular}{*5c}
        \toprule
        Problem & Approx & Work  & Depth & Adversary \\    
        \midrule
        \kc &  $(2+\eps)$ & $O(\batchsize \log^2 n)$ & $\tO(\log^2 n)$\protect \footnotemark & Adaptive\\
        \kc & $(2 + \eps)$ & $O(m+n)$ & $\tO(\log^3 n)$ & Static\\
        Orientation &  $(4+\eps)$    &  $O(\batchsize \log^2 n)$  & $\tO(\log^2 n)$ & Adaptive \\
        Matching &   Maximal    & $O(\batchsize (\alpha + \log^2 n))$ & $\tO(\log \Delta \log^2 n)$\protect \footnotemark[6] & Adaptive \\
        $k$-clique & Exact &    $O(\batchsize\alpha^{k-2}\log^2 n)$   &   $\tO(\log^2 n)$  & Adaptive   \\
        Coloring & $O(\alpha \log n)$\protect \footnotemark & $O(\batchsize\log^2 n)$ &  $\tO(\log^2 n)$ & Oblivious   \\
        Coloring & $O\left(2^{\alpha}\right)$ & $O(\batchsize\log^3 n)$ & $\tO(\log^2 n)$ & Adaptive\\
        \bottomrule
    \end{tabular}
    \vspace{-1em}
\end{table}
\addtocounter{footnote}{-2}
\footnotetext[\thefootnote]{All bounds are \whp, except for the work of static $k$-core and $O(\alpha \log n)$-coloring.}
\addtocounter{footnote}{1}
\footnotetext[\thefootnote]{$\tO$ hides a factor of $O(\log\log n)$.}
\addtocounter{footnote}{1}
\footnotetext[\thefootnote]{We denote  by $\alpha$ the \emph{current} 
arboricity of the graph \emph{after 
processing all updates including the most recent 
ones}.}

The following theorems give our theoretical bounds.

\begin{theorem}[Batch-Dynamic $k$-Core Decomposition]\label{thm:batch-dynamic}
   Given $G = (V, E)$ where $n = |V|$ and batch of updates $\batch$,
   our algorithm maintains $(2 + \eps)$-approximations of
    core values for all vertices (for any constant $\eps > 0$)
    in $O(|\batch|\log^2 n)$ amortized work and $O(\log^2 n \log\log n)$
    depth \whp, using $O(n\log^2 n + m)$ space.
\end{theorem}

Using the same parallel level data structure, we also obtain the following 
result for maintaining a low out-degree orientation. 

\begin{theorem}[Batch-Dynamic Low Out-Degree 
Orientation]\label{thm:low-outdegree}
Our algorithm maintains an $(4 + \eps)$-approximation 
of a minimum acyclic out-degree orientation, 
with the same bounds as~\cref{thm:batch-dynamic}, where the amortized
number of edge flips is $O(\batchsize \log^2 n)$.
\end{theorem}

A consequence of~\cref{thm:low-outdegree} is the following corollary.

\begin{corollary}[$O(\alpha)$ Out-Degree Orientation]\label{cor:arboricity-orientation}
Our algorithm maintains an $O(\alpha)$ out-degree orientation, where $\alpha$ is the current arboricity (\cref{def:arboricity}),
with the same bounds as~\cref{thm:low-outdegree}.
\end{corollary}

By the Nash-Williams theorem and the relationship between degeneracy, arboricity, and the maximum core number, we also obtain a batch-dynamic $(4+\eps)$-approximate densest subgraph algorithm 
which returns an approximate
value of the densest subgraph. This algorithm immediately follows from our 
core decomposition algorithm in which we return our maximum core number
as the estimate for the densest subgraph value.

\begin{corollary}[Densest Subgraph Value]
  Given $G = (V, E)$ with $n = |V|$ vertices and $m = |E|$ edges,
  for any constant $\eps > 0$, our
  algorithm finds an $(4+\eps)$-approximate densest subgraph value with the same bounds as in~\cref{thm:batch-dynamic}.
\end{corollary}

Using \cref{thm:low-outdegree}, we design a framework for parallel batch-dynamic algorithms on bounded-arboricity graphs 
for batch of updates $\batch$, which
in addition to problem-specific 
techniques allows us to 
obtain a set of batch-dynamic 
algorithms for a variety of other fundamental
graph problems including maximal matching, 
clique counting, and vertex
coloring. The coloring algorithms are based heavily on the 
sequential algorithms of Henzinger et al.~\cite{HNW20}, 
but we present them as an application of our framework. 
For the problems we consider in this paper, \cref{table:sequential} 
summarizes the update times of the previous best-known sequential results
for their respective settings.

\begin{theorem}[Batch-Dynamic Maximal Matching]\label{thm:mm}
We maintain a maximal matching in
$O(\batchsize(\alpha + \log^2n))$ amortized work and
$O(\log^2 n \left(\log \Delta + \log\log n\right))$ depth \whp,\footnote{$\Delta$ 
denotes the maximum \emph{current} degree of the graph \emph{after} processing
all updates.} in $O(n\log^2 n + m)$ space.
\end{theorem}

\begin{theorem}[Batch-Dynamic Implicit $O(2^{\alpha})$-Vertex Coloring]\label{thm:implicit}
We maintain an implicit $O(2^{\alpha})$-vertex coloring\footnote{An \emph{implicit} vertex coloring algorithm returns valid colorings for queried vertices.} in
$O(\batchsize \log^3 n)$ amortized work and 
$O(\log^2 n)$ depth \whp for updates, and $O(Q\alpha \log n)$ 
work and $O(\log n)$ depth
\whp, for $Q$ queries, using $O(n\log^2 n + m)$ space.
\end{theorem}

\begin{theorem}[Batch-Dynamic $k$-Clique Counting]\label{thm:clique}
We maintain the count of $k$-cliques in 
$O(\batchsize \alpha^{k-2}\log^2 n)$
amortized work and $O(\log^2 n \log\log n)$ depth \whp, in $O(m\alpha^{k-2} + n\log^2 n)$ space.
\end{theorem}

All of the above results are robust against \emph{adaptive} adversaries which have
access to the algorithm's previous outputs. 
The following algorithm is robust against \emph{oblivious} adversaries
which do not have access to previous outputs.

\begin{theorem}\label{thm:alogn-coloring}
We maintain an $O(\alpha \log n)$-vertex coloring in 
$O(\batchsize \log^2 n)$ amortized expected
work and $O(\log^2 n \log\log n)$ depth \whp, in $O(m+n\log^2 n + \alpha\log n)$ space. 
\end{theorem}

Our \kc, low out-degree orientation,
and vertex coloring algorithms are
work-efficient when compared to the 
best-known sequential, dynamic algorithms
for the respective problems~\cite{BHNT15,HNW20,SCS20}. 
For maximal matching, our algorithm is work-efficient
when $\alpha = \Omega(\log^2 n)$ when compared to the 
best-known sequential algorithm 
that is robust against adaptive 
adversaries~\cite{HeTZ14,neiman2015simple}; the extra
work when $\alpha = o(\log^2 n)$ 
comes from the fact that our bounds are with respect to the
\emph{current} arboricity, compared to~\cite{HeTZ14,neiman2015simple} whose bounds are with respect
to the \emph{maximum} arboricity over the sequence of updates.

The best-known batch-dynamic algorithm for 
$k$-clique counting, by Dhulipala et al.~\cite{DLSY21}, 
takes $O(\batchsize m \alpha^{k-4})$ expected work and 
$O(\log^{k-2}n)$ depth \whp,
using $O(m + \batchsize)$ space. Compared with their algorithm, 
our algorithm uses less work
when $m = \omega(\alpha^2 \log^2 
n)$. In many real-world networks,
$\alpha << \sqrt{m}$ (see e.g.,~\cref{table:sizes}, 
for maximum \kc values, which upper bound $\alpha$); thus,
our result is more efficient in many cases
at an additional multiplicative space cost of 
$O(\alpha^{k-2})$. We also obtain smaller
depth for all $k > 4$.
We provide further comparisons with the 
best-known sequential clique counting
algorithm~\cite{DvorakT13}, and
we describe more specific batch-dynamic challenges we face in designing the above algorithms in their respective sections.
The components of the PLDS used in each of the above results are summarized in~\cref{fig:usage}.

Finally, using ideas from our batch-dynamic \kc decomposition algorithm, 
we provide a new parallel
static $(2+\eps)$-approximate \kc decomposition algorithm.
We compare this algorithm with the best-known parallel
static exact algorithm of~\cite{dhulipala2017julienne}
which uses $O(m + n)$ expected work and $O(\rho \log m)$ depth \whp, where $\rho$
is the \emph{number of steps necessary to peel all vertices} 
($\rho$ could potentially be $\Omega(n)$). Hence, \cite{dhulipala2017julienne}
does not guarantee $\poly(\log n)$ depth.

\begin{theorem}\label{thm:static}
  Given $G = (V, E)$ with $n = |V|$ vertices and $m = |E|$ edges,
  for any constant $\eps > 0$, our
  algorithm finds an $(2+\eps)$-approximate \kc
  decomposition in $O(n+m)$ expected work and $O(\log^3 n)$ depth
  \whp, using $O(n+m)$ space.
\end{theorem}

\begin{figure}[!t]
    \centering
    \includegraphics[width=0.8\columnwidth]{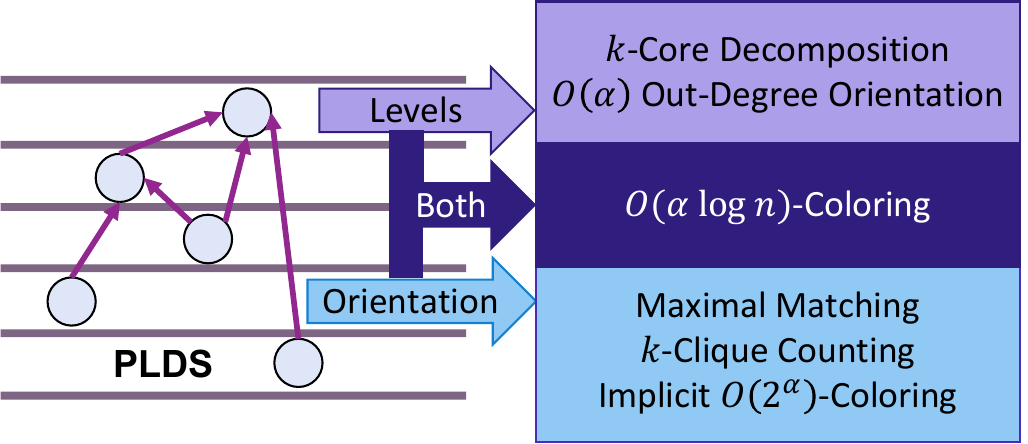}
    \caption{This figure shows what parts of the PLDS are used in each result. The level of each 
    vertex is used to determine the $k$-core decomposition (\cref{thm:batch-dynamic}) and low out-degree orientation (\cref{thm:low-outdegree} and \cref{cor:arboricity-orientation}). 
    The orientation of the edges
    is used for maximal matching (\cref{thm:mm}), implicit $O(2^{\alpha})$-coloring (\cref{thm:implicit}), and $k$-clique counting (\cref{thm:clique}). Finally,
    both are used for $O(\alpha \log n)$-coloring (\cref{thm:alogn-coloring}).}\label{fig:usage}
\end{figure}

\myparagraph{Experimental Contributions}
In addition to our theoretical contributions, we also provide
optimized multicore implementations of our \kc decomposition
algorithms.
We compare the performance of our algorithms
with state-of-the-art algorithms on a variety of real-world graphs
using a 30-core machine with two-way hyper-threading.
Our parallel
static approximate \kc algorithm achieves a 2.8--3.9x speedup over the
fastest parallel exact \kc
algorithm~\cite{dhulipala2017julienne} and
achieves a 14.76--36.07x self-relative speedup.

We show that our parallel
batch-dynamic \kc algorithm achieves up to $544.22\times$ speedups over the state-of-the-art sequential dynamic approximate \kc algorithm of Sun \etal~\cite{SCS20}, while
achieving comparable accuracy.
We also achieve up to $114.52\times$ speedups
over the state-of-the-art parallel
batch-dynamic exact \kc algorithm of Hua \etal~\cite{Hua2020},
and up to $723.72\times$ speedups against the state-of-the-art sequential
exact \kc algorithm of Zhang and Yu~\cite{Zhang2019a}.
Our batch-dynamic algorithm outperforms the best multicore static \kc algorithms by
up to $121.76\times$ on batch sizes that are less than $1/3$ of the number of
edges in the entire graph. 

Our algorithm exhibits improvements in runtime while
maintaining the same or smaller error, 
even when using only four threads (available on a standard laptop), 
and remains competitive at one thread.
We demonstrate that existing exact dynamic
implementations are not efficient or scalable enough to handle graphs
with billions of edges, whereas our algorithm is able to. Furthermore, our demonstrated speedups of up to two orders of magnitude indicates
that our implementation
not only fills the gap for processing graphs that are orders of magnitude 
larger than can be handled by existing implementations, but also
that it is the best option for many smaller networks.
Our code is publicly available at \url{https://github.com/qqliu/batch-dynamic-kcore-decomposition}. 
\section{Comparisons with Other Related Work}\label{app:related}

\paragraph{Parallel Exact Batch-Dynamic Algorithms}
The most recent,
state-of-the-art parallel batch-dynamic algorithm by Hua \etal~\cite{Hua2020} improves upon the previous parallel algorithms of
Aridhi \etal~\cite{Aridhi2016}, Wang \etal~\cite{WangYu17}, and Jin
\etal~\cite{Jin2018}.
Their algorithm relies on the concept of a \emph{joint edge set},
whose insertion and removal determines the core numbers of the vertices.
However, their algorithm could take $\Omega(n)$ depth as they use a 
standard depth-first search to traverse vertices in the joint edge set as
well as vertices outside
of the joint edge set. In comparison, our algorithm provably has $O(\log^2 
n\log\log n)$ depth \whp Our theoretical improvements also translate to 
practical gains since we demonstrate greater
scalability in our experiments.

Another recent work by Gabert \etal~\cite{Gabert21} provides a scalable 
exact $k$-core maintenance algorithm. Both their asymptotic work 
and depth is super-polylogarithmic (in fact, in the worst case it could be 
as bad as computing from scratch). Unfortunately, the code for their 
experiments is proprietary and hence not available for comparison. However, 
their reported
experimental results overall appear slower than our results, 
described in more detail in~\cref{sec:threads}.

\paragraph{Low Out-Degree Orientations} 
Many previous works 
give dynamic algorithms for low out-degree orientations with respect to
bounds on the \emph{maximum} arboricity that ever exists in the graph, 
$\alpha_{max}$~\cite{brodal99dynamicrepresentations,HeTZ14,BB20,Kowalik07,KKPS14,KS18,SW20}.
Noticeably, these sequential, dynamic works save a $O(\log n)$ 
factor in the running time compared to sequential dynamic algorithms that 
compute the orientation with respect to the \emph{current} 
arboricity~\cite{BHNT15,HNW20}. In practice, the arboricity 
of real-world graphs may vary as batches of updates are applied,
and in particular, the \kc numbers of each vertex can change drastically 
(e.g., many follows and unfollows can occur in a very short period of 
time following a viral post). Our work matches the 
update time of~\cite{BHNT15,HNW20} for maintaining a low out-degree 
orientation for the current $\alpha$. This explains why our
work bounds for maximal matching requires an additional
$O(\log n)$ factor compared to previous works \cite{HeTZ14,neiman2015simple}
that were in terms of $\alpha_{max}$.

\paragraph{Other Graph Problems}\label{sec:related-other-work} Using low out-degree orientations,
a number of works 
in the past have studied the other 
dynamic graph problems we study in this paper, 
including  maximal matching, vertex coloring, and clique 
counting~\cite{BS15,BS16,DvorakT13,DLSY21,FMN03,Kowalik10,neiman2015simple,HNW20,HeTZ14,SW20,BCKLRRV19,ChNi85,PPS16,HHH21,LT21}. 
The best update time for these problems in the sequential settings are summarized in~\cref{table:sequential}.

\begin{figure}[t!]
    \centering
    \caption{Previous best-known sequential algorithm results.}\label{table:sequential}
    \footnotesize
    \begin{tabular}{*4c}
        \toprule
        & \multicolumn{3}{c}{Summary of Best-Known Sequential Results} \\
        \cmidrule(lr){2-4}
        Problem & Approx & Update Time & Adversary \\    
        \midrule
        \kc &  $(2+\eps)$ & $O(\log^2 n)$~\cite{HNW20,SCS20}, \cref{lem:core-num} & Adaptive\\
        Orientation &  $(4+\eps)$    &  $O(\log^2 n)$~\cite{HNW20} & Adaptive \\
        Matching &   Maximal    & $O(\alpha_{max} + \log n/\log\log n)$~\cite{neiman2015simple,HeTZ14} & Adaptive \\
        $k$-clique & Exact & $O(\alpha_{max}^{k^2}\log^{k^2} n)$~\cite{DvorakT13} & Adaptive   \\
        Coloring & $O(\alpha \log n)$ & $O(\log^2 n)$~\cite{HNW20} & Oblivious   \\
        Coloring & $O\left(2^{\alpha}\right)$ & $O(\log^3 n)$~\cite{HNW20} & Adaptive\\
        \bottomrule
    \end{tabular}
\end{figure}

In the sequential 
setting, the best-known algorithm for 
$k$-clique counting uses $O(\log^{k^2} n)$ update 
time in bounded \emph{expansion} graphs for any $k$-vertex
subgraph~\cite{DvorakT13}. %
Bounded expansion is a more restricted class of 
graphs than bounded arboricity.\footnote{Graphs with bounded expansion 
have bounded arboricity, but not vice versa.}
Their algorithm crucially requires the 
\emph{fraternal augmentation} graph, $G'$,
which is created from an input directed graph, $G = (V, E)$, 
by adding an edge $(u, v)$ (direction chosen arbitrarily)
if and only if $(w, u)$ and $(w, v)$ exist. Provided an 
out-degree orientation of size $\sigma$, their algorithm runs in
$O(\sigma^{k^2}\log^{k^2}n)$ time; so for bounded arboricity graphs,
their algorithm can find any subgraph of size $k$ with 
$O(\alpha^{k^2}\log^{k^2}n)$ update time~\cite{DvorakT13}.
Our algorithm also gives a better update time in the sequential setting
than~\cite{DvorakT13} for counting cliques (for $\batchsize = 1$).
\section{Batch-Dynamic $k$-Core Decomposition}\label{sec:bounded-depth-arboricity}

In this section, we describe our parallel, batch-dynamic algorithm for
maintaining an $(2+\eps)$-approximate $k$-core decomposition (for any
constant $\eps > 0$) and prove its theoretical efficiency.

\subsection{Algorithm Overview}\label{sec:algoverview}

We present a \emph{parallel level data structure (PLDS)} that maintains a
$(2+\eps)$-approximate \kc decomposition
that is inspired by the class of sequential level data structures (LDS)
of~\cite{BHNT15,HNW20}.
Our algorithm achieves $O(\log^2 n)$
amortized work per update and  $O(\log^2 n \log \log n)$ depth \whp{}
We also present a deterministic version of our algorithm
that achieves the same work bound with $O(\log^3 n)$  depth.
Our data structure can also handle batches of
vertex insertions/deletions.
Our data structure requires $O(\log^2 n)$
amortized work, which matches the $O(\log^2 n)$ amortized update time
of~\cite{BHNT15,HNW20}.
\ifspaa
We also present a \emph{deterministic} version of our algorithm
that achieves the same work bound with $O(\log^3 n)$ depth
in~\cref{app:data-structure-impl}.

In addition to edge updates, 
our data structure also handles batches of
vertex insertions/deletions, discussed in~\cref{sec:vertex-ins-del}.
\fi
As in~\cite{HNW20},
our data structure can handle \emph{changing arboricity} that is not 
known a priori. 
Such adaptivity is necessary to
successfully maintain accurate approximations of coreness values. 

The LDS and our PLDS consists of a partition of 
the vertices into $K = O(\log^2 n)$ \defn{levels}.\footnote{When $m = o(n)$,
we can also show that $O(\log^2 m)$ levels suffice.}
We provide a very high level overview of PLDS in this section.
The levels are partitioned into equal-sized \defn{groups} of consecutive levels.
Updates are partitioned into insertions and deletions.
Vertices move up and down levels depending on the type of edge update incident
to the vertex. Rules governing the induced degrees
of vertices to neighbors in different levels determine whether a vertex
moves. Using information about the level of a vertex, we obtain a
$(2+\eps)$-approximation on the coreness of the vertex.

After every edge update, vertices update their levels depending on whether 
they satisfy two invariants. One invariant upper bounds the induced degree of each vertex $v$ in the subgraph 
consisting of all vertices in the same or higher level. Vertices whose degree exceeds this 
bound move up one or more levels. We process the levels from smallest to largest level and move
all vertices from the same level in parallel. The second invariant lower bounds the induced degree of 
each vertex $v$ in the subgraph consisting of all vertices in the level below $v$, the level of 
$v$ and all levels higher than the level of $v$. Vertices that violate this invariant calculate
a \defn{desire-level} or the closest level they can move to that satisfies this invariant. Then, 
vertices with the same desire-level are moved in parallel to that level. Finally,
the coreness estimates of the vertices are computed based on the current level of each vertex. We obtain
the low out-degree orientation by orienting edges from lower to higher levels (breaking ties by vertex index).
\cref{fig:invariants} shows the invariants maintained by our algorithm; \cref{fig:insertion-exp,fig:deletion-exp}
show how our algorithm processes insertion and deletion updates. Together, they demonstrate an example run of our 
algorithm.

\subsection{Sequential Level Data Structure (LDS)}
The sequential level data structures (LDS)
of~\cite{BHNT15,HNW20} maintains a low out-degree orientation
under dynamic updates.
Within their LDS, a vertex
moves up or down levels one by one, where a vertex $v$
(incident to an edge update) first checks whether an invariant is
violated, and then may move up or down one level. Then, the vertex
checks the invariants and repeats. Such movements may cause other vertices to move up or down 
levels. %
The LDS combined with our~\cref{sec:coreness}
directly gives an $O(\log^2 n)$ update time 
sequential, dynamic algorithm that outputs $(2+\eps)$-approximate coreness values.

Unfortunately, such a procedure can be slow in practice.
Specifically, a vertex that moves one level could cause a cascade
of vertices to move one level.
Then, if the vertex moves again, the same cascade of movements may
occur. An example is shown in~\cref{fig:bad-example}.
\begin{figure}[t]
    \centering
    \includegraphics[width=0.5\textwidth]{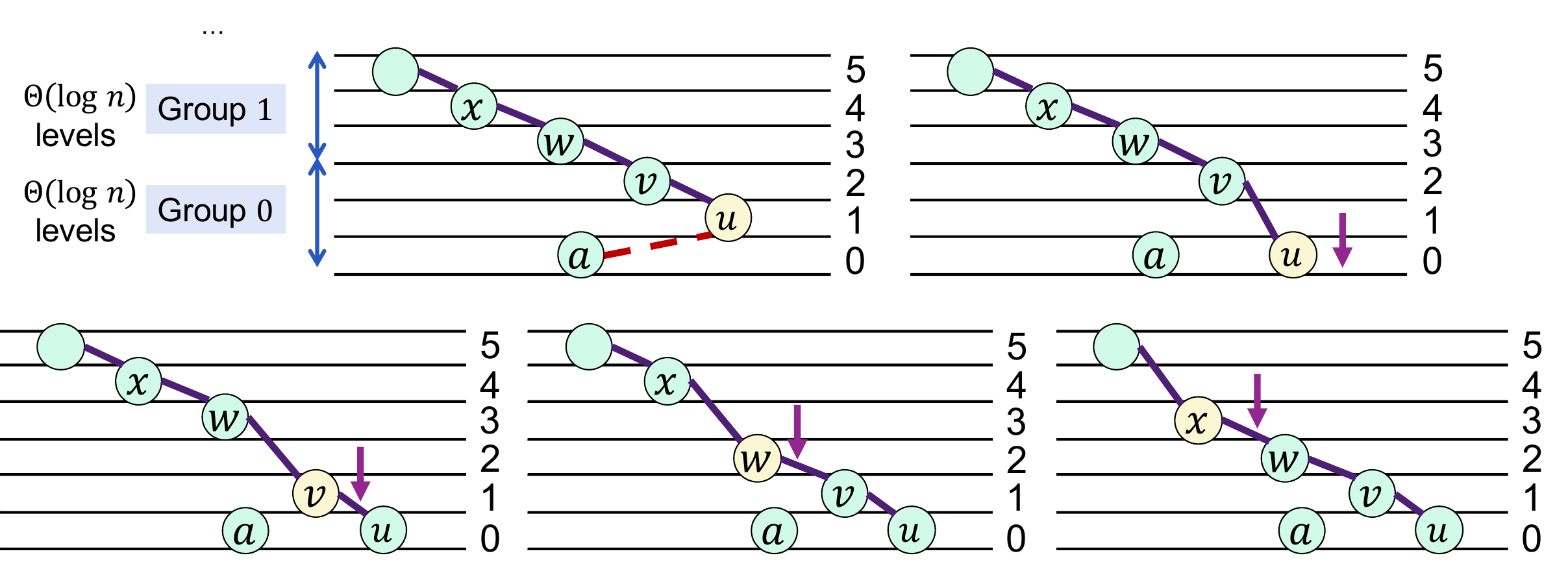}
    \caption{Example of a cascade of vertex movements
    caused by an edge deletion on $u$ (shown by the dashed red line). }\label{fig:bad-example}
\end{figure}
Furthermore, any trivial parallelization of the LDS to support a batch of updates
will run into race
conditions and other issues, requiring the use of locks which blows up the
runtime in practice.

Thus, our PLDS solves several challenges posed by the
sequential LDS. Given a batch $\batch$ of edge updates:
\textbf{(1)} our algorithm processes the levels in a careful order
that yields provably low depth for batches of updates;
\textbf{(2)} our insertion algorithm processes vertices on each level at
most once, which is key to the depth bounds---after vertices move up
from level $\ell$, no future step in the algorithm
moves a vertex up from level $\ell$; and
\textbf{(3)} our deletion algorithm moves vertices to their final
level in one step. In other words, a vertex moves at most once in a
deletion batch.

\subsection{Detailed PLDS Algorithm}\label{sec:detailed-plds}

As mentioned previously, the vertices of the input 
graph $G = (V, E)$ in our PLDS are
partitioned across $K$ \defn{levels}.
For each level $\ell = 0, \ldots, K - 1$, let $V_\ell$ be the set
of vertices that are currently assigned to level $\ell$.
Let $Z_{\ell}$ be the set of vertices in levels $\geq \ell$.
Provided a constant $\delta > 0$, the levels are partitioned
into \defn{groups} $\grp_0, \ldots, \grp_{\ceil{\log_{(1+\delta)} n}}$, 
where each group contains
$4\ceil{\log_{(1+\delta)} n}$ consecutive levels. Each $\ell \in
\left[i  \ceil{\log_{(1+\delta)}n}, \ldots, (i + 1)\ceil{\log_{(1+\delta)}n} - 
1\right]$ is a level in group $i$.
Our data structure consists of $K = O(\log^2 n)$ total levels.
The PLDS satisfies the following invariants as introduced 
in~\cite{BHNT15,HNW20}, which also govern how the
data structure is maintained. The invariants assume a given constant $\delta
> 0$ and a constant $\lambda > 0$.

\begin{invariant}[Degree Upper Bound]\label{inv:degree-1}
    If vertex $v \in V_{\ell}$, level $\ell < K$
    and $\ell \in \group_{i}$, then $v$ has at most
    $\coeff(1+\delta)^{i}$ neighbors in $Z_{\ell}$.
\end{invariant}

\begin{invariant}[Degree Lower Bound]\label{inv:degree-2}
    If vertex $v \in V_{\ell}$, level $\ell > 0$,
    and $\ell-1 \in \group_i$, then $v$ has at least
    $(1+\delta)^{i}$ neighbors in $Z_{\ell - 1}$.
\end{invariant}

\begin{figure}[t]
    \centering
    \includegraphics[width=0.4\textwidth]{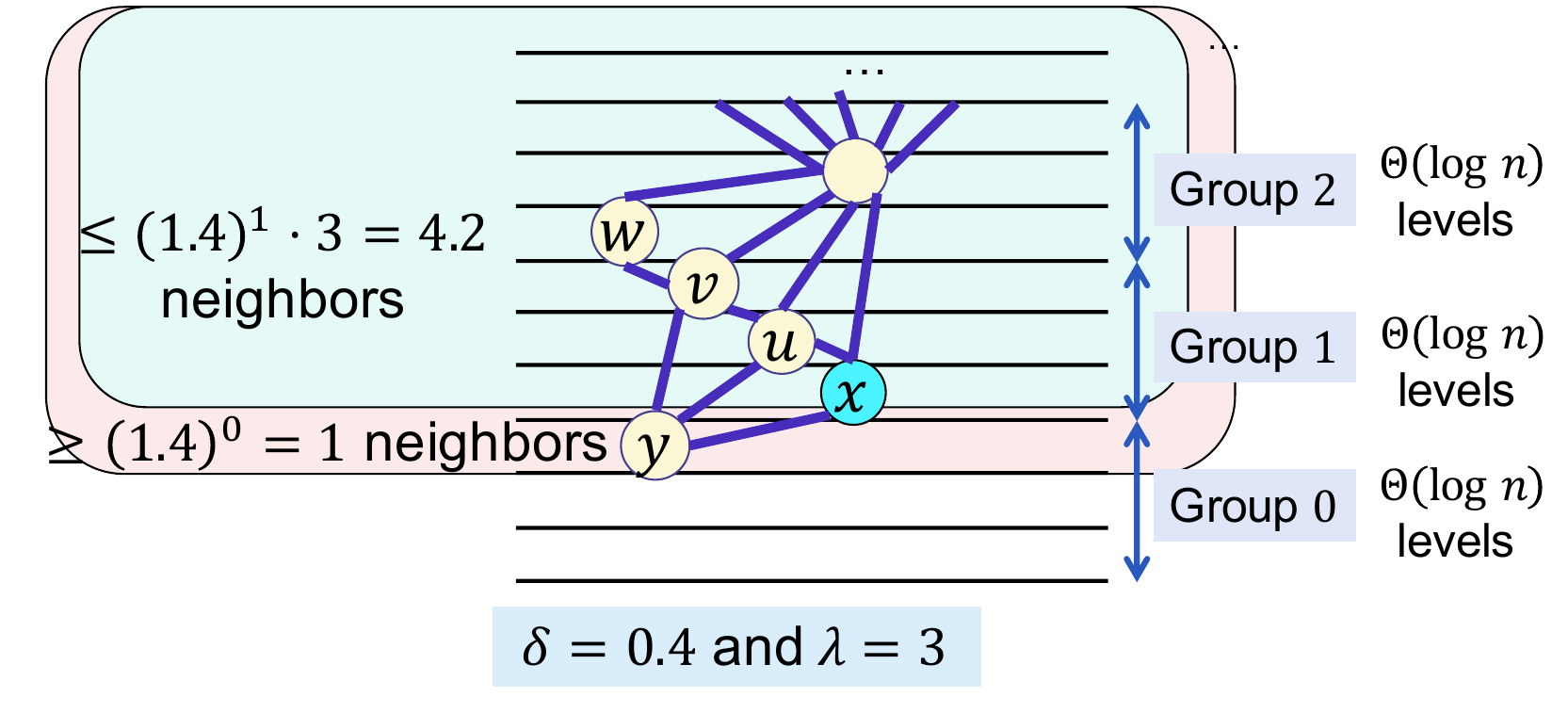}
    \caption{Example of invariants maintained by the PLDS for $\delta=0.4$ and $\lambda=3$.
    There are
    $\Theta(\log n)$ groups, each with $\Theta(\log n)$. Each vertex is in exactly one level of the structure and moves up and
    down by some movement rules.  For example, vertex $x$ (blue) is on level $3$ and in group $1$.}\label{fig:invariants}
\end{figure}

Vertices with no neighbors are placed in level $0$. An example
partitioning of vertices and maintained invariants is shown in~\cref{fig:invariants}.
Let $\level(v)$ be the level that $v$ is currently on.
We define the \defn{group number}, $\group(v)$, of a vertex $v$ to be the index
$i$ of the group
$\group_i$ where
$\level(v) \in \group_i$. Similarly, we define $\gn(\ell) = i$ to be
the group number for level $\ell$ where $\ell \in g_{i}$.
We define the \defn{up-degree}, $\up(v)$, of a vertex $v$ to be the number of
its neighbors in $Z_{\level(v)}$ (\defn{up-neighbors}),
and \defn{\upl-degree}, $\down(v)$, to be the
number of its neighbors in $Z_{\level(v)-1}$ (\defn{$\down$-neighbors}). 
These two notions of induced degree correspond to the requirements of
the two invariants of our data structure. We define neighbors $w$ of $v$
at levels $\level(w) < \level(v)$ to be the \defn{down-neighbors} of $v$.
Lastly, the \defn{desire-level} $\dl(v)$ of a vertex $v$ is the \emph{closest
level to the current level of the vertex}
that satisfies both~\cref{inv:degree-1} and~\cref{inv:degree-2}.

\begin{definition}[Desire-level]\label{def:desire-level}
    The \emph{desire-level}, $\dl(v)$, of vertex $v$ is the level $\level'$
    that minimizes $|\level(v) - \level'|$, and where
    $\down(v) \geq (1+\delta)^{i'}$ and $\up(v) \leq \coeff(1+\delta)^i$
    where $\level' - 1 \in \group_{i'}$,
    $\level' \in g_i$, and $i' \leq i$. In other words, the desire-level of $v$
    is the closest level $\ell'$ to the current level of $v$, $\level(v)$, where both~\cref{inv:degree-1}
    and~\cref{inv:degree-2} are satisfied.
\end{definition}

\begin{algorithm}[t!]\caption{Update($\batch$)}\label{alg:rebalance}
    \small
    \begin{algorithmic}[1]
        \Require{A batch of edge updates $\batch$.}
    \State Let $\batch_{\textit{ins}} =$ all edge insertions in $\batch$,
    and $\batch_{\textit{del}} =$ all edge deletions in $\batch$. \label{line:sep}
        \State Call RebalanceInsertions($\batch_{\textit{ins}}$).
        [\cref{alg:insert}]\label{line:insert}
        \State Call RebalanceDeletions($\batch_{\textit{del}}$).
        [\cref{alg:delete}]\label{line:delete}
    \end{algorithmic}
\end{algorithm}

We show that the invariants are always maintained except for a period of time when processing a new batch of insertions/deletions.
During this period,
the data structure undergoes a \emph{rebalance procedure}, 
where the invariants may be violated.
The main update procedure in~\cref{alg:rebalance}
separates the updates into insertions and deletions
(\cref{line:sep}), and then calls RebalanceInsertions
(\cref{line:insert}) and RebalanceDeletions
(\cref{line:delete}).
We make two \emph{crucial} observations:
when processing a batch of insertions,~\cref{inv:degree-2}
is never violated; and, similarly, when processing a batch of
deletions,~\cref{inv:degree-1} is never violated. Thus,
no vertex needs to move \emph{down} when processing an insertion batch
and no vertex needs to move \emph{up} when processing a deletion batch.
The two procedures are asymmetric, and so we first describe RebalanceInsertions (\cref{alg:insert}), 
and then describe RebalanceDeletions (\cref{alg:delete}).

\myparagraph{Data Structures}
Each vertex $v$ keeps track of its set of neighbors in two structures.
$\greater$ keeps track of the neighbors at $v$'s level and above.  We
denote this set of $v$'s neighbors by $\greater[v]$.  $\tbl_v$ keeps
track of neighbors of $v$ for every level below $\level(v)$---in
particular, $\tbl_v[j]$ contains the neighbors of $v$ at level $j <
\level(v)$.
\ifanon
\else
\iflong
    \cref{sec:data-structure-impl} and~\cref{app:data-structure-impl}.
\else
    the full paper.
\fi
\fi

\begin{algorithm}[!t]\caption{RebalanceInsertions($B_{\textit{ins}}$)}\label{alg:insert}
    \small
    \begin{algorithmic}[1]
    \Require{A batch of edge insertions $B_{\textit{ins}}$.}
    \State Let $\greater$ contain all up-neighbors of each vertex,
    keyed by vertex. So $U[v]$ contains all up-neighbors of $v$.
    \State Let $\tbl_{v}$ contain all neighbors of $v$ in levels $[0, \ldots, \level(v)-1]$,
    keyed by level number.
    \ParFor{each edge insertion $e = (u, v) \in \batch_{ins}$}
        \State Insert $e$ into the graph.
    \EndParFor
    \For{each level $l \in [0, \ldots, K-1]$ starting with $l = 0$} \label{line:outer-loop}
    \ParFor{each vertex $v$ incident to $B_{\textit{ins}}$ or is marked,
        where $\level(v) = l \cap \up(v) >
    (2+3/\lambda)(1+\delta)^{\gn(l)}$} \label{line:kc-check}
        \State Mark and move $v$ to level $l + 1$ and create $\tbl_v[l]$ to store
        $v$'s neighbors at level $l$.\label{line:move-higher}
        \EndParFor
        \ParFor{each  $w \in N(v)$ of a vertex $v$ that moved to level $l + 1$
        and $w$ stayed in level $l$}\label{line:insert-loop}
            \State $\greater[v] \leftarrow \greater[v] \setminus
            \left\{w\right\}, \tbl_v[l] \leftarrow \tbl_v[l] \cup
                    \{w\}$.\label{line:insert-modify}
        \EndParFor
        \ParFor{each  $u \in N(v)$ of a vertex $v$ that moved to level $l + 1$
        and $u$ is in level $l + 1$}             \label{line:mark-loop}
            \State Mark $u$ if $\up(u) > (2+3/\lambda)(1+\delta)^{\gn(l+1)}$.
            \label{line:mark-neighbor}
            \State $\greater[u] \leftarrow \greater[u] \cup \{v\}$, $\tbl_u[l] \leftarrow
            \tbl_u[l] \setminus \{v\}$.\label{line:neighbor-modify}
        \EndParFor
        \ParFor{each  $x \in N(v)$ of a vertex $v$ that moved to level $l + 1$
        and $x$ is in level $\ell(x) \geq l + 2$}\label{line:insert-ds-begin}
        \State $\tbl_x[l] \leftarrow \tbl_x[l] \setminus \{v\}, \tbl_x[l + 1]
        \leftarrow \tbl_x[l + 1] \cup \{v\}$.\label{line:insert-ds-end}
        \EndParFor
        \State Unmark $v$ if $\up(v) \leq (2+3/\lambda)(1+\delta)^{\gn(l+1)}$.
        Otherwise, leave $v$ marked.\label{line:unmark}
    \EndFor
    \end{algorithmic}
\end{algorithm}

\myparagraph{RebalanceInsertions($B_{\textit{ins}}$)}
\cref{alg:insert} shows the pseudocode.
Provided a batch of insertions $B_{\textit{ins}}$, we iterate through the $K$ levels from the
lowest level $\ell = 0$ to the highest level $\ell = K-1$
(\cref{line:outer-loop}). For each
level, in parallel we check the vertices incident to edge insertions in
$B_{\textit{ins}}$ or is marked to see if
they violate~\cref{inv:degree-1} (\cref{line:kc-check}). If a vertex $v$ in the
current level $l$
violates~\cref{inv:degree-1}, we move $v$ to level $l+1$
(\cref{line:move-higher}). After moving $v$, we update structures
$U[v], L_v$, and the structures of $w \in \neighbor(v)$ where
$\level(w) \in [l, l + 1]$.  First, we create $L_v[l]$ to store the
neighbors of $v$ in level $l$ (\cref{line:move-higher}).  If $v$ moved
to level $l + 1$ and $w$ stayed in level $l$, then we delete $w$ from
$U[v]$ and instead insert $w$ into $\tbl_v[l]$
(Lines~\ref{line:insert-loop}--\ref{line:insert-modify}). We do not need to make any data structure
modifications for $w$ since $v$ stays in $U[w]$. Similarly, no data
structure modifications to $v$ and $w$
are necessary when both $v$ and $w$ move to
level $l + 1$. For each neighbor of $v$ on level $l + 1$, we need to check
whether it now violates~\cref{inv:degree-1} (\cref{line:mark-loop}).
If it does, then we mark the vertex
(\cref{line:mark-neighbor}).
We process any such marked vertices when we process level $l + 1$. We also
update the $U$ and $L$ arrays of every neighbor of $v$ on level $l + 1$
(\cref{line:neighbor-modify}). Specifically, let $u$ be one such neighbor, we
add $v$ to $U[u]$ and remove $v$ from $L_u[l]$.
We conclude
by making appropriate
 modifications to $L$ for each neighbor on levels $\geq l + 2$
 (Lines~\ref{line:insert-ds-begin}--\ref{line:insert-ds-end}). Specifically,
 let $x$ be one such neighbor, we remove $v$ from $L_x[l]$ and add $v$ to
 $L_x[l+1]$.
All neighbors of vertices that moved can be checked
and processed in parallel.
Finally, $v$ becomes unmarked if it satisfies all invariants; otherwise, it
remains marked and must move again in a future step (\cref{line:unmark}).

\begin{figure*}[!t]
    \centering
    \includegraphics[width=0.9\textwidth]{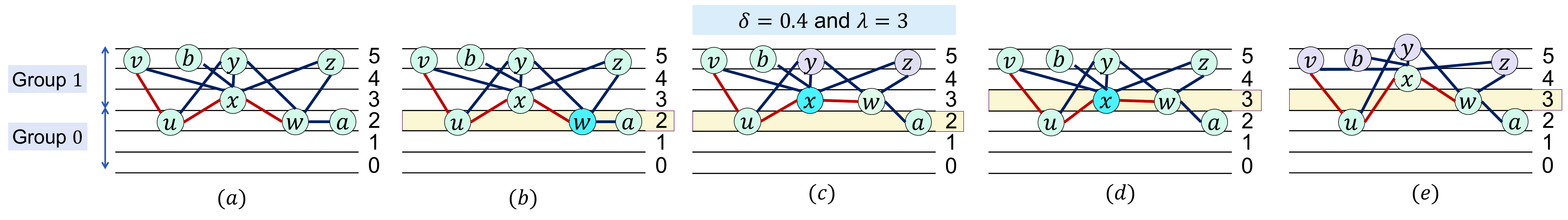}
    \caption{Example of RebalanceInsertions described in the text for $\delta=0.4$ and $\lambda=3$. The red lines represent the
    batch of edge insertions.}
    \label{fig:insertion-exp}
\end{figure*}

    \cref{fig:insertion-exp} shows an example of our entire insertion procedure
    described in~\cref{alg:insert} for $\delta=0.4$ and $\lambda=3$. The red lines in the example represent the
    batch of edge insertions. Thus, in $(a)$, the newly inserted edges are the
    edges $(u, v)$, $(u, x)$, and $(x, w)$. We iterate from the bottommost level
    (level $0$) to the topmost level (level $K - 1$).

    The first
    level where we encounter vertices that are marked or are adjacent to an edge
    insertion is level $2$. Since level $2$ is part of group $0$, the cutoff
    for~\cref{inv:degree-1} is $(2 + 3/\lambda)(1+\delta)^0 = 3$ provided $\lambda
    = 3$ and $\delta = 0.4$.
    In level $2$, only $w$ violates~\cref{inv:degree-1}
    since the number of its neighbors on levels $\geq 2$ is $4$ ($x$, $y$, $z$, and $a$),
    so $\up(w) = 4 > 3$ (shown in $(b)$).
    Then, in $(c)$, we move $w$ up to level $3$. We need to update the data
    structures for neighbors of $w$ at level $3$ and above (as well as $w$'s own
    data structures); the vertices with data structure updates are $x$, $w$, $y$, and
    $z$. After the move, $x$ becomes marked because it now
    violates~\cref{inv:degree-1} (the cutoff for level $3$ is $(2+3/3) (1+0.4) =
    4.2$ since level $3$ is in group $1$); $w$ becomes unmarked because it no longer
    violates~\cref{inv:degree-1}. %
    In $(d)$, we move on to process level $3$. The only vertex that is marked or
    violates~\cref{inv:degree-1} is $x$. Therefore, we move $x$ up one level (shown in
    $(e)$) and update relevant data structures (of $x$, $v$, $y$, $z$, and $b$).

\begin{algorithm}[t!]\caption{RebalanceDeletions($\batch_{\textit{del}}$)}\label{alg:delete}
    \small
    \begin{algorithmic}[1]
        \Require{A batch of edge deletions $\batch_{\textit{del}}$.}
        \State Let $\greater$ contain all up-neighbors of each vertex,
        keyed by vertex. So $U[v]$ contains all up-neighbors of $v$.
        Let $\tbl_{v}$ contain all neighbors of $v$ in levels $[0, \ldots,\level(v)-1]$,
        keyed by level number.
        \ParFor{each edge deletion $e = (u, v) \in \batch_{\textit{del}}$}
            \State Remove $e$ from the graph.
        \EndParFor
    \ParFor{each vertex $v$ where $\down(v) < (1+\delta)^{\gn(\level(v) -
    1)}$}\label{inv:2-violate}
    \State Calculate $\dl(v)$ using CalculateDesireLevel($v$).\label{line:calculate-dl}
    \EndParFor
    \For{each level $l \in [0, \ldots,K-1]$ starting with level $l =
    0$}\label{line:del-outer-loop}
        \ParFor{each vertex $v$ where $\dl(v) = l$}\label{line:move-to-dl-start}
            \State Move $v$ to level $l$.\label{line:move-to-dl-end}
        \EndParFor
        \ParFor{each vertex $v$ where $\dl(v) = l$}\label{line:second-loop}
        \ParFor{each neighbor $w$ of $v$
        where $\level(w) \geq l$}\label{line:violates-deg-1}
            \State Let $p_v$ and $p_w$ be the previous levels of $v$ and $w$,
                respectively, before the move.
                \If{$\level(w) = l$}\label{line:u1}
                \State $\tbl_w[p_v] \leftarrow \tbl_w[p_v] \setminus \{v\}$,
                    $\tbl_v[p_w] \leftarrow \tbl_v[p_w] \setminus \{w\}$.
                \State $\greater[w] \leftarrow \greater[w] \cup \{v\}, \greater[v] \leftarrow \greater[v] \cup \{w\}$.
            \Else
                \If{$p_v > \level(w)$} 
                    \State $\greater[w] \leftarrow \greater[w] \setminus \{v\},
                    \tbl_v[\level(w)] \leftarrow \tbl_v[\level(w)]
                    \setminus \{w\}$. %
                \ElsIf{$p_v = \level(w)$}
                \State $\greater[w] \leftarrow \greater[w] \setminus
                \{v\}$.
                \Else{
                $\tbl_w[p_v] \leftarrow \tbl_w[p_v] \setminus \{v\}$.}
                \EndIf
                \State $\tbl_w[l] \leftarrow \tbl_w[l] \cup \{v\},
                \greater[v] \leftarrow \greater[v] \cup
                \{w\}$.\label{line:u-end}
                \EndIf
                \If{$\down(w) < (1+\delta)^{\gn(\level(w) - 1)}$}\label{line:update-dl-check}
                \State Recalculate $\dl(w)$
                using~\cref{alg:desire-level}.\label{line:update-dl}
                \EndIf
            \EndParFor
            \EndParFor
        \EndFor
        \end{algorithmic}
    \end{algorithm}

\myparagraph{RebalanceDeletions($B_{\textit{del}}$)}
Unlike in LDS, deletions in PLDS are handled by moving each vertex 
at most once, directly to its final level
(the vertex \emph{does not move} again during this procedure).
We show in the analysis that this guarantee is \emph{crucial to
obtaining low depth}. The pseudocode is shown in \cref{alg:delete}.
For each vertex $v$ incident to an edge deletion, we 
check whether it violates
\cref{inv:degree-2} (\cref{inv:2-violate}). On~\cref{inv:2-violate},
$\gn(\ell(v) - 1)$ gives the
group number $i$ where $\ell(v) - 1 \in g_i$.
If $v$ violates~\cref{inv:degree-2},
we calculate its desire-level, $\dl(v)$,
using CalculateDesireLevel
(\cref{line:calculate-dl}), described next.
We iterate through the levels from $l = 0$ to $l = K - 1$
(\cref{line:del-outer-loop}).  Then, in parallel for each vertex $v$
whose desire-level is $l$, we move $v$ to level $l$
(Lines~\ref{line:move-to-dl-start}--\ref{line:move-to-dl-end}).  
We update the data structures of each $v$ that moved and $w \in
\neighbor(v)$ where $\level(w) \geq l$
(Lines~\ref{line:second-loop}--\ref{line:u-end}). Specifically, 
we need to update $U[v], U[w], L_v,$ and $L_w$ if $v$ was originally
an up-neighbor of $w$ and becomes a down-neighbor or vice versa.
Finally, we update
the desire-level of neighbors of $v$ that no longer satisfy
\cref{inv:degree-2} (Lines~\ref{line:update-dl-check}--\ref{line:update-dl}). We process all vertices
that move and their neighbors in parallel.

\cref{fig:deletion-exp} shows an example of~\cref{alg:delete} for $\delta=1$ and $\lambda=3$.
In $(a)$, the newly deleted edges are $(x, z)$ and $(y, w)$. For each vertex adjacent to an edge deletion,
we calculate its desire-level, or the closest level to its current
level that satisfies~\cref{inv:degree-2}.
In $(b)$,
    only $x$ and $z$ violate~\cref{inv:degree-2}. The lower bound on
    the number of neighbors that must be
    at or above level $3$ for $x$ and level $4$ for
    $z$ is
    $(1+\delta)^1 = 2$ since $\delta = 1$ and levels $3$ and $4$ are in
    group $1$. (Recall that the lower bound is calculated with respect to the level
    \emph{below} $x$ and $z$.)
    We calculate that the desire-levels of $x$ and $z$ are both $3$.
    The desire-levels of $y$ and $w$ are their current levels
    because they do not violate the invariant.
    Then, we iterate from the bottommost level
    (starting with level $0$) to the topmost level (level $K - 1$).
    Level $3$ is the first level where vertices want to move. Then, we move $x$
    and $z$ to level $3$ (shown in $(c)$). We only need to update the data
    structures of neighbors at or above $x$ and $z$ so we only update
    the structures of $x$, $y$, and $z$. \cref{inv:degree-2} is no longer
    violated for $x$ and $z$. In fact, our algorithm guarantees that each
    vertex \emph{moves at most once}. We check whether any of $x$ or $z$'s
    up-neighbors violate \cref{inv:degree-2}. Indeed, 
    $y$ now violates the invariant. In $(d)$, we recompute the desire-level of
    $y$ and its desire-level is now $4$. Then, we move $y$ to level $4$ in $(e)$.

\begin{figure*}[htpb]
    \centering
    \includegraphics[width=0.9\textwidth]{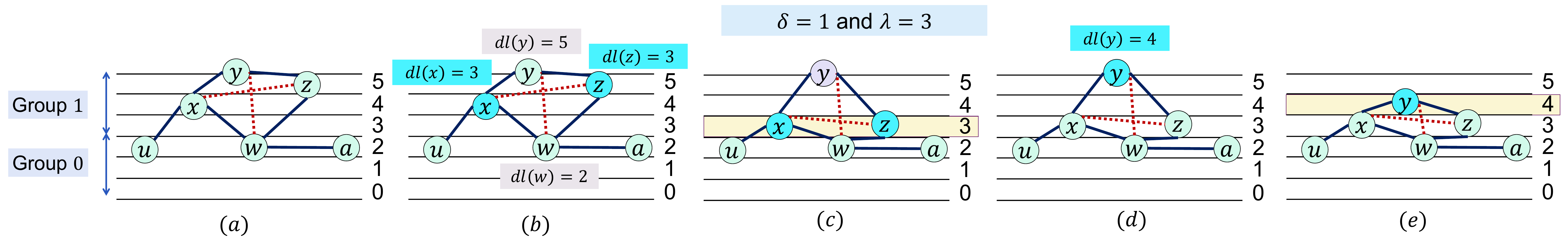}
    \caption{Example of RebalanceDeletions described in the text for $\delta=1$ and $\lambda=3$. The red dotted lines represent the
    batch of edge deletions. 
    }
    \label{fig:deletion-exp}
\end{figure*}

\begin{algorithm}[!t]\caption{CalculateDesireLevel($v$)}\label{alg:desire-level}
    \small
    \begin{algorithmic}[1]
        \Require{A vertex $v$ that needs to move to a level $j <
            \level(v)$.}
        \Ensure{The desire-level $\dl(v)$ of vertex $v$.}
        \State  $d \leftarrow \down(v), \prevind  \leftarrow 1, \ind \leftarrow 2$
        \While{$d < (1+\delta)^{\gn(\level(v) - p)}$ and $\level(v) - p > 0$}
            \State $d \leftarrow d + \sum_{j=\prevind}^{\ind-1}\big|L_v[\level(v)-j-1]\big|$ \label{line:sum-level}
            \If{$d \geq (1+\delta)^{\gn(\ell(v) - i)}$}\label{line:kc-check-2}
                \State Binary search within levels $[\level(v) - \ind + 1, \ldots,\level(v) - \prevind]$ to find the
                    closest level to $\level(v)$ that
                    satisfies Invariants~\ref{inv:degree-1}
                    and~\ref{inv:degree-2}; \Return this level.\label{line:binary}
            \EndIf
            \State  $\prevind \leftarrow \ind, \ind \leftarrow \min(2\cdot
            \ind,\level(v))$.\label{line:double-count}
\EndWhile
    \State \Return 0.
    \end{algorithmic}
\end{algorithm}

\myparagraph{CalculateDesireLevel($v$)}
\cref{alg:desire-level} shows the procedure for calculating the desire-level, $\dl(v)$, of vertex $v$, which is used
in~\cref{alg:delete}.
Let $\gn(\level)$ be the index $i$ where level $\ell \in \group_i$.
We use a doubling procedure followed by a binary search to calculate
the desire-level.  We initialize a variable $d$ to $\down(v)$ (number of
neighbors at or above level $\level(v) - 1$).
Starting with level $\level(v) - 2$, we add the number of neighbors in
level $\level(v) - 2$ to $d$
(\cref{alg:desire-level},
\cref{line:sum-level}). This procedure
checks whether moving $v$ to $\level(v) - 1$
satisfies~\cref{inv:degree-2} (\cref{line:kc-check-2}). If it passes the
check, then we are done and we move $v$ to $\level(v) - 1$. Otherwise,
we iteratively double the number of levels from which we count neighbors until we find a level where~\cref{inv:degree-2} is
satisfied
(\cref{line:double-count}). 
On each iteration, we sum
the number of neighbors (\cref{line:sum-level}) 
in the range of levels using a parallel reduce.
We continue until we find a level where~\cref{inv:degree-2} is
satisfied. Let this level be $\ell'$ and the previous cutoff be
$\ell_{\textit{prev}}$. Finally, we perform a binary search within the range
$[\ell',\ldots, \ell_{\textit{prev}}]$ to find the \emph{closest} level to $\level(v)$
that satisfies~\cref{inv:degree-2} (\cref{line:binary}).

\ifanon
\else
\iflong
\subsection{Vertex Insertions and Deletions}\label{sec:vertex-updates}
We can handle vertex insertions and deletions by inserting vertices which have
zero degree and considering deletions of vertices to be a batch of edge
deletions of all edges adjacent to the deleted vertex. When we insert a vertex
with zero degree, it automatically gets added to level $0$ and remains in level
$0$ until edges incident to the vertex are inserted. For a vertex deletion, we
add all edges incident to the deleted vertex to a batch of edge deletions. Note,
first, that all vertices which have $0$ degree will remain in level $0$. Thus,
there are at most $O(m)$ vertices which have non-zero degree.

Because we have $O(\log^2 m)$ levels in our data structure, we rebuild the data
structure once we have made $\frac{m}{2}$ edge updates (including edge updates
from edges incident to deleted vertices). Rebuilding the data structure requires
$O(m\log^2 n)$ total work which we can amortize to the $\frac{m}{2}$ edge updates to
$O(\log^2 n)$ amortized work. Running~\cref{alg:insert} and~\cref{alg:delete} on
the entire set of $O(m)$ edges requires $\tilde{O}(\log^2 n)$ depth.
\fi
\subsection{Coreness and Low-Outdegree Orientation}\label{sec:compute-core}
\fi
\subsection{Efficiency Analysis}\label{sec:efficiency}

We now analyze the work and depth of our \plds.
\ificml
We defer proofs
to the full paper due to space constraints.
\fi
First, it is easy to show that there exists a level where both invariants are
satisfied. This allows our PLDS to assign each vertex to
a single level.

\begin{lemma}\label{lem:good-level}
    If a vertex $v$ violates~\cref{inv:degree-1}, then there
    exists a level $l >
    \level(v)$ where $v$ satisfies both~\cref{inv:degree-1}
    and~\cref{inv:degree-2}. If a vertex $w$ violates~\cref{inv:degree-2}, then
    there exists a level $l < \level(w)$ where $w$ satisfies both invariants or
    $l = 0$ (it is in the bottommost level).
\end{lemma}

\begin{proof}
    First note that no vertex can simultaneously violate
    both~\cref{inv:degree-1} and~\cref{inv:degree-2}. Thus, suppose first that
    $v$ violates~\cref{inv:degree-1}. Then, this means that the number of
    neighbors
    of $v$ on levels $\geq \level(v)$ is more than
    $\coeff(1+\delta)^{g(v)}$ where $g(v)$ is the group number of $v$.
    If $v$ still violates~\cref{inv:degree-1} on level $\level(v) + 1$,
    then we keep moving $v$ to the next level.

    Otherwise, $v$ does not violate~\cref{inv:degree-1} on level $\level(v) +1$.
    Since we know that $v$ violated~\cref{inv:degree-1} on level $\level(v)$,
    then after we move $v$ to $\level(v) + 1$, $v$'s \ld is greater than
    $(1+\delta)^{\gn(\level(v))}$;
    hence, $v$ also does not violate~\cref{inv:degree-2}. The very last
    level of the $K$ levels
    has up-degree bound $\coeff(1+\delta)^{\ceil{\log_{1+\delta}(n)}} > 2n$
    when a vertex can be adjacent to at most $n -1$ vertices. 
    Hence, there must exist a level at or below the last level where both invariants
    are satisfied.
    A similar argument holds for~\cref{inv:degree-2}.
\end{proof}

Then, we make the following two observations that a batch of
insertions never violates~\cref{inv:degree-2} and a batch of
deletions never violates~\cref{inv:degree-1}.
This is true because deletions can never increase the \hd of any vertex and
insertions can never decrease the \ld of any vertex. 
\ificml
Using these two facts,
we can immediately obtain our depth bounds for our insertion and
deletion procedures.
\fi

\iflong
\begin{observation}[Batch Insertions]\label{lem:batch-insertions-inv}
    Given a batch of insertions, $\batch_{ins}$,~\cref{inv:degree-2} is never violated
    while $\batch_{ins}$ is applied.
\end{observation}

\begin{proof}
    The first part of the algorithm inserts the edges into the data structure.
    Since no edges are removed from the data structures, the degrees of all the
    vertices after the insertion of edges cannot decrease.
    \cref{inv:degree-2} was satisfied
    before the insertion of the edges, and hence, it remains satisfied after the
    insertion of edges because no vertices lose neighbors.
    We prove that the lemma holds for the remaining part of~\cref{alg:insert}
    via induction on the level $i$ processed by the procedure.
    In the base case, when $i = 0$, all vertices $v$ in the level which
    violate~\cref{inv:degree-1} are moved up to a level $\dl(v) > 0$.
    By definition of desire-level, $v$ is moved to a level
    where~\cref{inv:degree-2} is still satisfied, by~\cref{lem:good-level}.
    Vertices from level $0$ which move
    to levels $k \geq 1$
    cannot decrease
    the \ld for neighbors in all levels $j$ where
    $j > 1$.
    Thus,~\cref{inv:degree-2} cannot be violated for these
    vertices. Vertices not adjacent to $v$ are not affected by the move.

    We assume that~\cref{inv:degree-2} was
    not violated up to level $i$ and prove it is not violated while processing
    vertices on level $i + 1$. By our induction hypothesis, no vertices
    violate~\cref{inv:degree-2} before we process level $i + 1$. Then, when we
    process level $i + 1$, no vertices move down to a lower level than $i + 1$
    by construction of our algorithm
    because~\cref{inv:degree-2} is not violated for any vertex on level $i + 1$
    and if~\cref{inv:degree-1} is violated for any vertex $w$, $w$ must move up
    to a higher level. Any vertex $w$ which moves up to a higher level
    cannot decrease
    the \ld of neighbors of $w$. Hence, no vertex on levels $\geq i + 1$
    can violate~\cref{inv:degree-2}. The \ld of vertices on levels $< i +1$ are
    not affected by the move. Hence, no vertices on levels $< i + 1$
    violate~\cref{inv:degree-2}. Finally, if a vertex on level $i + 1$
    violates~\cref{inv:degree-1}, it will move to a level $j > i + 1$ where both
    invariants are satisfied by~\cref{lem:good-level}.
\end{proof}

\begin{observation}[Batch Deletions]\label{lem:batch-deletes}
    Given a batch of deletions, $\batch_{del}$,~\cref{inv:degree-1} is never
    violated while $\batch_{del}$ is applied.
\end{observation}

\begin{proof}
    \cref{alg:delete} first applies all the edge deletions in the batch. Edge
    deletions cannot make the \hd of any vertex greater; thus, no vertex
    violates~\cref{inv:degree-1} after applying the edge deletions. We prove
    that the rest of the algorithm does not violate~\cref{inv:degree-1} via
    induction over the levels $i$. Specifically, we use as our induction
    hypothesis that after processing the $i$'th level, no vertices
    violate~\cref{inv:degree-1}. In the base case, when $i = 0$, no vertices
    violate~\cref{inv:degree-1} at the beginning, and vertices from levels $i >
    0$ move to level $0$. This means that during the processing of level $i =
    0$, vertices only move to level $0$ from a higher level. Thus, all such
    vertices that move will move to a lower level. Since vertices which move to
    lower levels do not increase the \hd of any other
    vertices, no vertex can violate~\cref{inv:degree-1} at the end of processing
    level $0$. We now prove the case for processing level $i + 1$. In this case,
    we assume by our induction hypothesis that no vertices
    violate~\cref{inv:degree-1} after we finish  processing level $i$. Thus,
    all vertices that want to move to level $i + 1$ and violate~\cref{inv:degree-2}
    are at levels $j > i + 1$.
    Such vertices move down
    and thus cannot increase the
    up-degree of any vertex. This means that
    after moving all vertices that want to move to level $i + 1$, no vertices
    violate~\cref{inv:degree-1}.
\end{proof}

  \myparagraph{Batch Insertion Depth Bound}
    Using our observations,
    the depth of our batch insertion algorithm (\cref{alg:insert}) depends on the
    following lemma which states that once we have \emph{processed} a level
    (after finishing the corresponding iteration of~\cref{line:outer-loop}),
    no vertex will want to move from any level lower than that level.
    This means that each level is processed exactly once, resulting in at most
    $O(\log^2 n)$ levels to be processed sequentially.

    \begin{lemma}\label{lem:move-up}
        After processing level $i$ in~\cref{alg:insert},
        no vertex $v$ in levels $\level(v) \leq i$ will
        violate~\cref{inv:degree-1}.
        Furthermore, no vertex $w$ on levels $\level(w) > i$ will have
        $\dl(w) \leq i$.
    \end{lemma}
    
    \begin{proof}
    We prove this via induction. For the base case $i = 0$, all vertices on
    level $0$ are part of each other's \hd; then, no vertices which move up from
    $i = 0$ can cause the \hd of any vertices remaining in level $0$ to
    increase. We now assume the induction hypothesis for $i - 1$
    and prove the case for $i$.
    Vertices on level $j \leq i$ already contain vertices on levels $\geq i$ in
    its \hd. Such vertices on levels $\geq i$
    when moved to a higher level are still part of the
    \hd of vertices on levels $j \leq i$.
    Hence, no vertices on levels $j \leq i$ will
    violate~\cref{inv:degree-1} due to vertices in levels $\geq i$
    moving up to a level $l > i$. Then, in order for a vertex $w$ with
    $\level(w) > i$ to have $\dl(w) \leq i$, some neighbors of $w$ must have moved
    to a level $\leq i$. By~\cref{lem:batch-insertions-inv}, no vertices move
    down during~\cref{alg:insert}, so this is not possible.
\end{proof}

\ifanon
\else
\begin{proof}
    \cref{alg:delete} first applies all the edge deletions in the batch. Edge
    deletions cannot make the \hd of any vertex greater; thus, no vertex
    violates~\cref{inv:degree-1} after applying the edge deletions. We prove
    that the rest of the algorithm does not violate~\cref{inv:degree-1} via
    induction over the levels $i$. Specifically, we use as our induction
    hypothesis that after processing the $i$'th level, no vertices
    violate~\cref{inv:degree-1}. In the base case, when $i = 0$, no vertices
    violate~\cref{inv:degree-1} at the beginning, and vertices from levels $i >
    0$ move to level $0$. This means that during the processing of level $i =
    0$, vertices only move to level $0$ from a higher level. Thus, all such
    vertices that move will move to a lower level. Since vertices which move to
    lower levels do not increase the \hd of any other
    vertices, no vertex can violate~\cref{inv:degree-1} at the end of processing
    level $0$. We now prove the case for processing level $i + 1$. In this case,
    we assume by our induction hypothesis that no vertices
    violate~\cref{inv:degree-1} after we finish  processing level $i$. Thus,
    all vertices that want to move to level $i + 1$ and violate~\cref{inv:degree-2}
    are at levels $j > i + 1$.
    Such vertices move down
    and thus cannot increase the
    up-degree of any vertex. This means that
    after moving all vertices that want to move to level $i + 1$, no vertices
    violate~\cref{inv:degree-1}.
\end{proof}
\fi

\myparagraph{Batch Deletion Depth Bound}
For the batch deletion algorithm (\cref{alg:delete}), we prove that, starting from the lowest level,
after all vertices with $\dl(w) = i$ are moved to the $i$'th level,
no vertex $v$ will have $\dl(v) \leq i$.
This means that each level is processed exactly once, resulting in at most
$O(\log^2 n)$ levels to be processed sequentially.

\begin{lemma}\label{lem:smaller-level}
    After processing all vertices that move to level $i$ in~\cref{alg:delete},
    no vertex $v$ needs to be moved to any level $j \leq i$ in a future iteration
    of~\cref{line:del-outer-loop}; i.e., no vertex $v$ has $\dl(v) \leq
    i$ after processing $i$.
\end{lemma}
    
\begin{proof}
     We prove this via induction. In the base case when $i = 0$, all
     vertices with $\dl(v) = 0$ are moved to level $0$. All vertices which have
     $\dl(v) = 0$ are vertices which have degree $0$. Thus, all vertices that
     do not have $\dl(v) = 0$ have degree $\geq 1$ and have $\dl(w) \geq 1$.
     Hence, after moving all vertices with $\dl(v) = 0$ to level $0$, no
     additional vertices need to be moved to level $0$.  Assuming our induction
     hypothesis, we now show our lemma holds for level $i + 1$.
     All vertices
     that move to level $i + 1$ violated~\cref{inv:degree-2} and hence have \ld $<
     (1+\delta)^{\gn(j - 1)}$ at level $j > i + 1$ and \ld $\geq
     (1+\delta)^{\gn(i)}$ at level $i + 1$.
     After moving all vertices with $\dl(v) = i+1$ to level $i + 1$, no
     vertices on levels $k \leq i + 1$ have their \ld decreased by the move.
     We conclude the proof with vertices at levels $l > i + 1$. Suppose
     for the sake of contradiction
     that there exists some vertex $w$ on level $l
     > i + 1$ which has $\dl(w) \leq i + 1$ after the move.
     In order for $\dl(w) \leq i + 1$, some neighbor(s) of $w$ must move below level
     $i$, a contradiction. Finally, due to~\cref{lem:batch-deletes}, no vertices
     below level $i + 1$ will move up.
\end{proof}

\iflong

\ifanon
\else
Before we present the final lemmas for the depth of our algorithm, we discuss
the depth incurred from our data structures and also
briefly the deterministic and space-efficient settings.

\paragraph{Deterministic Setting}

In the deterministic setting, we maintain the list of neighbors using dynamic
arrays, which also means that we maintain and access the sizes of these arrays
in $O(1)$ work and depth. Because we are using dynamic arrays, we need to
occasionally resize the arrays in $O(1)$ amortized work and $O(1)$
depth.
Finally, we can modify the arrays in $O(1)$ work and
depth (not counting the depth for resizing).

\paragraph{Randomized Setting}
\fi
We describe the depth of our parallel data structures next.
We provide a set of linear-space data structures in~\cref{app:data-structure-impl} 
at the cost of increased depth. 

\fi
\ifanon
\else
\paragraph{Space-Efficient Setting}

In the space-efficient setting, we replace the structure used to represent
$\tbl_{v_i}$ with a linked list. Inserting and deleting vertices from the linked
list requires $O(1)$ work and depth (assuming we are given a pointer to the
vertex). Then, resizing the dynamic arrays (pointed to by the linked lists to
maintain the set of elements in each level)
require $O(1)$ amortized work and $O(1)$ depth.
\fi

\begin{lemma}\label{cor:random-depth}
    \cref{alg:rebalance} returns a randomized parallel level data structure that
    maintains~\cref{inv:degree-1} and~\cref{inv:degree-2} and has $O(\log^2
    n \log\log n)$ depth, \whp{}, and $O(n\log^2 n + m)$ space.
\end{lemma}

\begin{proof}
By~\cref{lem:move-up} and~\cref{lem:smaller-level},~\cref{alg:insert} (\cref{line:outer-loop})
and~\cref{alg:delete} (\cref{line:del-outer-loop})
iterates through $O(\log^2 n)$ levels sequentially.
Thus, the depth of algorithms is determined by the depth 
of the procedures that are run in each level the algorithm iterates through.

    We maintain the list of neighbors using separate parallel hash tables for each 
vertex $v$.  One hash table contains $v$'s neighbors at the same or higher levels.
Vertex $v$'s neighbors in levels below $\ell(v)$ are placed in a separate hash table
for each level. Parallel lookups into the hash tables require $O(1)$  
depth \whp, and inserting and deleting elements within the tables require
$O(\log^*n)$ depth \whp
Simultaneously changing the values within
    the hash table require $O(\log^* n)$ depth \whp Then, the depth per level of the
    structure is dominated by~\cref{alg:desire-level}. 
    
    The only additional depth we need to consider
is the depth incurred from~\cref{alg:desire-level}. Both the doubling search and
the binary search require $O(\log K)=O(\log \log n)$ depth.
All other contributions come from concurrently modifying and accessing dynamic arrays
and hash tables and can be done in $O(\log^*n)$ depth \whp
    
    Using the above, we successfully prove that the depth of~\cref{alg:rebalance} is $O(\log^2n\log\log n)$ \whp{} %
The extra space in addition to storing the graph is $O(n\log^2n)$ because we must have $O(\log^2 n)$ size dynamic arrays for each vertex to 
track their neighbors at lower levels (i.e., the neighbors in
$L_v$). Thus, the total depth of
    our randomized algorithm is $O(\log^2 n \log\log n)$ \whp{}, and the space used is $O(n\log^2 n + m)$.
\end{proof}

\ifanon
\else
\begin{lemma}\label{cor:random-depth}
    \cref{alg:rebalance} returns a randomized level data structure that
    maintains~\cref{inv:degree-1} and~\cref{inv:degree-2} and has $O(\log^2
    m\log\log m)$ depth, \whp{} and the space used is $O(n\log^2 n + m)$.
\end{lemma}

We show that with slightly
more complicated data structures involving linked lists,
we can obtain \emph{space-efficient}
\ificml
    structures in the full version of our paper~\cref{full}.
\fi
\iflong
    structures.
\fi
\iflong
\begin{corollary}\label{cor:space-efficient}
    \cref{alg:rebalance} returns a space-efficient, deterministic, level data
    structure that maintains~\cref{inv:degree-1} and~\cref{inv:degree-2} and has
    $O(\log^4 m)$ worst-case depth.
\end{corollary}

\begin{proof}
    The proof is the same as the proof of~\cref{lem:depth} except we
    replace~\cref{alg:desire-level} with an $O(\log^2 n)$ linear in number of
    levels search. The specific data structure we use for each vertex $v$ is a linked list
    with each vertex of the linked list representing a level $\leq \ell(v) - 1$
    which contains one or
    more neighbors of $v$. Then, each vertex in the linked list contains a pointer to
    a dynamic array containing the neighbors in that level.
    The linked list has size at most $O(\log^2 n)$.
    Thus, the total depth is $O(\log^4 m)$.
\end{proof}
\fi
\fi

\subsection{Potential Argument for Work Bound}\label{sec:work}

Our work bound uses the potential functions presented in Section 4
of~\cite{BHNT15}. We show that we can analyze our algorithm using these
potential functions and our parallel algorithm serializes to a set of sequential
steps that obey the potential function. We obtain the following lemma by the
potential argument provided in this section.

\begin{lemma}\label{cor:depth}
    For a batch of $\batchsize <m$ updates, \cref{alg:rebalance} returns a
    PLDS that maintains~\cref{inv:degree-1} and~\cref{inv:degree-2}
    in $O(\batchsize\log^2n)$ amortized work and $O(\log^2 n \log\log n)$
    depth \whp{}, using    $O(n\log^2 n + m)$ space.
\end{lemma}

\iflong
\subsubsection{Proof of Work Bound}

Unlike the algorithm presented in~\cite{HNW20,BHNT15}, in each round,
to handle deletions, we recompute the
$\dl(v)$ of any vertex $v$ that we want to move to a lower level.
Specifically, we compute and move $v$ to the closest level that satisfies
both~\cref{inv:degree-1} and~\cref{inv:degree-2}.
This is a different algorithm from the algorithm presented
in~\cite{HNW20,BHNT15}, and so we present for completeness a work argument for our
modified algorithm. The work bound we present accounts for the work of any one vertex's
movement up or down levels using the potential function argument
of~\cite{BHNT15}.
Note that this potential function also gives us
the amortized work per edge update of our algorithm since there exists a corresponding set of
sequential updates that cannot do less work than the set of parallel updates. Although the algorithm is different, the below
potential work bound argument follows closely the 
work bound proof 
presented in Bhattacharya \etal~\cite{BHNT15}. However, we repeat the proof again here (with some modifications) 
for completeness.

\paragraph{Charging the Cost of Moving Levels}
The strategy behind our potential function is to use the \emph{increase} in our
potential function due to edge updates to pay for the \emph{decrease} in
potential due to vertices moving up or down levels, which is enough to account for the work of moving the vertices.
We can then charge our costs
to the increase in potential due to edge updates. 
Below, we bound the increase
in potential due to edge updates and the decrease in potential due to vertex
movements.

We use the following potential function to calculate our potential.
First, recall some notation.
Let $Z_i$ be
the set of vertices in levels $i$ to $K$. In other words, $Z_i = \bigcup_{j =
i}^{K}
V_j$. Let $N(u, Z_i)$ be the set of neighbors of $u$ in the induced subgraph
given by $Z_i$. 
Let $\ell(u)$ be the current level that $u$ is on.
Finally,
let $\gn(\ell)$ be the group number of level $\ell$; in other words, $\ell \in
\group_{\gn(\ell)}$.
Let $f: [n] \times [n] \rightarrow \{0, 1\}$ be a function where
$f(u, v) = 1$ when
$\level(u) = \level(v)$ and $f(u, v) = 0$ when $\level(u) \neq \level(v)$. Using
the potential functions defined in~\cite{BHNT15}, for some constant $\lambda >
0$: %

\begin{align}
    \Pi &= \sum_{v \in V}\Phi(v) + \sum_{e \in E}\Psi(e)\\
    \Phi(v) &= \lambda\sum_{i = 0}^{\level(v) - 1}\max(0, \coeff(1+\delta)^{\gn(i)} -
    |N(v, Z_i)|)\label{eq:phi}\\
    \Psi(u, v) &= 2\left(K - \min(\level(u), \level(v))\right) + f(u, v)\label{eq:psi}
\end{align}
We first calculate the potential changes for insertions and deletions of edges.

\paragraph{Insertion} The insertion of an edge $(u, v)$ creates a new edge with
potential $\Psi(u, v)$. The new potential has value at most $2K + 1$.
With an edge insertion $\Phi(u)$ and $\Phi(v)$ cannot increase.
Thus, the potential increases by at most $2K + 1$.

\paragraph{Deletion} The deletion of edge $(u, v)$ increases potentials
$\Phi(u)$ and $\Phi(v)$ by at most $(2\lambda + 3)K$ and $2K$, respectively. %
It does not increase
any other potential since the potential of edge $(u, v)$ is eliminated.
\\\\
First it is easy to see that the potential $\Pi$ is always non-negative. Thus,
we can use the positive gain in potential over edge insertions and deletions to
pay for the decrease in potential caused by moving vertices to different levels.

Now we discuss the change in potential given a movement of a vertex to a higher or
lower level. Moving such a vertex decreases the potential and we show that this
decrease in potential is enough to pay for the cost of moving the vertex to a
higher or lower level.

\paragraph{A vertex $v$ moves from level $i$ to level $\dl(v) < i$ due
to~\cref{alg:delete}}

Since vertex $v$ moved down at least one level, this means that prior to the
move, its \ld is $\down(v) < (1+\delta)^{\gn(\ell(v)- 1)}$.
It is moved to a level $\dl(v)$ where its \ld
is at least $(1+\delta)^{\gn(\dl(v) - 1)}$ and its up-degree is at most
$\coeff(1+\delta)^{\gn(\dl(v))}$ (or it is moved to level $0$).

The potential before the move is at least

\begin{align*}
&\lambda\sum_{i = 0}^{\dl(v) - 1}\max\left(0, \coeff(1+\delta)^{\gn(i)} - |N(v,
Z_i)|\right)\\
&+ \sum_{i = \dl(v)}^{\level(v) - 1} (\lambda +
\add)(1+\delta)^{\gn(i)}
\end{align*}

since we only move a vertex to a lower level
if $\down(v) < (1+\delta)^{\gn(\level(v)-1)}$ and we move it to the closest level
$\dl(v)$ where~\cref{inv:degree-2} is no longer violated. To derive the second
term, since we moved vertex $v$ to level $\dl(v)$, we know that its degree
$|N(v, Z_{\dl(v)})| < (1+\delta)^{\gn(\dl(v))}$ (otherwise, we could have moved $v$
to level $\dl(v) + 1$). Then, substituting $(1+\delta)^{\gn(i)}$ for all levels $i
\geq \dl(v)$ into $\Phi(v)$ allows us to obtain $\sum_{i = \dl(v)}^{\level(v) -
1} (\lambda + 3)(1+\delta)^{\gn(i)}$.
Then, when it reaches its final level,
we know that it is at the highest level it can move to or at level $0$. In
both cases, $$\Phi(v) =
\lambda\sum_{i = 0}^{\dl(v) - 1}\max\left(0, \coeff(1+\delta)^{\gn(i)} - |N(v,
Z_i)|\right)$$ %
after the move.
In this case, $\Phi(v)$ decreases by at least
$\sum_{i = \dl(v)}^{\level(v) - 1} (\lambda + \add)(1+\delta)^{\gn(i)}$.

We need to account for two potential increases: the increase in $\Psi$ and the
increase in $\Phi$ from neighbors of $v$. There are less than $(1+\delta)^{\gn(\dl(v))}$
such neighbors that we need to consider. Namely, there are less than $(1+\delta)^{\gn(\dl(v))}$
neighbors in levels $\geq \dl(v)$ that we need to consider for the potential increase.
This is due to the fact that we moved $v$ to the \emph{highest} level that satisfies the invariants.
If $v$ has $\geq (1+\delta)^{\gn(\dl(v))}$ neighbors in $Z_{\dl(v)}$, then the desire-level
of $v$ would be $\dl(v) + 1$ since $v$ satisfies~\cref{inv:degree-2} at level $\dl(v) + 1$
and we can increase its $\dl(v)$. Furthermore, we only need to consider 
neighbors in levels $\geq \dl(v)$ since only these neighbors will contribute to the potential increase
by~\cref{eq:phi} and~\cref{eq:psi}. 

We first consider the increase in $\Psi$.
The total potential increase in $\Psi(u, v)$ (\cref{eq:psi}) 
summed over the increase for every edge $(u, v)$ where $\level(u)
\geq \dl(v)$ is at
most $2(\level(v)-\dl(v))(1+\delta)^{\gn(\dl(v))}$.  
This is due to the fact that for each edge $(u, v)$, the
potential gain from $\Psi$ is upper bounded by $2$
for every level in $[\dl(v), \level(v) - 1]$. Thus, in total over $< (1+\delta)^{\gn(\dl(v))}$
such neighbors results in a total potential increase of less than $2(\level(v)-\dl(v))(1+\delta)^{\gn(\dl(v))}$.

Now we consider the potential increase in $\Phi$.
For this potential increase, we need to account for the increase in potential of every neighbor
whose edge is flipped by the move. 
Decreasing the degree of each neighbor by one for each of
$|N(v, \dl(v))| < (1+\delta)^{\gn(\dl(v))}$ neighbors results in the total increase in $\Phi$.
In other words, for each flipped edge $(v, w)$, $N(w, Z_i)$
decreases by $1$ for each level $i \in [\dl(v) + 1, \ell(v)]$. The total increase in $\Phi$ is then
at most $\lambda(\ell(v) - \dl(v))(1+\delta)^{\gn(\dl(v))}$
by~\cref{eq:phi} %
over all flipped edges since there are 
less than $(1+\delta)^{\gn(\dl(v))}$ such neighbors in levels $\geq \dl(v) + 1$ 
and so the total number of flipped edges
is less than $(1+\delta)^{\gn(\dl(v))}$. 

Then, in total, the potential decrease is at least

\begin{align*}
    \left(\sum_{i = \dl(v)}^{\level(v) - 1} (\lambda + \add)(1+\delta)^{\gn(i)}\right)
    &- 2(\ell(v) - \dl(v))(1+\delta)^{\gn(\dl(v))}\\ 
    - \lambda(\ell(v) - \dl(v))(1+\delta)^{\gn(\dl(v))}
&\geq (\ell(v) -
\dl(v))(1+\delta)^{\gn(\dl(v))}
\end{align*}

which is enough to pay for the at most $(1+\delta)^{\gn(\dl(v))}$ edge flips as well
as the $O(\ell(v) - \dl(v))$ work for computing the desire-level.
The total number of edge flips is upper bounded by $|N(v, \dl(v))|$. Since we
moved $v$ to $\dl(v)$ and not $\dl(v) + 1$, we know that $v$
satisfies~\cref{inv:degree-2} at $\dl(v)$ and not at $\dl(v) + 1$. Then, this
means that $|N(v, \dl(v)))| < (1+\delta)^{\gn(\dl(v))}$. Hence, our number of edge
flips is also bounded by $(1+\delta)^{\gn(\dl(v))}$.

\paragraph{A vertex $v$ moves from level $i$ to level $i + 1$ due
to~\cref{alg:insert}}
In order for~\cref{alg:insert} to move a vertex from level $i$ to $i+1$, it must
have violated~\cref{inv:degree-1} and that $\up(v) >
(2+3/\lambda)(1+\delta)^{\gn(i)}$ before the move.
Before and after the move, $\Phi(v) = 0$, since in these cases $\down(v) >
\coeff(1+\delta)^{\gn(i-1)}$ and $\down(v) > \coeff(1+\delta)^{\gn(i)}$,
respectively. 
Thus, $\Phi(v)$ does not change in
value. Furthermore, the $\Phi(w)$ of its neighbors $w$ cannot increase.
Then, this leaves us with the potential change in $\Psi(v, w)$.

$Z_i$ is the set of neighbors that $v$ has to iterate through within its
data structures if $v$ goes up a level. The potential decrease for every
neighbor of $v$ on $i = \level(v)$ is $1$.
The potential decrease for every neighbor on level $i+1$ is $1$.
Finally, the potential decrease for every neighbor in levels $> \dl(v)$ is $2$.
Then, the potential decrease for every neighbor in $Z_i$ is at least $1$
and is enough to pay for the $O(|Z_i|)$ cost of iterating and moving the neighbors
of $v$ in its data structures.

\paragraph{Parallel Amortized Work}
The last part of the proof that needs to be shown is that any set of parallel level
data structure operations that is undertaken by~\cref{alg:insert}
or~\cref{alg:delete} has a sequential set of operations of the form detailed
above (i.e., moving $v$ to $\dl(v)$ or moving $v$ from level $i$ to $i+1$)
that consists of the same or strictly larger set of operations.

\begin{lemma}\label{lem:same-seq}
    For any set of operations performed in parallel by~\cref{alg:insert}
    or~\cref{alg:delete}, there exists an identical set of sequential operations
    to the set of parallel operations.
\end{lemma}

\begin{proof}
    In~\cref{alg:insert}, the parallel set of operations consists of moving all
    vertices that violate~\cref{inv:degree-1} in the same level $i$ up to
    level $i + 1$. Again, suppose we choose an arbitrary order to move the
    vertices in level $i$ to level $i + 1$. Given two neighbors in the order $v$
    and $w$, if $v$ moves to level $i+1$, the \hd of $w$ still includes $v$;
    since the \hd of any vertex $w$ is not affected by the previous vertices that
    moved to level $i + 1$, $w$ moves to $i + 1$ on its turn. This order
    provides a sequential set of operations that is equivalent to the parallel
    set of operations.

  In~\cref{alg:delete}, the parallel set of operations consists of moving a set
    of vertices down from arbitrary levels to the same level $i$. We show that there
    exists an identical set of sequential operations to the parallel
    operations. First, any vertex whose
    $\dl(v) = i$ considered all vertices in levels $\geq i - 1$ in its
    calculation of $\dl(v)$.
    Thus, any other vertex $w$ moving from a level $j > i$ to level $i$ is
    included in calculating the desire-level of vertex $v$. Suppose we pick an
    arbitrary order to move the vertices that have $\dl(v) = i$ to level $i$.
    Then, the desire-level of any vertex $w$ whose $\dl(w) = i$ does not change
    after $v$ is moved to level $i$. Hence, when it is $w$'s turn in the order,
    $w$ moves to level $i$. This arbitrary order is a sequential set of
    operations that is identical to the parallel set of operations.
\end{proof}
\fi

\begin{lemma}\label{thm:work}
     For a batch of $\batchsize<m$ updates, \cref{alg:rebalance} requires
     $O(\batchsize\log^2n)$ amortized work with high probability.
    The required space
    is $O(n\log^2 n + m)$ using the randomized
    data structures. %
\end{lemma}
\iflong
\begin{proof}
    Our potential argument handles the cost of moving neighbors of a vertex $v$
    between different levels. Namely, our potential argument shows that such
    costs of updating neighbor lists of nodes require $O(\log^2 n)$ amortized
    work per edge update to the structure since we showed that the $O(\log^2 n)$
    potential increase from each edge insertion or deletion is enough to pay for the 
    cost of moving vertices to different levels.

    Then, it remains to calculate the
    amount of work of~\cref{alg:desire-level}.
    We can obtain the size of each neighbor list in $O(1)$ work and depth.
    If we show that the work of
    running~\cref{alg:desire-level} is asymptotically
    bounded by the work of calculating the set of neighbor vertices that need to
    be moved between neighbor lists for a vertex, then we can also charge this
    work to the potential. To compute the first lower bound
    on $\dl(v)$, we maintain a cumulative
    sum of the total number of neighbors for each vertex at or below the current
    level $\level(v)$. %
    Then, we sequentially double the number of elements we use to compute
    the next level. We use $O(\level-\dl(v))$ work to compute $\dl(v)$.

    Finally, we also bound the work of the final binary search. Let $R$ be the
    size of the range of values in which we perform our binary search. The size of
    the number of possible levels becomes
    smaller as we decrease our range of values to search.
    Whenever we go right in the binary search, we perform $R/2$ work.
    Whenever we go left in the binary search, we also perform at most
    $R/2$ work. Thus, the total amount of work we perform while
    doing the binary search is $O(R)$. And by the argument above, the
    amount of work is $O\left(|Z_{\dl(v)} \setminus Z_{\level(v)}|\right)$.

    The total work of~\cref{alg:desire-level} is $O(|Z_{\dl(v)} \setminus
    Z_{\level(v)}| + (\level-\dl(v)))$ which we can successfully charge to the
    potential.
    We conclude that the amount of work per update is $O(\log^2 n)$.
\end{proof}

\subsection{Estimating the Coreness and Orientation}\label{sec:compute-core}

\ifanon
\else
To obtain a low-outdegree orientation
we maintain an orientation of each edge from
the endpoint at a lower level to the endpoint with a higher level. So
edges are directed from lower to higher levels. For any two vertices on
the same level, we direct the edge from the vertex with higher ID to the
vertex with lower ID. Performing this procedure results in an $8\alpha$
outdegree orientation, where $\alpha$ is the arboricity of the
graph.
\fi

\myparagraph{$(2+\eps)$-Approximation of Coreness}\label{sec:coreness}
The \emph{coreness estimate},
$\kest(v)$, is an estimate of the coreness of a vertex $v$.
We compute
a coreness estimate using \emph{only} $v$'s level 
and the number of levels per group (which is fixed). 
We show how to use such information to obtain a
$(2+\eps)$-approximation to the actual coreness of $v$ for any constant
$\eps > 0$. (We can find an approximation for any fixed $\eps$
by appropriately setting $\delta$ and $\lambda$.)
To calculate $\kest(v)$, we
find the largest index $i$ of a
group $\group_i$, where $\ell(v)$ is at least as high as the highest level
in $\group_i$.

\begin{definition}[Coreness Estimate]\label{def:core-estimate-number}
    The \emph{coreness estimate} $\kest(v)$ of vertex $v$ is
    $(1+\delta)^{\max(\floor{(\level(v) + 1)/4\ceil{\log_{1+\delta} n}}-1, 0)}$, where
    each group has $4\ceil{\log_{(1+\delta)} n}$ levels.
\end{definition}

To see an example, consider vertex $y$ in~\cref{fig:deletion-exp} $(e)$.
We estimate $\kest(y) = 1$ since the highest level that is the last level of
a group and is equal to or below level $\level(y) = 4$ is level $2$.
Level $2$ is part of group $0$, and  so our coreness estimate for $y$ is
$(1+\delta)^0 = 1$. This is a $2$-approximation of its actual coreness of
$2$.
Using~\cref{def:core-estimate-number}, we prove that our PLDS maintains a $\coeff(1+\delta)$-approximation of the coreness value of each vertex, for any constants $\lambda >
0$ and $\delta > 0$. 
Therefore, we obtain the following lemma giving the desired
$(2+\eps)$-approximation.
Our experimental analysis shows that our theoretical bounds limit the maximum 
error of our experiments, although our average errors are much smaller. 
To get a maximum error bound of $(2+\eps)$ for
any $\eps > 0$, we
can set $\delta = \eps/3$ and $\lambda = \frac{9}{\eps} + 3$.

By~\cref{lem:decomp-bound}, it suffices to return $\kest(v)$ as the estimate of
the coreness of $v$;
\fi
\iflong
this proves the approximation factor in
\cref{thm:batch-dynamic}.
\fi
\ificml
The proof of the error bound can be
found in our full paper.
\fi

\begin{lemma}\label{lem:core-num}
    Let $\kest(v)$ be the coreness estimate
    and $\core(v)$ be the coreness of $v$,
    respectively. If $\core(v) >
    \coeff(1+\delta)^{g'}$, then
    $\kest(v) \geq (1+\delta)^{g'}$. Otherwise, if $\core(v) <
    \frac{(1+\delta)^{g'}}{\coeff(1+\delta)}$,
    then $\kest(v) < (1+\delta)^{g'}$.
\end{lemma}

\begin{proof}
        For simplicity, we assume the number of levels per group is
    $4\ceil{\log_{(1+\delta)} m} + 1$ (a tighter analysis can accommodate the
    case when the number of levels per group is $\ceil{\log_{(1+\delta)}m}$).
    Let $\topmost(g')$ be the topmost level of group $g'$.
    In the first case, we show that if $\core(v) > \coeff(1+\delta)^{g'}$, then $v$
    would be in a level higher than
    $\topmost(g')$
    in our level data structure.
    This would also imply that $\kest(v) \geq (1+\delta)^{g'}$.
    Suppose for the sake of contradiction
    that $v$ is located at some level $\level(v)$ where
    $\level(v) \leq \topmost(g')$.
    This means that $\up(v) \leq \coeff(1+\delta)^{g'}$ at level
    $\ell(v)$. Furthermore, by the invariants of our level data structure,
    each vertex $w$ at the same or lower level has $\up(w) \leq \coeff(1+\delta)^{g'}$.
    This means that when we remove all vertices starting at level $0$
    sequentially up to and including $\ell(v)$, all vertices removed have degree $\leq
    \coeff(1+\delta)^{g'}$ when removed. Thus, when we reach $\ell(v)$, $v$ also has
    degree $\leq \coeff(1+\delta)^{g'}$. This is a contradiction with $\core(v) >
    \coeff(1+\delta)^{g'}$. It must then be the case that $v$ is at a level higher than
    $T(g')$ and $\kest(v) \geq (1+\delta)^{g'}$.

    Now we prove that if $\core(v) < \frac{(1+\delta)^{g'}}{\coeff(1+\delta)}$, then
    $\kest(v) < (1+\delta)^{g'}$. We assume for sake of contradiction that $\core(v) <
    \frac{(1+\delta)^{g'}}{\coeff(1+\delta)}$ and $\kest(v) \geq (1+\delta)^{g'}$.
    To prove this case, we consider the following process, which we call the
    \emph{pruning} process. Pruning is done on a subgraph $S \subseteq G$. We
    use the notation $d_S(v)$ to denote the degree of $v$ in the subgraph
    induced by $S$.
    For a given subgraph $S$, we \emph{prune} $S$ by repeatedly removing
    all vertices $v$ in $S$ whose $d_S(v) <
    \frac{(1+\delta)^{g'}}{\coeff(1+\delta)}$.
    Note that in this argument, we need only consider levels from the same group
    $g'$ before we reach a contradiction, so we assume that all levels are in the group
    $g'$. Let $j$ represent the number of levels below level
    $T(g')$. (Recall that because $\kest(v) \geq (1+\delta)^{g'}$, $\level(v) \geq
    T(g')$, if we consider a level $\level(v) > T(g')$, then the \ld cannot decrease
    due to~\cref{inv:degree-2} becoming stricter.
    This only makes our proof
    easier, and so for simplicity, we consider $\level(v) = T(g')$.)
    We prove via induction that
    the number of vertices pruned from the subgraph induced
    by $Z_{T(g') - j}$ \emph{must}
    be at least
    \begin{align*}
        \left(\frac{(2+3/\lambda)(1+\delta)}{2}\right)^{j - 1}\left((1+\delta)^{g'} -
    \frac{(1+\delta)^{g'}}{\coeff(1+\delta)}\right)
    \end{align*}
    or otherwise, $\core(v) \geq \frac{(1+\delta)^{g'}}{\coeff(1+\delta)}$.
    We first prove the base case when $j = 1$.
    In the base case, we know that
    $d_{Z_{T(g') - 1}}(v) \geq (1+\delta)^{g'}$ by~\cref{inv:degree-2}.
    Then, if fewer than
	$(1+\delta)^{g'} - \frac{(1+\delta)^{g'}}{\coeff(1+\delta)}$
    neighbors of $v$ are pruned from the graph, then $v$ is part of a $\geq
    \frac{(1+\delta)^{g'}}{\coeff(1+\delta)}$-core and $\core(v) \geq
    \frac{(1+\delta)^{g'}}{\coeff(1+\delta)}$, a contradiction.

    Thus, at least $(1+\delta)^{g'} - \frac{(1+\delta)^{g'}}{\coeff(1+\delta)}$
    vertices must be pruned in this case.
    We now assume the induction hypothesis for $j$
    and prove that this is true for step $j + 1$.
    By~\cref{inv:degree-2}, each vertex on level $T(g') - j$ and above has
    degree \emph{at least} $(1+\delta)^{g'}$ in graph $Z_{T(g') - j - 1}$. Then,
    in order to prune all $X$  vertices from the previous induction step, we must
    prune at least $\frac{(1+\delta)^{g'}X}{2}$ edges, since each edge decreases
    the degree of two vertices by $1$; all adjacent edges of a pruned vertex are
    also pruned/removed.
    Each vertex that is pruned can remove
    \emph{at most} $\frac{(1+\delta)^{g'}}{\coeff(1+\delta)}$ edges when it is
    pruned, by definition of our pruning procedure since we prune vertices with
    degree $< \frac{(1+\delta)^{g'}}{\coeff(1+\delta)}$. Thus, the \emph{minimum} number
    of vertices we must prune in order to prune the $X =
    \left(\frac{(2+3/\lambda)(1+\delta)}{2}\right)^{j - 1}\left((1+\delta)^{g'} -
    \frac{(1+\delta)^{g'}}{\coeff(1+\delta)}\right)$ vertices from the previous step is
    
    \begin{align*}
        &\frac{\text{\# edges that need to be pruned}}{\text{max \# edges pruned per pruned
        vertex}} = \frac{(1+\delta)^{g'}X}{2
        \frac{(1+\delta)^{g'}}{\coeff(1+\delta)}}\\
        &= \frac{(2+3/\lambda)(1+\delta)}{2} X
    \end{align*}
    
    \begin{align*}
        &=
        \left(\frac{(2+3/\lambda)(1+\delta)}{2}\right)^{j}\left((1+\delta)^{g'}
        -\frac{(1+\delta)^{g'}}{\coeff(1+\delta)}\right).
    \end{align*}

    Thus, we have proven our argument for the $(j + 1)$-st induction step.
    Note that for $j =\ceil{ \log_{(2+3/\lambda)(1+\delta)/2} (4m + 1) }$, we have $j\leq
    4\ceil{\log_{(1+\delta)}(m)} + 1$. This is because, since we pick $\lambda$ to
    be a constant greater than $0$,
    $2 + 3/\lambda > 2$ and for large enough $m$,
    $\log_{(2+3/\lambda)(1+\delta)/2}(4m + 1) \leq 4\ceil{\log_{(1+\delta)}(m)}
    + 1$.
    Then, by our induction, if we substitute $4\ceil{\log_{(1+\delta)}(m)} + 1$ for
    $j$,
    \begin{align*}
        &\left(\frac{(2+3/\lambda)(1+\delta)}{2}\right)^{4\ceil{\log_{(1+\delta)}(m)}}\left((1+\delta)^{g'} -
            \frac{(1+\delta)^{g'}}{\coeff(1+\delta)}\right)\\
        &> 4m \cdot (1/2) = 2m.
    \end{align*}
    This means we must prune at least $2m + 1$ vertices at this step, which we cannot because
    there are at most $2m$ vertices in a level that is not level $0$.
    This last step holds because
    all vertices with degree $0$ must be on the first level. Hence,
    all vertices not on level $0$ must be adjacent to at least one edge,
    and $n \leq 2m$ where $n$ is the number of vertices on the level that is not
    level $0$.
    Thus, our assumption is incorrect and we
    have proven our desired property.
\end{proof}

We show that~\cref{lem:core-num} implies~\cref{lem:decomp-bound}.

\begin{lemma}\label{lem:decomp-bound}
    The coreness estimate $\kest(v)$ of a vertex $v$
     satisfies $\frac{\core(v)}{(2+\eps)} \leq \kest(v) \leq
    (2+\eps)\core(v)$ for any constant $\eps > 0$. %
\end{lemma}

\begin{proof}
    Suppose $\kest(v) = (1+\delta)^g$. Then, by~\cref{lem:core-num},
we have
$        \frac{(1+\delta)^{g}}{\coeff(1+\delta)} \leq \core(v) \leq
        \coeff(1+\delta)^{g+1}$.
        Then, substituting $\kest(v) = (1+\delta)^{g}$ and
        solving the above bounds, $\frac{\core(v)}{\coeff(1+\delta)} \leq
    \kest(v) \leq \coeff(1+\delta)\core(v)$.
    For any constant $\eps> 0$, there exists constants $\lambda, \delta > 0$
    where $\frac{\core(v)}{2(1+\eps)} \leq \kest(v) \leq
    2(1+\eps)\core(v)$.
\end{proof}

\iflong

\ificml
Furthermore, using similar analysis to that of \cite{HNW20}, our PLDS directly gives a
low out-degree orientation with the same complexity bounds if we orient all edges towards neighbors at higher
levels and orient edges from high ID to low ID
between neighbors in the same level.

\begin{theorem}\label{thm:batch-dynamic-out-k-core}
   Provided an input graph with $m$ edges, and a batch of $\batch$
    updates, our algorithm maintains a $(2 + \eps)$-approximation of
    the coreness values and a $(8+\eps)\alpha$-low-outdegree orientation
    for all vertices (for any constant $\eps > 0$)
    in $O(\batch\log^2 m)$ amortized work and $O(\log^2 m \log\log m)$
    depth \whp{}, using $O(n\log^2 m + m)$ space.
\end{theorem}
\fi

\ifanon
\else
    Using our deterministic and space-efficient data structures, we can obtain
    the following additional results.

\begin{corollary}\label{cor:batch-dynamic-det}
    Provided an input graph with $m$ edges, and a batch of updates $\batch$,
    our algorithm maintains a $(2 + \eps)$-approximation of
    the coreness values for all vertices (for any constant $\eps > 0$)
    in $O(|\batch|\log^2 m)$ amortized work and $O(\log^3 m)$
    depth worst-case, using $O(n\log^2 m + m)$ space.
\end{corollary}

\begin{corollary}\label{cor:batch-dynamic-space}
    Provided an input graph with $m$ edges, and a batch of updates $\batch$,
    our algorithm maintains a $(2 + \eps)$-approximation of
    the coreness values for all vertices (for any constant $\eps > 0$)
    in $O(|\batch|\log^2 m)$ amortized work and $O(\log^4 m)$
    depth worst-case, using $O(n + m)$ space.
\end{corollary}

Using this data structure, we show a bound for the outdegree
if we orient all edges towards neighbors at higher
levels and orient edges from higher ID to lower ID in the same level.

Our data structure also provides an approximation of the densest subgraph within
the input graph by the Nash-Williams theorem.

\begin{corollary}\label{cor:densest-subgraph}
    The PLDS maintained by~\cref{alg:rebalance} provides an $(8 +
    \eps)$-approximation on the density of the densest subgraph within the
    input graph.
\end{corollary}
\fi

For arbitrary batch sizes, getting better than a 2-approximation for coreness values is P-complete~\cite{anderson84pcomplete}, and so there is unlikely to exist a polylogarithmic-depth algorithm with such guarantees.

\begin{proof}[Proof of~\cref{thm:low-outdegree}]
The approximation factor for our algorithm is given by~\cref{lem:decomp-bound}. The work and
depth bounds of our algorithm is given by~\cref{thm:work} and~\cref{cor:random-depth}. Altogether, we prove our main
theorem.
\end{proof}

\begin{corollary}\label{cor:densest-subgraph}
    Our algorithm maintains a $(4+\eps)$-approximation of the densest
    subgraph value with the same bounds as~\cref{thm:batch-dynamic}.
\end{corollary}

\begin{proof}
    The degeneracy $d$ of our input graph is equal to the maximum
    core number of any node in the graph. It is a well-known fact
    that $\frac{d}{2} \leq \alpha \leq d$ where $\alpha$ is the 
    arboricity of the graph. By the Nash-Williams theorem, the arboricity
    of the input graph $G = (V, E)$ is equal to $\ceil{\max_{S \subseteq G} \frac{|E(S)|}{|V(S)| - 1}}$. Hence, returning the maximum approximate core number returned by our algorithm gives a
    $(4+\eps)$-approximation on the densest subgraph value. 
\end{proof}

\subsection{$O(\alpha)$ Out-Degree Orientation}\label{sec:orient}

We orient all edges from vertices in lower 
levels to higher levels, breaking ties for 
vertices on the same level by using their
indices. Such an orientation can be maintained dynamically
in the same work and depth as our \plds via
a parallel hash table keyed by the edges and where the
values give the orientation. 
Specifically, we require the following data structures for maintaining a low out-degree
orientation. First, we maintain a parallel hash table, $H$,
containing the edges of the graph. 
The edge $(u, v)$ is the key in the hash table
where $u \pred v$ (i.e.\ the index of $u$ is less than the index of $v$). 
The value for key $(u, v)$ is $0$ if the edge is oriented from $u$ to
$v$ and $1$ if the edge is oriented from $v$ to $u$. The pseudocode
is shown in~\cref{alg:low-outdegree}. Additionally, we make a slight
modification to our update algorithm that keeps track of the edges that
were searched when a vertex moves to a higher or lower level.
The pseudocode for our algorithm is given in~\cref{alg:low-outdegree}.

\begin{algorithm}[!t]\caption{$\lowoutdegorient(\batch)$}
    \label{alg:low-outdegree}
    \small
    \begin{algorithmic}[1]
    \Require{A batch $\batch$ of updates.}
    \Ensure{A set of edges $F$ that were flipped after processing the 
    batch of updates. An edge $(u, v) \in F$ represents the orientation
    of the edge \emph{before} the flip. Also returns oriented updates $(u, v) \in \batch$ where for edge deletions $(u, v)$ is the orientation of the edge \emph{before} the deletion and for edge insertions $(u, v)$ is the orientation of the edge \emph{after} the insertion.}
    \State $F \leftarrow \emptyset$.
    \ParFor{each searched edge $(u, v)$ for a vertex that moved levels}
        \If{$H[(u, v)] = 0$ and ($(\level(u) > \level(v)$ or ($\level(u) = \level(v)$ and $v < u$))}
            \State $F \leftarrow F \cup (u, v)$.
        \ElsIf{$H[(u, v)] = 1$ and ($\level(v) > \level(u)$ or ($\level(u) = \level(v)$ and $u < v$))}
            \State $F \leftarrow F \cup (u, v)$.
        \EndIf
    \EndParFor
    \State $J \leftarrow \emptyset$.
    \ParFor{each edge update $\{u, v\} \in \batch$}
        \If{$\{u, v\}$ is an insertion}
            \State Add to $J$ the orientation of
            edge \emph{after} processing $\batch$.
        \Else
            \State Add to $J$ the orientation
            of edge \emph{before} processing $\batch$.
        \EndIf
    \EndParFor
    \Return $F, J$.
    \end{algorithmic}
\end{algorithm}

\begin{proof}[Proof of~\cref{cor:arboricity-orientation}]
    Let the degeneracy of the graph be $d$. As is well-known, the degeneracy of
    the graph is equal to $k_{max}$ where $k_{max}$ is the maximum \kc of the
    graph.  Furthermore, it is well-known that $\frac{d}{2} \leq \alpha \leq d$.
    By~\cref{lem:decomp-bound}, the vertices in the largest \kc in the graph are
    in a level with group number at most $\log_{(1+\delta)}((2+3/\lambda)(1 + \delta)d) +
    1$. This means that the \hd of each vertex in that group is at most
    $(2+3/\lambda)(1+\delta)^{\log_{(1+\delta)}((2+3/\lambda)(1 + \delta)d)} = (4 + \eps)d$
    for any constant $\eps > 0$ for appropriate settings of $\lambda, \delta > 0$.
    We then also obtain an $(8+\eps)\alpha$ out-degree orientation
    where $\alpha$ is the arboricity of the graph.
\end{proof}

\subsection{Deterministic and Space-Efficient Data Structures}\label{app:data-structure-impl}

In addition to the randomized data structures presented in Section 3.4,
we present two additional sets of data structures that we can use to
obtain a \emph{deterministic} and a \emph{space-efficient} $(2 +
\eps)$-approximate \kc algorithms.

The work of all of our randomized, deterministic, and space-efficient
algorithms are the same; however, using randomization allows us to obtain a
better depth with slightly less complicated data structures.

\paragraph{Deterministic Data Structures}
We initialize an array $\greater$, of size $n$.
Each vertex is assigned a unique index in $\greater$. The entry for the
$i$'th vertex, $\greater[i]$, contains a pointer to a dynamic array
that stores the neighbors of vertex $v_i$ at levels $\geq \level(v_i)$.
Each vertex $v_i$ also stores another dynamic array,
$\tbl_{v_i}$, that contains pointers to a set of dynamic arrays storing
the neighbors of $v_i$ partitioned by their levels $j$ where $j < \level(v_i)$.
Specifically, we maintain
a separate dynamic array for each
level from level $0$ to level $\level(v_i) - 1$ storing the neighbors of
$v_i$ at each respective level.
We also maintain the current level of each vertex in an array.

To perform a batch of insertions into a dynamic array, we insert the elements at
the end of the array. The array is resized and doubles in size if too many
elements are inserted into the array (and it exceeds its current size). For a batch
of deletions, the deletions are initially marked with a ``deleted'' marker
indicating that the element in the slot has been deleted. A counter is used to
maintain how many slots contain ``deleted.'' Then, once a constant fraction of
elements (e.g.\ $1/2$) has ``deleted'' marked in their slots, the array is cleaned
up by reassigning vertices to new slots and resizing the array.

\begin{lemma}\label{lem:depth}
    \cref{alg:rebalance} returns a deterministic level data structure that maintains
    \cref{inv:degree-1} and~\cref{inv:degree-2} and has $O(\log^3 n)$ worst-case
    depth and $O(n\log^2 n + m)$ space. %
\end{lemma}

\begin{proof}
    All edge updates can be partitioned into $\batch_{ins}$ and $\batch_{del}$
    in parallel in $O(\log n)$ depth. Then, it remains to bound the depth
    of~\cref{alg:insert} and~\cref{alg:delete}.

    \cref{alg:insert} iterates through all $K = O(\log^2 n)$ levels
    sequentially. By~\cref{lem:move-up}, no vertices on levels $\leq i$ will
    violate~\cref{inv:degree-1} after processing level $i$. Thus, by the end of
    the procedure no vertices violate~\cref{inv:degree-1}.
    By~\cref{lem:batch-insertions-inv},~\cref{inv:degree-2} was never violated
    during~\cref{alg:insert}. Thus, both invariants are maintained at the end of
    the algorithm. Since we iterate through $O(\log^2 n)$ levels and, in each
    level, we require checking the neighbors at one additional level which can
    be done in parallel in $O(1)$ depth, the total depth of this procedure is
    $O(\log^2 n)$.
    For each level, an additional depth of $O(\log n)$ might be
    necessary to compute the element offsets and then resize the arrays.
    Then, \cref{alg:insert} requires $O(\log^3 n)$ worst-case depth.

    \cref{alg:delete} iterates through all $K = O(\log^2 n)$ levels
    sequentially. By~\cref{lem:batch-deletes} and~\cref{lem:smaller-level},
    after processing level $i$, no vertices on a level higher than $i + 1$
    will have $\dl(v) \leq i+1$ and no vertices on levels $\leq i$ will
    violate~\cref{inv:degree-1}. Thus, by the
    end of the procedure all vertices satisfy~\cref{inv:degree-2}.
    Furthermore,~\cref{inv:degree-1} was never violated
    due to~\cref{lem:batch-deletes}. There are $O(\log^2 n)$ levels and for each
    level we require running~\cref{alg:desire-level} to obtain the $\dl(v)$ of
    each affected vertex $v$ that should be moved to each level.

    Running~\cref{alg:desire-level} requires $O\left(\log \log n\right)$ depth
    to obtain the first level that satisfies invariants
    for each affected vertex $v$ and $O\left(\log \log
    n \right)$ depth for the final binary search that determines the closest
    level to $\level(v)$ that satisfies the invariants.
    In conclusion, ~\cref{alg:delete} requires $O(\log^3 n)$ worst-case depth.

    Altogether,~\cref{alg:rebalance} requires $O(\log^3 n)$ worst-case depth.
\end{proof}

\paragraph{$O(m)$ Total Space Data Structures}

Here we describe how to reduce the total space usage of our data structures to $O(m)$.
All of our previous data structures use
$O(n\log^2 n + m)$ space, which means that when $m = O(n)$, we use space
that is superlinear in the size of the graph. To reduce the total space to $O(m)$,
we maintain two structures for $\tbl_{v_i}$. We can use either
the deterministic or randomized structures for the other structures. Each
$\tbl_{v_i}$ is maintained as a linked list.
The $j$'th node in the linked list
maintains the number of neighbors of $v_i$ at
the $j$'th non-empty level (a non-empty level is one where $v_i$ has neighbors at
that level) that is less than $\level(v_i)$. The node representing a level is
removed from the linked list when the level becomes empty.
Each node in $\tbl_{v_i}$ contains pointers to vertices at the level represented
by the node.
Each vertex then contains pointers to every edge it is
adjacent to and every edge contains pointers to the two nodes in the two linked
lists representing the levels on which the endpoints of the edge reside. Using
either dynamic arrays or hash tables for the lists of neighbors allow us to
maintain these data structures in $O(m)$ space. Since we only maintain a node in
our linked list for every non-empty level, our linked list contains $O(m)$ nodes.

Using the data structures above, we can prove equivalent results
to~\cref{thm:batch-dynamic}.

\begin{lemma}\label{cor:space-efficient}
    \cref{alg:rebalance} returns a deterministic level data
    structure that maintains~\cref{inv:degree-1} and~\cref{inv:degree-2} and has
    $O(\log^4 m)$ depth, while using $O(m)$ space.
\end{lemma}

\begin{proof}
    The proof is the same as the proof of~\cref{lem:depth} except that we
    replace~\cref{alg:desire-level} with a linear search in the linked list, which has size at most the number of levels, which is $O(\log^2 n)$. The specific data structure we use for each vertex $v$ is a linked list
    with each node of the linked list representing a level $\leq \ell(v) - 1$
    which contains one or
    more neighbors of $v$. Then, each node in the linked list contains a pointer to
    a dynamic array containing the neighbors in that level.
    Thus, the total depth is $O(\log^4 n)$.
\end{proof}

\subsubsection{Overall Work and Depth Bounds}\label{sec:overall-work-bounds}
Our deterministic and space-efficient structures also give the following
corollary using our above depth bound arguments.

\fi

Using~\cref{lem:depth} %
and~\cref{cor:space-efficient}, %
we obtain the following two corollaries.

\begin{corollary}\label{cor:deterministic-result}
    For a batch of $\batchsize <m$ updates, \cref{alg:rebalance} returns a
    PLDS that maintains~\cref{inv:degree-1} and~\cref{inv:degree-2}
    in $O(\batchsize\log^2n)$ amortized work and $O(\log^3 n)$
    depth, using $O(n\log^2 n + m)$ space.
\end{corollary}

\begin{corollary}\label{cor:space-efficient-result}
    For a batch of $\batchsize <m$ updates, \cref{alg:rebalance} returns a
    PLDS that maintains~\cref{inv:degree-1} and~\cref{inv:degree-2}
    in $O(\batchsize\log^2n)$ amortized work and $O(\log^4 m)$
    depth, using $O(m)$ space.
\end{corollary}
\subsection{Handling Vertex Insertions and Deletions}\label{sec:vertex-ins-del}
We can handle vertex insertions and deletions by inserting vertices that have
zero degree and considering deletions of vertices to be a batch of edge
deletions of all edges adjacent to the deleted vertex. When we insert a vertex
with zero degree, it automatically gets added to level $0$ and remains in level
$0$ until edges incident to the vertex are inserted. For a vertex deletion, we
add all edges incident to the deleted vertex to a batch of edge deletions. Note,
first, that all vertices which have $0$ degree will remain in level $0$. Thus,
there are at most $O(m)$ vertices which have non-zero degree.

In this setting, we may need to rebuild the PLDS from scratch. 
Instead of maintaining $\ceil{4\log^2 n}$ levels, we maintain $\ceil{8\log^2 n}$ levels in this setting.
Doubling the number of levels is a very loose bound to ensure that we can handle two times the number
of vertices in the graph before we perform a rebuild of our entire structure.
To maintain $O(\log^2 n)$ levels in our data structure,
we rebuild the data structure once we have made $n/2$ vertex updates. 
Rebuilding the data structure requires
$O(n\log^2 n)$ total work which we can amortize to the $n/2$ vertex updates to obtain
$O(\log^2 n)$ amortized work \whp Running~\cref{alg:insert} and~\cref{alg:delete} on
the entire set of $O(n + m)$ vertices and edges requires $O(\poly\log n)$ depth \whp{}
depending on the specific set of data structures we use.

Lastly, in order to obtain a set of vertices which are re-numbered consecutively
(in order to maintain our space bounds), we perform parallel integer sort or
hashing.

\section{Experimental Evaluation}\label{sec:eval}
In this section, we compare the performance of our dynamic \plds
with existing approaches on a set of large real-world graphs.
Our results show that our algorithms consistently achieve speedups, by up to two
orders of magnitude, compared with all of the previous state-of-the-art dynamic
$k$-core decomposition algorithms.

\myparagraph{Evaluated Algorithms}
We evaluate two versions of our algorithm: $\textbf{\emph{PLDS}}$: an exact implementation of our
theoretical algorithm and $\textbf{\emph{PLDSOpt}}$: a
version with $\ceil{\log_{1 + \delta} n/50}$ levels per
group. $\textbf{\emph{PLDS}}$ maintains the
approximation guarantees given by~\cref{lem:decomp-bound}, while
$\textbf{\emph{PLDSOpt}}$ achieves better performance while
maintaining slightly worse approximation bounds.

We compare our algorithms with the following \emph{dynamic}
implementations:
\textbf{\emph{Sun}}: the sequential, approximate algorithm of Sun
\etal~\cite{SCS20}, specifically their faster, round-indexing algorithm,
which is publicly available~\cite{suncode};
\textbf{\emph{Hua}}: the parallel, exact algorithm of Hua
\etal~\cite{Hua2020}, kindly provided by the authors;
\textbf{\emph{Zhang}}: the sequential, exact algorithm of Zhang
and Yu~\cite{Zhang2019a}, kindly provided by the authors;
and \textbf{\emph{LDS}}: our implementation of the sequential, approximate
algorithm of Henzinger \etal~\cite{HNW20}, but using our coreness approximation
procedure in~\cref{sec:coreness}.
All are state-of-the-art algorithms, outperforming previous algorithms in their
respective categories.

We also implemented \textbf{\emph{ApproxKCore}},
our new static parallel approximate
$k$-core decomposition algorithm (\cref{thm:static}).
We compared it with \textbf{\emph{ExactKCore}}, the
state-of-the-art parallel, static, exact \kc algorithm of Dhulipala
\etal~\cite{dhulipala2017julienne}.

\myparagraph{Setup}
We use \texttt{c2-standard-60} Google Cloud
instances (3.1
GHz Intel Xeon Cascade Lake CPUs with a total of 30 cores with two-way hyper-threading, and 236 GiB RAM)
and \texttt{m1-megamem-96} Google Cloud instances (2.0 GHz Intel Xeon Skylake CPUs with a total of 48
cores with two-way hyper-threading, and 1433.6 GB
RAM). We use hyper-threading in our parallel experiments by default.
Our programs are written in C++, use a work-stealing scheduler~\cite{BlAnDh20}, and
are compiled using \texttt{g++} (version 7.5.0) with the \texttt{-O3}
flag.  We terminate experiments that take over 3 hours.
PLDS and PLDSOpt finished within 3 hours for all experiments. 

\myparagraph{Datasets} We test our algorithms on $11$ real-world
undirected graphs from SNAP~\cite{leskovec2014snap}, the DIMACS Shortest Paths
challenge road networks~\cite{dimacs}, and the Network
Repository~\cite{nr}, namely
\defn{dblp}, \defn{brain}, \defn{wiki},
\defn{orkut}, \defn{friendster}, \defn{stackoverflow},
\defn{usa}, \defn{ctr},
\defn{youtube}, and \defn{livejournal}.  We also used \defn{twitter}, a
symmetrized version of the Twitter network~\cite{kwak2010twitter}.
We remove duplicate edges, zero-degree vertices, and self-loops. \cref{table:sizes}
reflects the graph sizes \emph{after} this removal, and gives the largest $k$-core values. 
Both \emph{stackoverflow} and \emph{wiki} are temporal networks; for these, we maintain the
edge insertions and deletions in the temporal order
from SNAP. \emph{usa} and \emph{ctr} are two high-diameter road
networks and \emph{brain} is a highly dense human brain network
from NeuroData (\hyperlink{https://neurodata.io/}{https://neurodata.io/}).
All experiments are run on the \texttt{c2-standard-60} instances,
except for \textit{twitter} and \textit{friendster}, which
are run on the \texttt{m1-megamem-96} instances as they require more memory.

\begin{table}[t]
\caption{Graph sizes and largest values of $k$ for $k$-core decomposition.} \label{table:sizes}
\begin{center}
\footnotesize
\begin{tabular}[!t]{l|r|r|r}
\toprule
{Graph Dataset} & Num. Vertices & Num. Edges & Largest value of $k$\\
\midrule
{\emph{ dblp  }  }           & 317,080          &1,049,866   & 113\\
{\emph{ brain        }  }        & 784,262          &267,844,669  & 1200\\
{\emph{ wiki  }  } & 1,094,018        &2,787,967  & 124 \\
{\emph{ youtube }  }         & 1,138,499        &2,990,443 & 51 \\
{\emph{ stackoverflow }  }    & 2,584,164        &28,183,518 & 198 \\
{\emph{ livejournal  }  }    		 & 4,846,609  &42,851,237 & 372 \\
{\emph{ orkut        }  }    & 3,072,441        &117,185,083 & 253 \\
{\emph{ ctr     }  }     & 14,081,816       &16,933,413 & 3 \\
{\emph{ usa        }  }     & 23,947,347       &28,854,312  & 3\\
{\emph{ twitter      }  }    	 & 41,652,230       &1,202,513,046 & 2488\\
{\emph{ friendster}}         & 65,608,366       &1,806,067,135 & 304\\
\end{tabular}
\end{center}
\vspace{-1em}
\end{table}

\myparagraph{Ins/Del/Mix Experiments}
Our experiments are run for \emph{three different types of batched updates},
referred to by:
(1) \expins: starting with an empty graph,
\emph{all} edges are inserted in multiple size $|\batch|$ batches of insertion updates,
(2) \expdel: starting with the original graph, \emph{all} edges are deleted
in multiple size $|\batch|$ batches of deletion updates, 
and (3) \mix: starting with the initial
graph minus a random set $I$ of $|\batch|/2$ edges, a set $D$ of $|\batch|/2$ random edges is chosen
among the edges in the graph; then,
a single size $|\batch|$ mixed batch of updates with insertions $I$ and deletions $D$
is applied.
For the temporal graphs, \emph{stackoverflow} and \emph{wiki},
the order of updates in the batches follows the order
in SNAP~\cite{leskovec2014snap}. For the rest, 
updates are generated by taking two
random permutations of the edge list, one for \expins and one for \expdel.
Batches are generated by taking regular intervals of the
permuted lists.
For \mix, $I$ and $D$ are chosen uniformly at random.

Some past works only ran experiments in the \mix
setting~\cite{Hua2020,Zhang2019a}, while others~\cite{SCS20} also consider
\expins and \expdel. In this paper, we run experiments in all three
settings.
For \expins and \expdel, we consider the average running time across
all batches as a good indicator of how well the algorithm performs.
For \mix, we test each algorithm and dataset $3$ rounds each and take
the average.

We use the original timing functions provided by \hua, \sun, \zhang, and \ekcore.
We use the original code of \hua and \zhang for \mix and modify their code to
perform \expins and \expdel.
We note that \hua's timing function does not include the time to
process the graph and maintain their data structures; we include all
such times in our code. All other benchmarks also include this time.
If we include this time in their implementation, their running times
increase by  up to $8\times$  for some experiments. This explains some
of \hua's experimental performance improvements over the other
benchmarks.

The static algorithms, \ekcore and \akcore, are re-run on the entire
graph after each batch of updates in \expins and \expdel.
For the \mix batch,
we order all insertions in the batch before all deletions.
Then, we generate two static graphs per batch, one following
all insertions, and the other following all deletions.
We re-run the static algorithms on each static graph and take the
average of the times to obtain comparable per-batch running times. We
do this because some of the deletion updates may cancel the insertion
updates in the batch.

\subsection{PLDS Implementation Details}
We implemented our algorithms using the primitives from the Graph Based Benchmark Suite~\cite{dhulipala2018theoretically}.
We implemented the PLDS with work, depth, and space bounds given in~\cref{thm:batch-dynamic}.
One can choose to instead implement our space-efficient version of our data structure
in exchange for additional $\poly(\log n)$ factors in the theoretical depth. 

Our data structure uses
concurrent hash tables with linear
probing~\cite{shun2014phase},
which support $x$
concurrent insertions, deletions, or finds in $O(x)$ amortized work
and $O(\log^*x)$ depth \whp~\cite{Gil91a}.
For deletions, we used the folklore \emph{tombstone}
method: when an element is deleted, we mark the slot in the table as a
tombstone, which can be reused, or cleared during a table resize.
We also use
dynamic arrays, which support adding or deleting $x$ elements from the
end in $O(x)$ amortized work and $O(1)$ depth.

We first assign each vertex a unique ID in $[n]$.
Then, we maintain an array $\greater$ of size $n$ keyed by vertex ID
that returns a parallel
hash table containing neighbors of $v$ on levels $\geq
\level(v)$.
For each vertex $v$, we maintain a dynamic array $\tbl_v$ keyed by
indices $i \in [0, \level(v) - 1]$.
The $i$'th entry of the array contains a pointer to a parallel hash
table containing the neighbors of $v$ in level $i$.
Appropriate
pointers exist that allow $O(1)$ work to access elements
in structures. Furthermore, we maintain a hash table which contains
pointers to vertices $v$ where $\dl(v) \neq \level(v)$, partitioned by their levels.  This allows us to quickly
determine which vertices to move up (in~\cref{alg:insert}) or move
down (in~\cref{alg:delete}).

\iflong
We make one modification in our parallel
implementation of our insertion procedure from our theoretical algorithm
which is instead of moving vertices up level-by-level, we perform a parallel
filter and sort that calculates the desire-level of vertices we move up. This
results in more work theoretically, but we find that, practically, it results in
faster runtimes.
Also, notably, in practice, we optimized the performance of our PLDS
by considering $\ceil{\frac{\log_{(1+\delta)}m}{50}}$ levels per group instead
of $\ceil{\log_{(1+\delta)}m}$.
We also implemented a version of our structure that \emph{exactly
follows} our theoretical algorithm and compared the performance of
both structures. We see that even such a simple optimization resulted
in significant gains in performance, up to $23.89\times$.

\iflong
\fi

\fi

\subsection{Accuracy vs.\ Running Time}\label{sec:expaccuracy}
We start by evaluating the empirical error ratio of the per-vertex
core estimates given by our implementations (\pldsopt, \plds, \lds)
and Sun on \dblp and \lj, using batches of size $10^5$ and
$10^6$, respectively.
\cref{fig:exp1} shows the average batch time (in seconds) against the
average and maximum \emph{per-vertex} core estimate error ratio.
This error ratio is computed as
$\max\left(\frac{\kest(v)}{\core(v)},
\frac{\core(v)}{\kest(v)}\right)$ for each vertex $v$
(where $\kest(v)$ is the core estimate
and $\core(v)$ is the exact core value).
The average is the error ratio averaged across all vertices
and the maximum is the maximum error.
If the exact core number is $0$, we ignore
the vertex in our error ratio since our algorithm guarantees 
an estimate of $0$; for vertices of non-zero degree, the lowest
estimated core number is $1$ for all implementations.

The parameters we use for \pldsopt, \plds, and \lds
are all combinations of $\delta = \{0.2, 0.4, 0.8, 1.6, 3.2, 6.4\}
$ and $\lambda = \{3, \allowbreak 6, \allowbreak 12, \allowbreak 24, \allowbreak
48, \allowbreak 96\}$. We call these \emph{theoretically-efficient parameters}, since they maintain the work-efficiency of our algorithms.
For Sun, we use all combinations of their parameters  $\eps_{\textit{sun}} =
\lambda_{\textit{sun}} = \{0.2,\allowbreak  0.4,\allowbreak
0.8,\allowbreak  1.6,\allowbreak  3.2\}$, and $\alpha_{\textit{sun}} =
\allowbreak \{2(1+3\eps_{\textit{sun}})\}$.
We also tested $\alpha_{\textit{sun}} = \{\allowbreak
1.1, 2, 3.2\}$, as done in Sun et al.'s work~\cite{SCS20}.
When $\alpha = 1.1$, the theoretical efficiency
bounds by Sun et al.~\cite{SCS20} no
longer hold, but they yield better estimates empirically.
We compare this heuristic setting to a similar one in
our algorithms, where we replace $(2 + 3 /
\lambda)$ with $1.1$ in our code
(where our efficiency bounds no longer
hold) for $\delta = \{0.4, 0.8, 1.6, 3.2\}$. We refer to these as the \emph{heuristic parameters}.

\begin{figure*}[!t]
\centering
\includegraphics[width=\textwidth]{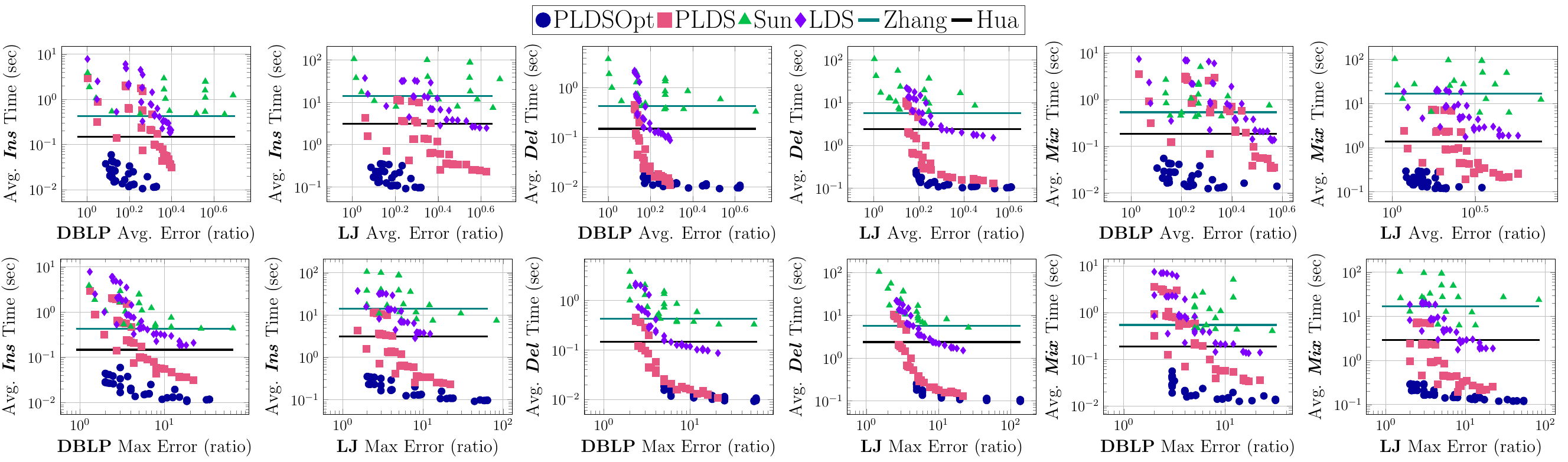}
\caption{
Comparison of the average per-batch time versus the average (top row) and maximum (bottom row)
per-vertex core estimate error ratio of \pldsopt, \plds, Sun, and \lds, using
varying parameters, on the \dblp and \lj graphs, with
batch sizes $10^5$ and $10^6$, respectively. Experiments were run for \expins,
\expdel, and \mix. The data
uses theoretically-efficient parameters as well as the heuristic parameters
where $(2 + 3 / \lambda) = \alpha_{\textit{sun}} = 1.1$.
Runtimes for Hua and Zhang are shown as horizontal lines.
}
\label{fig:exp1}
\end{figure*}

\cref{fig:exp1} shows that, using theoretically-efficient parameters, our
PLDSOpt, PLDS, and LDS implementations are faster than Sun, Zhang, and Hua,
for parameters that give
similar average and maximum per-vertex core estimate error ratios.
Furthermore, besides PLDS,
PLDSOpt \emph{outperforms all other algorithms},
regardless of approximation factor and error.
This set of experiments demonstrates the flexibility of our algorithm;
one can achieve smaller error at the cost of slightly increased
runtime. However, as the experiments demonstrate, PLDSOpt still
outperforms all other algorithms even when the parameters are tuned to
give small error; this performance gain is maintained for \expins,
\expdel, and \mix.
Greater speedups are achieved on \lj
compared to
\emph{dblp}. Such a result is expected since larger batches allow for
greater parallelism.

Concretely, compared with Zhang, PLDSOpt achieves $7.19$--$147.59\times$, $19.70$--$58.41\times$, and $9.75$--$142.79\times$
speedups on \expins, \expdel, and \mix batches, respectively.
Compared with Hua, PLDSOpt achieves $2.49$--$33.95\times$,
$6.81$--$24.51\times$, and $2.94$--$21.77\times$ speedups.
Against PLDS, PLDSOpt obtains $2.98$--$47.8\times$, 
$1.03$--$25.58\times$, and $1.5$--$76.94\times$ 
speedups for \expins, \expdel, and \mix, respectively, on parameters that give similar
approximations.
Compared with Sun, on parameters that give similar
theoretical guarantees and smaller empirical average
error, PLDSOpt achieves
$21.34$--$544.22\times$, $25.49$--$128.65\times$,
and $19.04$--$248.36\times$
speedups for \expins, \expdel, and \mix, respectively.
Neither Zhang nor Hua guarantee polylogarithmic work.
The peeling-based algorithm of \sun can have large depth and they do not provide
a concrete bound on their amortized work for their faster, round-indexing implementation.
Thus, the speedups we obtain over the benchmarks
are due to the greater theoretical efficiency 
and because our algorithms are parallel.

Finally, PLDSOpt achieves average error in the ranges $1.26$--$2.13$,
$1.47$--$4.20$, and $1.28$--$2.33$ for \expins, \expdel, and \mix, respectively.
PLDS gives comparable average errors in the ranges $1.27$--$4.22$, $1.33$--$3.39$, 
and $1.63$--$5.73$, for \expins, \expdel, and \mix, respectively, 
while running slower than \pldsopt for all parameters,
despite the guarantee that the maximum error of PLDS is bounded by $(1+\delta)(2 + 3/\lambda)$
(\cref{lem:decomp-bound}). Thus, our optimized version allows us to obtain good error
bounds empirically while drastically improving performance.\\

\begin{figure*}[hbt!]
\begin{center}
    \centerline{\includegraphics[width=0.98\textwidth]{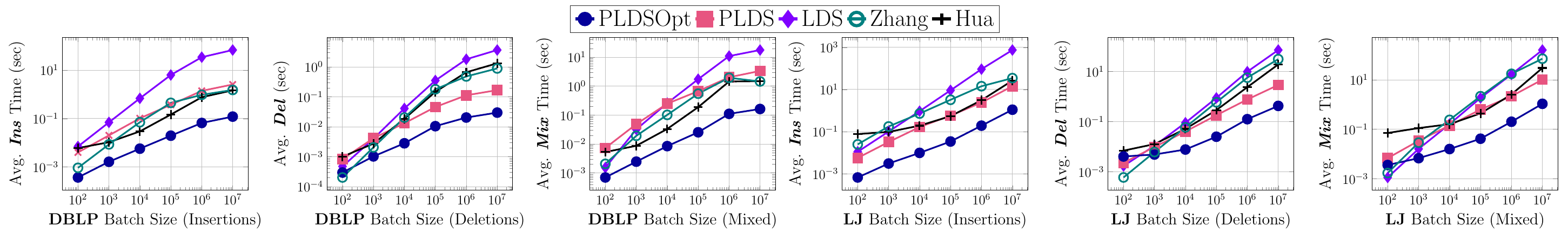}}
\caption{
Average \expins, \expdel, and \mix
per-batch running times on varying batch sizes for
\pldsopt, \plds, \lds, \zhang, and \hua on \dblp 
and \lj.
  }
\label{fig:exp2}
\end{center}
\vspace{-1em}
\end{figure*}

For \emph{all of the remaining experiments}, set 
$\delta = 0.4$ and $\lambda = 3$.

\subsection{Batch Size vs.\ Running Time} 
\cref{fig:exp2} shows the average
per-batch running times for \expins, \expdel, and \mix
on varying batch sizes for \pldsopt, \plds, \hua, \lds, and \zhang
on \dblp and \lj. We do not run this experiment on Sun
since their implementation does not have batching.
Our experiments show that \pldsopt is faster for all batch sizes except for the
smallest \expdel and \mix batches.  

Against PLDS, \pldsopt achieves a speedup over all batches from
$10.85$--$21.25\times$, $2.81$--$5.65\times$, and $10.42$--$29.28\times$ for \expins, \expdel, and \mix, respectively,
on \dblp and $8.47$--$16.9\times$, $1.99$--$7.18\times$, and $1.9$--$15.26\times$  for
\expins, \expdel, and \mix, respectively,
on \lj 
for all but the batch of size $100$ for \expdel. On the batch size of $100$ , \plds performs better than \pldsopt by a $1.79\times$ factor. 
Compared with Hua, \pldsopt achieves speedups over all batches
from $5.17$--$16.43\times$, $3.39$--$44.58\times$, and
$2.53$--$13.05\times$ for \expins, \expdel, and \mix, respectively,
on \dblp and $15.97$--$114.52\times$, $1.71$--$45.01\times$,
and $9.10$--$19.82\times$ for \expins, \expdel, and \mix, respectively, on \lj. Compared with Zhang, \pldsopt achieves speedups of
$2.49$--$22.74\times$, $2.00$--$29.92\times$, and $2.95$--$21.57\times$ for \expins, \expdel, and \mix, respectively,
on \dblp, and $31.53$--$95.33\times$, $1.25$--$73.19\times$ and $4.26$--$87.05\times$ for \expins, \expdel, and \mix, respectively, on \lj
on all but the smallest batches for \expdel and \mix.
For \expdel with a batch size of $100$, \zhang is the fastest with speedups of
$1.46\times$ and $6.86\times$ over \pldsopt on \dblp and \lj, respectively.
For \mix with batch size $100$, \lds is the fastest with speedups of
$3.19\times$ over \pldsopt on \lj.
For small batch sizes,
sequential algorithms perform better than parallel algorithms since
the runtimes of parallel algorithms are dominated by parallel overheads.

\begin{figure*}[!t]
\begin{center}
\centerline{\includegraphics[width=0.98\textwidth]{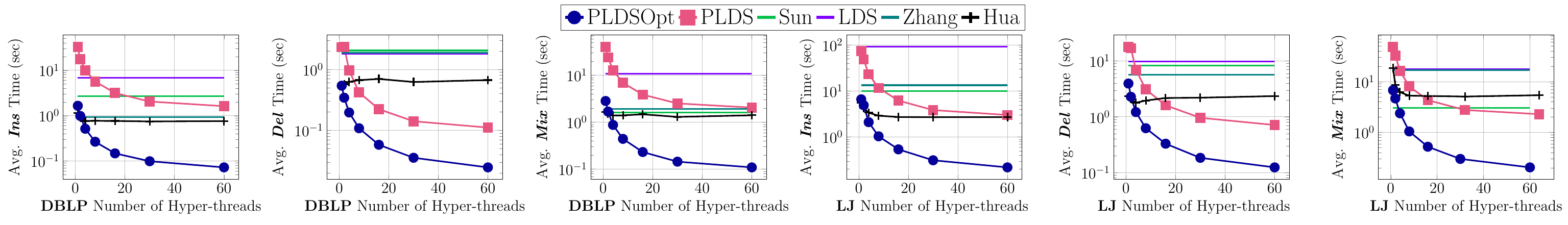}}
\caption{Parallel speedup of PLDSOpt, PLDS, and Hua, with respect to their
    single-threaded running times on \textit{dblp} and
    \textit{livejournal} 
    on \expins, \expdel, and \mix batches of size $10^6$ for all algorithms.
    The ``60'' on the $x$-axis indicates 30 cores with hyper-threading. LDS, Sun,
    and Zhang are shown as horizontal lines since they are sequential. }
\label{fig:exp3}
\end{center}
\vspace{-1em}
\end{figure*}

\subsection{Thread Count vs.\ Running Time}\label{sec:threads}
\cref{fig:exp3}
shows the scalability of PLDSOpt, PLDS, and Hua with respect to
their single-thread running times on \textit{dblp} and \textit{livejournal} using a batch size of $10^6$.
\lds, \sun, and \zhang are represented as horizontal lines since they are
sequential.
For \expins, \expdel, and \mix batches,
PLDSOpt and PLDS achieve up to $30.28\times$, $32.02\times$, and $33.02\times$,
and $26.46\times$, $25.33\times$, and $21.15\times$,
self-relative speedup, respectively.
Hua achieves up to a $3.6\times$ self-relative speedup.
We see that our PLDS algorithms achieve
greater self-relative speedups than Hua. Also,
with just $4$ threads
(available on a standard laptop),
\pldsopt already outperforms all other
algorithms.
\hua's algorithm performs DFS/BFS, which could lead to linear depth, potentially explaining
the bottleneck to their scalability with more
cores.

Gabert \etal~\cite{Gabert21} present a parallel batch-dynamic $k$-core decomposition
algorithm but their code is proprietary. However, their algorithm
appears slower and less scalable based on their paper's stated results. For
example, their
algorithm on $10^5$ edges using $32$ threads for the \lj graph requires $4$
seconds, while our algorithm on a batch of $10^6$ edges using $30$ threads (more edges and fewer threads) requires a \emph{maximum} of $0.35$ seconds.
Also, they appear to exhibit a maximum of $8\times$ self-relative 
speedup on \lj while we exhibit $21.2\times$ self-relative speedup on \lj.

\begin{figure}[htpb]
\begin{center}
\centering
\includegraphics[width=1.01\columnwidth]{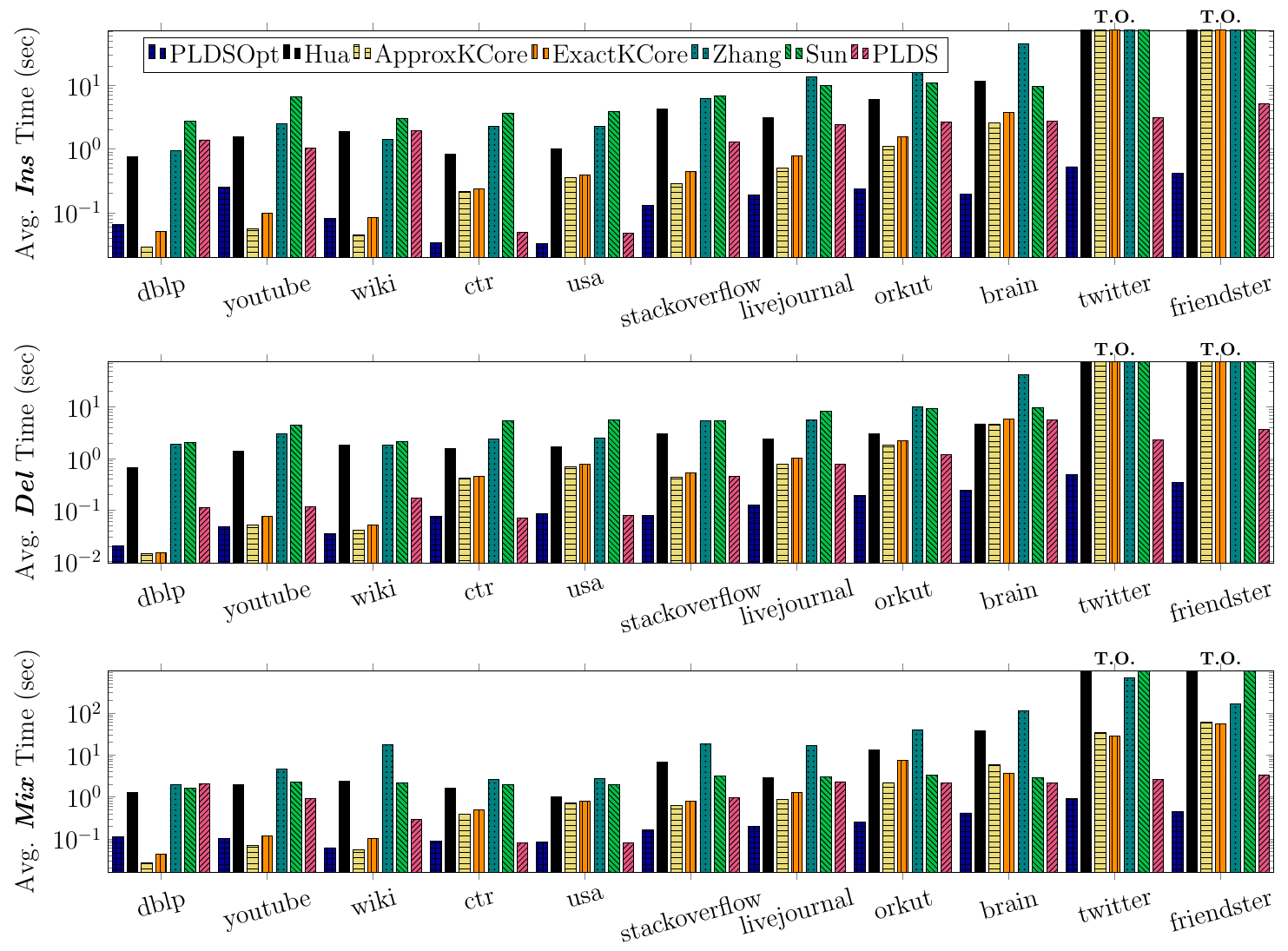}
    \caption{Average per-batch running times for PLDSOpt, Hua,
        PLDS, \sun, \zhang, ApproxKCore, and ExactKCore, on
        \textit{dblp},
        \textit{youtube},
        \textit{wiki},
        \textit{ctr},
        \textit{usa},
        \textit{stackoverflow},
        \textit{livejournal},
        \textit{orkut},
        \textit{brain},
        \textit{twitter}, and
\textit{friendster} with batches of size $10^6$ (and approximation
settings $\delta = 0.4$ and $\lambda= 3$
for PLDSOpt and PLDS).
All benchmarks (except PLDSOpt and PLDS)
timed out (T.O.) at 3 hours for \textit{twitter} and
\textit{friendster} for \expins and \expdel.
\hua and \sun timed out on \tweet and \frie for \mix.
The top graph shows insertion-only, middle graph shows deletion-only, and
bottom graph shows mixed batch runtimes. %
}\label{fig:exp4}
\end{center}
\end{figure}
\subsection{Results on Large Graphs} %
\cref{fig:exp4}
shows the runtimes of PLDSOpt, PLDS, \hua, \sun, and \zhang compared
with the static algorithms \ekcore and \akcore on additional graphs,
using \expins, \expdel, and \mix batches, all of size $10^6$.
\ekcore and \akcore are run from scratch over the entire graph after every batch
since they do not handle batch updates.
\pldsopt and \plds finished for all graphs and experiments while all
other algorithms timed out on \expins and \expdel batches for \tweet
and \frie.  \zhang
was able to finish on \mix
because their indexing algorithm
(used to create their data
structures provided the initial graph without the mixed batch) was
able to finish; since only one mixed batch is used to update the
graph, the sum of the time needed for indexing plus the update time of
one batch fell under the timeout. The same is true for \ekcore and \akcore.
However, these algorithms were
not able to finish for \expins and \expdel because the sum of the
update times across all batches is too high.

\pldsopt is faster than all other dynamic algorithms on all types of batches,
except for \plds on \ctr and \usa.
We report concrete speedups for experiments which finished within the
timeout. For \expins, it gets $10.01$--$229.71\times$ speedups over \zhang,
$6.20$--$58.66\times$ speedups over \hua, $26.02$--$119.77\times$
speedups over \sun, and $1.45$--$23.89\times$ speedups over \plds. For \expdel, it gets $30$--$176.48\times$ speedups over \zhang,
$15.79$--$52.36\times$ speedups over \hua, $41.02$--$100.34\times$ speedups
over \sun, and $2.51$--$23.45\times$ speedups over \plds
(except on \ctr and \usa).
For \mix, it gets $17.54$--$723.72\times$ speedups over \zhang,
$11.34$--$91.95\times$ over \hua, $6.95$--$35.59\times$ speedups over \sun,
and $2.81$--$18.68\times$ speedups over \plds (except on \ctr and \usa).
These massive speedups over previous work demonstrate the utility of
\pldsopt not only on large graphs but also on smaller graphs.
Notably, our \pldsopt and \plds algorithms perform not
only well on dense networks but also on very sparse road networks. For
\ctr and \usa, \plds performs better than \pldsopt, achieving up to a
$1.09\times$ speedup on \expdel and $1.12\times$ speedup on \mix.

Compared to the static algorithms, \pldsopt achieves speedups for all
but the smallest graphs, \dblp, \wiki, and \ytb. For these graphs,
the batch of size $10^6$ accounts for more than $1/3$ of the edges, and so even if the static algorithm reprocesses the entire graph
per batch, it does not process many more edges past
the batch size.
Thus, it is expected that the parallel static algorithms
perform better on small graphs and large batches. For all but the smallest
graphs, \pldsopt obtains $2.22$--$13.09\times$, $5.56$--$19.64\times$,
and $4.4$--$121.76\times$ speedups over the \emph{fastest} static algorithm
for each graph for \expins, \expdel, and \mix, respectively. \ekcore and \akcore
both timeout for \expins and \expdel on \tweet and \frie;
otherwise, we expect to see the
large improvements that we see for \mix on these experiments.

\subsection{Accuracy of Approximation Algorithms} We also computed
the average and maximum errors of all of our approximation algorithms
for our experiments shown in~\cref{fig:exp4}.
According to our
theoretical proofs, the maximum error (for PLDS) should be $(2+3/3) (1+0.4) = 4.2$.
We confirm that the maximum empirical error for PLDS falls under this 
constraint.
PLDSOpt achieves an average error of $1.24$--$2.37$  compared to errors of
$1.26$--$3.48$ for PLDS,
$1.01$--$4.17$ for ApproxKCore, and $1.03$--$3.23$ for \sun.
PLDSOpt gets a maximum error of
$3$--$6$ compared to $2$--$4.19$ for PLDS,
$3$--$5$ for ApproxKCore, and $3$--$5.99$ for \sun. We conclude that
our error bounds match those of the current best-known algorithms and
are sufficiently small to be of use for many applications.

\begin{figure*}[htb!]
\centering
\includegraphics[width=\textwidth]{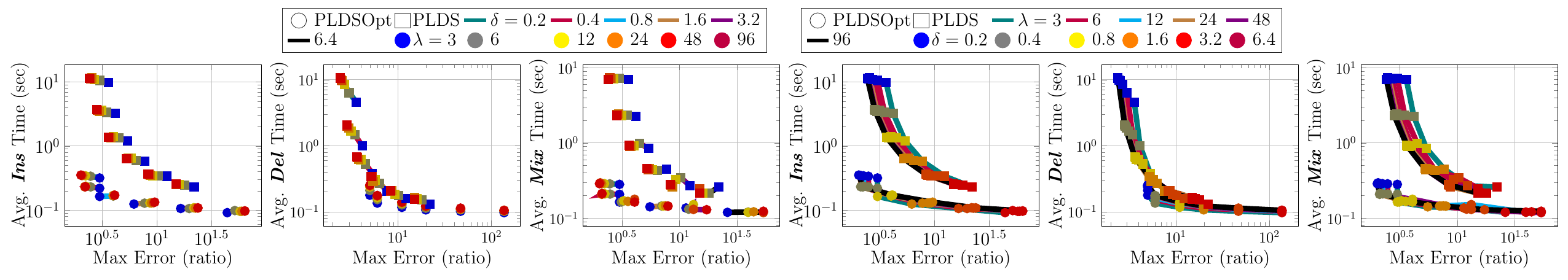}
\caption{
    Sensitivity analysis of PLDSOpt and PLDS on \textit{livejournal}.
    The first three plots fix $\delta$;
    each line is a fixed $\delta$ value
    and each point is a different $\lambda$
    value. The last three plots fix $\lambda$ and vary $\delta$.
}
\label{fig:sensitivity}
\end{figure*}

\subsection{Sensitivity of PLDS and PLDSOpt to $\delta$ and $\lambda$}
In~\cref{fig:sensitivity}, we provide a sensitivity
analysis for the parameters
$\delta$ and $\lambda$ on the \emph{maximum} error
of our PLDS and PLDSOpt algorithms since our theoretical guarantees are for the
\emph{maximum error} and as we showed in
Section~\ref{sec:expaccuracy}, the average error does not vary
significantly for our chosen set of parameters. The first three graphs
of~\cref{fig:sensitivity} shows the effect of fixing $\delta$ while
varying $\lambda$ and the last three show the opposite.

We see that for both \pldsopt and \plds, different $\lambda$ values do not
affect either the error by much (each line is essentially a cluster
of points). This matches what we expect
theoretically. 
Recall our bound on error, $(1+\delta)(2+3/\lambda)$; suppose
we set $\delta = 0.4$ and $\lambda = 3$ as in our experiments. This leads to an
upper bound of $4.2$. If we increase $\lambda$ to $6$, this only
decreases the error to $3.5$. On the other hand, if $\delta$ is increased
to $0.8$, then the error increases to $5.4$, resulting in greater sensitivity to
$\delta$. 

However, \emph{increasing} $\delta$ leads to a drastic decrease in
running time (each line is a decreasing curve) at the expense of a
large increase in error. Again, this matches what we expect
theoretically, since $\delta$ affects the number of levels in \plds and
\pldsopt (recall that in our algorithm, the number of levels per group is
$\ceil{\log_{(1+\delta)}(m)}$). A larger number of levels leads to
larger running time and we see this in our results. We do not see as
large an increase for \pldsopt since we divide the number of levels by
$50$.  This means
that for \lj the number of levels per group is
$\ceil{\log_{(1+\delta)}(42851237)/50} = 1$ for all $\delta < 0.42$.
We see this in our experiments as the curves for \pldsopt are flat for
$\delta \in [0.8, 6.4]$.

For the rest of the experiments, we fix $\delta = 0.4$ and $\lambda = 3$ based
on our sensitivity analysis;
these parameters offer a reasonable tradeoff between approximation error and
speed, as shown in \cref{fig:exp1} and \cref{fig:sensitivity}. For \sun,
we choose the parameters $\eps = \lambda = 2$ and $\alpha = 2$ since we observe
these parameters give similar approximation errors to the parameters that we chose
for our algorithms.

\subsection{Space Usage}
For each program, we implemented functions that measured the space usage of the
data structures used in the algorithms (specifically, the private and public
variables maintained in their data structure classes); for all of the algorithms,
we do not count ephemeral space usage needed by auxiliary structures that are not
maintained as either private or public variables of their data structure
class. For this set of experiments, we only test on \expins and \expdel since
maximum space is used when the entire graph is present in memory.

\cref{fig:exp10} shows the results of our space-bound experiments.
Although \plds uses more memory than most other implementations,
our PLDSOpt uses less memory than \hua and \zhang
in most settings (up to $1.34\times$ factor less memory than the minimum
space used by either) for
\textit{dblp} and up to 1.08x additional space in a few cases;
for \textit{livejournal}, it uses up to 1.72x additional space compared
to the minimum space used by \hua and \zhang. \sun uses more space
than \pldsopt for most cases; although for a few parameters for deletions in
\dblp, it uses up to $1.9\times$ less space. Since we have a $O(\log^2 n)$ 
factor in our space usage bound, we expect a slight increase in our space usage
compared to algorithms with linear space; however, as we
demonstrated, empirically our space usage is not much greater, and we
believe that this small extra space usage is a small price to pay for
the large improvement in performance obtained by our algorithms.
We provide theoretical space-efficient implementations of our PLDS which may also prove
to be more space-efficient in practice.

\begin{figure*}[htb!]
\begin{center}
    \centerline{\includegraphics[width=\textwidth]{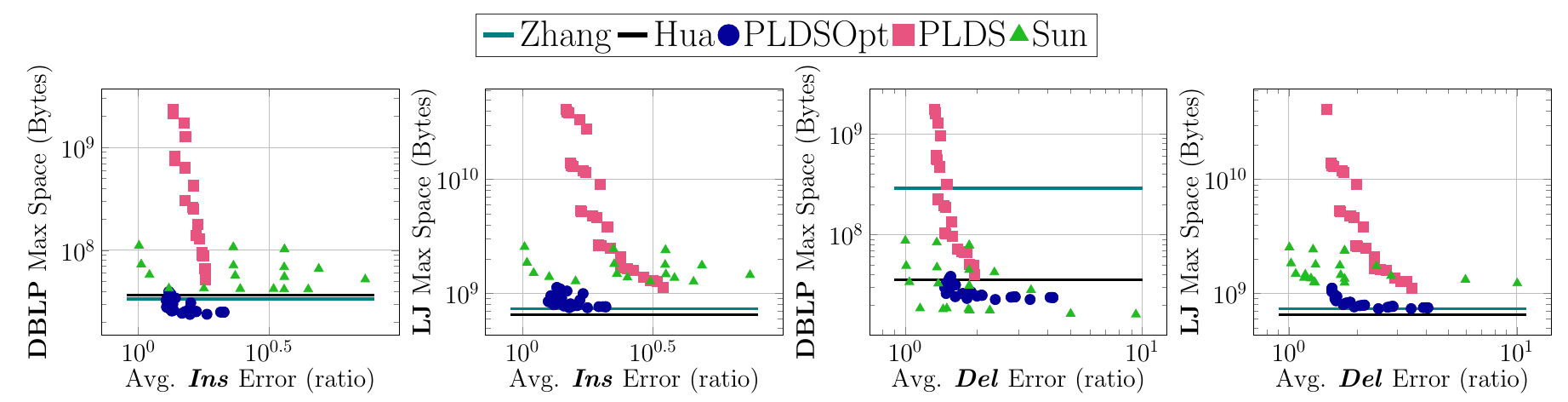}}
    \caption{Maximum space usage in bytes for PLDSOpt, Hua, Zhang,
        PLDS, and Sun in terms of the average error.
        We varied $\delta$ and $\lambda$
        and computed the error ratio and space usage
        for the programs on \textit{dblp} and \textit{livejournal}.
        We tested against  \expins and \expdel batches of size $10^5$
        for \textit{dblp} and  batches of size $10^6$
        for \textit{livejournal}.}\label{fig:exp10}
\end{center}
\end{figure*}
\section{Static $(2+\eps)$-Approximate 
$k$-Core}\label{sec:static-kcore}
Due to the $\mathsf{P}$-completeness of \kc decomposition for $k\geq
3$~\cite{anderson84pcomplete}, all known static exact \kc algorithms
do not achieve polylogarithmic depth.
We introduce a linear work and polylogarithmic depth
$(2+\eps')$-approximate \kc decomposition algorithm (with only one-sided error)
based on the parallel bucketing-based peeling algorithm for static \emph{exact} \kc
decomposition of Dhulipala et al.~\cite{dhulipala2017julienne}. The
algorithm maintains a mapping $M$ from $v \in V$ to a set of
\emph{buckets}, with the bucket for a vertex $M(v)$ changing over the
course of the algorithm. The algorithm starts at
$k=0$,  peels all vertices with degree at most $(2+\eps)(1+\eps)^k$ where
$\eps$ is set to $\frac{\sqrt{4\eps' + 9} - 3}{2}$, %
increments $k$, and repeats until the graph
becomes empty. The approximate core value of $v$ 
is $(1+\eps)^{k-1}$ where we use the value of $k$ when $v$ is peeled.
We observe that the dynamic algorithm in this paper can be
combined with a peeling algorithm like the above to yield a
linear-work approximate \kc{} algorithm with
polylogarithmic depth.

\newcommand{\codevar}[1]{\ensuremath{\mathit{#1}}}
\newcommand{\eqs}{\ensuremath{\leftarrow}}
\begin{algorithm}[!t]
\caption{Static Approximate \kc{} Decomposition} \label{alg:static}
\small
\begin{algorithmic}[1]
\Require{An undirected graph $G(V,E)$.}
\Ensure{An array of $(2+\eps')$-approximate coreness values for any constant $\eps' > 0$.}
  \State $\forall v \in V$, let $C[v] = |N(v)|$.\label{kcore:init}
  \State $\codevar{finished} \eqs 0, t \eqs 0, \eps \eqs \frac{\sqrt{4\eps' + 9} - 3}{2}, \delta \eqs \frac{2}{\eps}$.\label{kcore:init2}
  \State Let $M$ be a bucketing structure formed by initially
  assigning each $v \in V$ to the $\left \lceil \log_{1+\eps} C[v] \right \rceil$'th bucket.\label{kcore:init_bucket}
  \While { $(\codevar{finished} < |V|)$ }\label{kcore:peelstart}
    \State $(I,\codevar{bkt}) \eqs $ Vertex IDs and bucket ID of next
    (peeled) bucket in $M$.\label{kcore:nextbucket}
    \State $t \leftarrow bkt$.\label{kcore:checkinner}
    \For{iteration $j \in [\ceil{\log_{1+\delta}(n)}]$}

    \State $R \eqs$ $\{(v,r_v)\ | \ v \in N(I)$, $r_v = |\{(u,v) \in E\ |\ u \in I\}|\}$.\label{kcore:r}
    \State $U \eqs$ Array of length $|R|$.

    \ParFor {$R[i] = (v,r_v)$, $i \in [0, |R|)$}\label{kcore:computebktstart}
      \State $\codevar{inducedDeg} = C[v] - r_v$
      \State $C[v] = \max(\codevar{inducedDeg}, \left \lceil (1+\eps)^{t-1} \right \rceil)$
      \State $\codevar{newbkt} = \max(\left \lceil \log_{1+\eps} C[v] \right \rceil, t)$
      \State $U[i] = (v, \codevar{newbkt})$\label{kcore:computebktend}
    \EndParFor

    \State Update $M$ for each $(u, \codevar{newbkt})$ in $U$. \label{kcore:updatebkt}
    \State $next$-$bkt \leftarrow $ bucket ID of the next smallest bucket in $M$. 
    \If{$(1+\eps)^{next\text{-}bkt} \leq (2+\eps)(1+\eps)^t$}
        \State $(I, next\text{-}bkt) \eqs $ Vertex IDs of the next (peeled) bucket in $M$.\label{kcore:nextbucket}
    \Else
        \State \textbf{break}
    \EndIf
    \EndFor
  \EndWhile\label{kcore:peelend}
  \State \Return $C$.
\end{algorithmic}
\end{algorithm}

Algorithm~\ref{alg:static} shows pseudocode for our
approximate \kc{} algorithm, which computes an approximate coreness
value for each vertex. The algorithm sets the initial coreness
estimates, $C[v]$, of each vertex to its degree
(Line~\ref{kcore:init}).
Then, it maintains a parallel bucketing data structure $M$, which maps each vertex to the $\lceil \log_{1+\eps} C[v]
\rceil$'th bucket (\cref{kcore:init_bucket}). It initializes a variable $\codevar{finished}=0$ to keep track of the number of vertices peeled and a variable $t=0$ used to compute the approximate core values (\cref{kcore:init2}).
The rest of the algorithm performs peeling, where the peeling
thresholds are powers of $(1+\eps)$.
The peeling loop (\cref{kcore:peelstart}--\cref{kcore:peelend})
first extracts the lowest non-empty bucket from $M$
(\cref{kcore:nextbucket}), which consists of $I$, a set of vertex IDs of
vertices that are being peeled, and the bucket number \codevar{bkt}.
If more than $\log_{1+\delta}(n)$ rounds of peeling have occurred at
the threshold $(2+\eps)(1+\eps)^{t}$ (where we set $\delta = \frac{2}{\eps}$), the algorithm increments $t$
(\cref{kcore:checkinner}). Next,
the algorithm computes in parallel an array $R$ of pairs $(v, r_v)$, where $v$ is a
neighbor of some vertex in $I$ and $r_v$ is the number of neighbors
of $v$ in $I$ (\cref{kcore:r}). Finally, the algorithm computes in parallel the new buckets for
the affected neighbors $v$ (\cref{kcore:computebktstart}--\cref{kcore:computebktend}).
The coreness estimate is updated to the maximum of the peeling
threshold of the previous level and the current induced degree of $v$
after $r_v$ of its neighbors are removed. Finally, the algorithm updates the
buckets using the new coreness estimates for the updated vertices
(\cref{kcore:updatebkt}), which can be done in parallel using our bucketing data structure.

We provide an example of this algorithm below.

\begin{figure}[t]
    \centering
    \includegraphics[width=0.9\columnwidth]{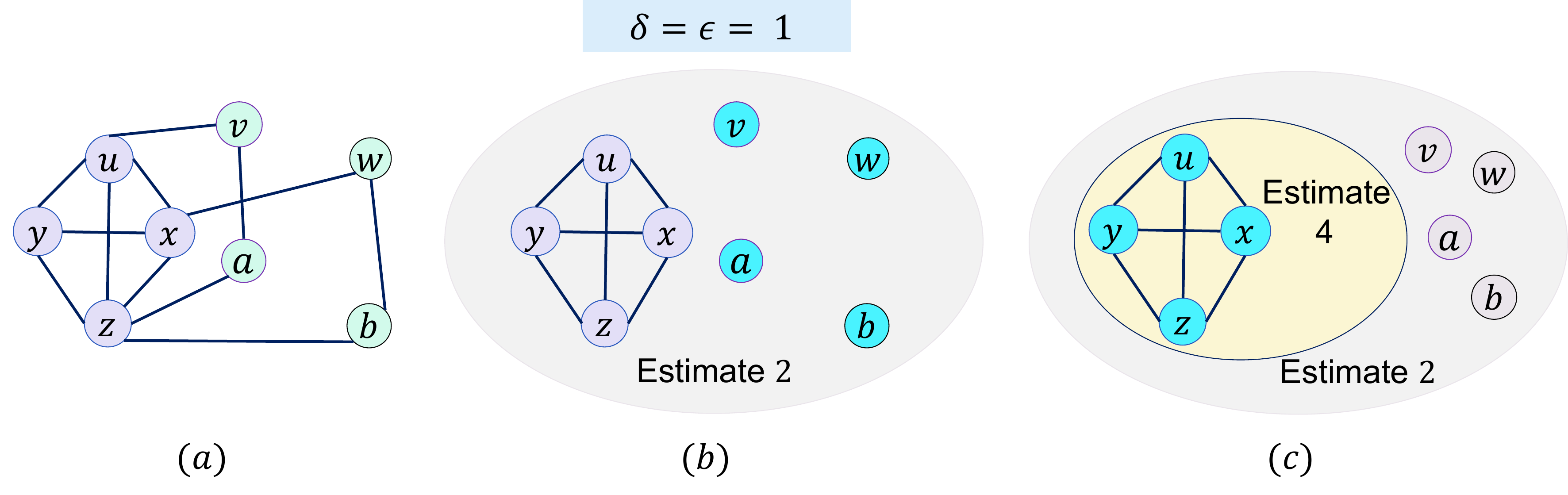}
    \caption{Example of a run of~\cref{alg:static} described
    in~\cref{exp:static}.}\label{fig:static}
\end{figure}

\begin{example}\label{exp:static}
    \cref{fig:static} shows a run of~\cref{alg:static} on an example
    graph.
    Given the parameters $\eps = \delta = 1$, the two buckets that the vertices
    of the input graph (shown in $(a)$) are partitioned
    into are bucket index $1$
    (green vertices)
    and bucket index $2$ (purple vertices). Vertices $v$, $w$, $a$, and $b$ have degree
    $2$ so they are put into the bucket with index $\ceil{\log_2(2)} = 1$. Since
    $u$, $x$, $y$, and $z$ have degree $\geq 3$, they are put into the bucket with index
    $\ceil{\log_2(3)} = 2$.

    Since the bucket with index $1$ has the smaller
    bucket index, we peel off all the vertices in that bucket (the green
    vertices) and we assign the core estimate of $(1+\eps)^1 = 2$
    to all vertices in that bucket (shown in $(b)$). We update
    the buckets of all neighbors of the peeled vertices; however, since
    $u$, $x$, $y$, and $z$ all still have degree $\geq 3$, they remain in the bucket with index
    $2$. Finally, we peel bucket index $2$ and assign all vertices
    in that bucket an estimate of $(1+\eps)^2 = 4$
    (shown in $(c)$). In this example, the estimates produced are
    $3$-approximations of the real coreness values.
\end{example}

\ificml
In the full paper, we prove
\fi
\iflong
    We prove below
\fi
that \cref{alg:static} finds an $(2+\eps)$-approximate \kc
decomposition in $O(m)$ expected work and $O(\log^3 m)$ depth \whp, using $O(m)$ space, as stated in
\cref{thm:static}.
We give the approximation guarantees of 
our algorithm using lemmas from~\cite{Ghaffari2019}, 
and use an efficient parallel semisort implementation~\cite{gu15semisort} for our work bounds.

\iflong
\begin{theorem}\label{thm:static-app}
    For a graph with $m$ edges,\footnote{Our bounds in this paper assume
  $m =\Omega(n)$ for simplicity, although our algorithms work even if
  $m=o(n)$.} for any constant $\eps > 0$, there is an
  algorithm that finds an $(2+\eps)$-approximate \kc
  decomposition in $O(m)$ expected work and $O(\log^3 m)$ depth
  with high probability, using $O(m)$ space.
\end{theorem}
\begin{proof}
    Our approximation guarantee is given by Observation 4 of~\cite{Ghaffari2019}. 
    Using Observation 4, the number of vertices with core number $(1+\eps)^t$
    after one round of peeling in~\cref{alg:static} shrinks by a factor of $\frac{2}{2+\eps}$. 
    Let $V_{\leq (1+\eps)^t}$ be the number of vertices with core number
    at most $(1+\eps)^t$. After removing all vertices with degree at most $(2+\eps)(1+\eps)^t$,
    the number of vertices with core number $(1+\eps)^t$ and with degree greater than 
    $(2+\eps)(1+\eps)^t$ is at most $\frac{2(1+\eps)^t |V_{\leq (1+\eps)^t}|}{(2+\eps)(1+\eps)^t} = 
    \frac{2|V_{\leq (1+\eps)^t}|}{2+\eps}$. Since $|V_{\leq (1+\eps)^t}| \leq n$, the maximum number of rounds
    needed to peel all vertices with core number at most $(1+\eps)^t$ is $\log_{(2+\eps)/2}(n)$.
    By induction on $t$ (\cref{kcore:checkinner}), after $\log_{(2+\eps)/2}(n)$ rounds,
    all vertices with core number at most $(1+\eps)^t$ are removed. Hence, in round
    $t+1$, all vertices have core number greater than $(1+\eps)^t$ and have core 
    number at most $(2+\eps)(1+\eps)^{t+1}$; hence, we obtain a $(2+\eps)(1+\eps) = 
    2+\eps'$ approximation (for any constant $\eps' > 0$ and appropriate setting 
    of $\eps$) when we give coreness approximations of $(1+\eps)^t$ to all vertices 
    peeled for $t + 1$. 
    
    Our algorithm uses a number of data structures that we use to obtain our work,
    depth, and space bounds. Our parallel bucketing data structure
    (\cref{kcore:init_bucket}) can be maintained via a sparse set
    (hash map), or by using the bucketing data structure
    from~\cite{dhulipala2017julienne}.
    The outer loop iterates for $O(\log n)$ times (\cref{kcore:peelstart}). Within
    each iteration of the outer loop, we iterate for
    $O(\log_{(1+\delta)}n) = O(\log n)$ rounds for constant $\delta = \frac{2}{\eps}$.
    After obtaining a set of vertices, we update the buckets using
    semisort in $O(\log n)$ depth \whp{}~\cite{dhulipala2017julienne}.
    Thus the overall depth of the algorithm is $O(\log^3 m)$ for any
    constant $\delta > 0$.

    The work of the algorithm can be bounded as follows. We charge the
    work for moving a vertex from its current bucket to a lower bucket
    within a given round to one of the edges that was peeled from the vertex
    in the round. Thus the total number of bucket moves done by the
    algorithm is $O(m)$. Each round of the algorithm also peels a
    number of edges and aggregates, for each vertex that has a
    neighbor in the current bucket, the number of edges incident to
    this vertex that are peeled (the $r_v$ variable in the algorithm).
    We implement this step using a randomized
    semisort~\cite{gu15semisort}. Since $2m$ edges are peeled in
    total, the overall work is $O(m)$ in expectation.

    Lastly, we bound the space used by the algorithm. There are a
    total of $O(\log_{1+\eps} n) = O(\log n)$ buckets for any
    constant $\eps > 0$. Each vertex appears in exactly one bucket,
    and thus the overall space of the bucketing structure is $O(n)$.
    The algorithm also semisorts the edges peeled from the graph in
    each step. Since all $m$ edges could be peeled and removed within
    a single step, and thus semisorted the overall space used by the
    algorithm is $O(m)$.
\end{proof}
\fi

The approximation guarantees provided by our algorithm are essentially
the best possible, under widely believed conjectures.  Specifically,
Anderson and Mayr~\cite{anderson84pcomplete} show that the
optimization version of the High-Degree Subgraph problem, namely to
compute the largest core number, or \emph{degeneracy} of a graph
cannot be done better than a factor of $2$. Thus, obtaining a
polynomial work and polylogarithmic depth $(2-\eps)$-approximation
to the coreness value of each vertex would yield a $(2-\eps)$-approximation to the optimization version of the High-Degree Subgraph
problem, and show that $\mathsf{P}=\mathsf{NC}$, contradicting a
widely-believed conjecture in parallel complexity theory.

In recent years, several results have given parallel
algorithms that obtain a $(1+\eps)$-approximation to the coreness
values in distributed models of computation such as the Massively
Parallel Computation  model~\cite{Esfandiari2018,
Ghaffari2019}. These results work by performing a
\emph{random sparsification}
of the graph into
a subgraph that approximately
preserves the coreness values. They then send this subgraph to a
single machine, which runs the sequential peeling algorithm on the
subgraph to find approximate coreness values. Crucially, this second
peeling step on a single machine can have $\Theta(n)$ depth, and thus,
this approach does not yield a polylogarithmic depth algorithm in the
work-depth model of computation.

\section{Framework for Batch-Dynamic Graph Algorithms from Low Out-Degree Orientations}\label{sec:framework}

In this section, we introduce a framework that we will use in all of our batch-dynamic
algorithms that use our batch-dynamic low out-degree orientation algorithm (\cref{sec:compute-core}).
Our framework assumes three different methods for each of the problems (maximal matching,  $k$-clique counting, and vertex coloring) that
we solve. Specifically, these three methods handle batches of insertions and deletions separately; let $\batchedgeflip$,
$\batchinsert$, and $\batchdelete$ denote these three methods. 

We assume for simplicity that all updates in the batch $\batch$ are \emph{unique}, which means that
no edge deletion occurs on an inserted edge in the batch and vice versa. Furthermore, we assume that the updates are \emph{valid}, meaning that
if an edge insertion $(u, v)$ is in $\batch$, then $(u,v)$ does not exist in the graph, and if an 
edge deletion $(w, x)$ is in $\batch$, then edge $(w, x)$ exists in the graph.
Such assumptions are only 
\emph{simplifying} assumptions because it is easy to perform preprocessing on $\batch$ in $O(\batchsize \log n)$ work and $O(\log n)$
depth to ensure that these assumptions are satisfied. 
In fact, our implementations in \cref{sec:eval}
do perform this preprocessing on the input batches.
To find all unique updates, we perform a parallel sort in $O(\batchsize\log n)$
work and $O(\log n)$ depth~\cite{JaJa92,blelloch96parallel,dhulipala2018theoretically}; 
we first sort on the edge and then
the timestamp of the update. Then, we perform a parallel filter in 
$O(\batchsize)$ work and $O(1)$ depth~\cite{JaJa92,blelloch96parallel,dhulipala2018theoretically} 
where we keep each edge with the 
latest timestamp. Then, we perform another parallel filter to keep 
only edge insertions of nonexistent edges and edge deletions 
of edges that exist in the graph. This preprocessing ensures $\batch$ follows
our simplifying assumptions and do not exceed the complexity
bounds of our \plds, and hence, we assume all input batches contain
unique and valid updates. The work and depth for preprocessing are subsumed by the bounds for the algorithms.

\begin{algorithm}[!t]\caption{GraphProblemUpdate$(G, \batch)$}
    \label{alg:framework}
    \small
    \begin{algorithmic}[1]
    \Require{A graph $G = (V, E)$ and a batch $\batch$ of unique and valid 
    updates.}
    \Ensure{A solution to the relevant graph problem.}
    \State $\mathsf{Update}(\batch)$ [\cref{alg:rebalance}]. \label{line:update-plds}
    \State $A \leftarrow \lowoutdegorient(\batch)$.\label{line:flipped} %
    \State Perform parallel filter on $\batch$ to obtain a batch of insertions, $\batchinserts$, and a batch of deletions, $\batchdeletes$.\label{line:split}
     
    \State $\batchedgeflip(A, \batchinserts, \batchdeletes)$.\label{line:flips-batch}
    \State $\batchdelete(\batchdeletes)$.\label{line:insert-batch}
    \State $\batchinsert(\batchinserts)$.\label{line:delete-batch}
    \end{algorithmic}
\end{algorithm}

\paragraph{Detailed Framework} The pseudocode for our framework is shown
in~\cref{alg:framework}. We first update the \plds by calling the update
procedure (\cref{alg:rebalance}) on the batch of updates 
in~\cref{line:update-plds}. Afterwards, we call our low out-degree orientation
algorithm to obtain the set of edges that were flipped, placed in set $A$ 
(\cref{line:flipped}). %
Then, we take the batch of updates $\batch$ and split the 
batch into a batch of insertions, $\batchinserts$, and a batch of 
deletions,
$\batchdeletes$ (\cref{line:split}). 
We call $\batchedgeflip$ (\cref{line:flips-batch}) on the set of flipped edges 
$A$, which processes the edge flips accordingly for each problem.
Finally, we call the problem specific functions $\batchdelete$
and $\batchinsert$ (\cref{line:insert-batch,line:delete-batch}) on $\batchdeletes$ and $\batchinserts$, 
respectively; we first call $\batchdelete$ and then $\batchinsert$.

\paragraph{Analysis}

By~\cref{cor:arboricity-orientation}, our low out-degree orientation algorithm
gives a $O(\alpha)$ out-degree orientation. Furthermore, the
amortized work of the algorithm
indicates that $O(\batchsize \log^2 n)$ amortized flips occur with each 
batch $\batch$. Suppose that $\batchedgeflip(A)$ takes $O(|A|\flipswork)$ work and $O(\flipsdepth)$ depth; 
$\batchinsert(\batchinserts)$ takes
$O\left(|\batchinserts|\insworkalpha\right)$ work and 
$O\left(\insdepth\right)$ depth, and $\batchdelete\left(\batchdeletes\right)$ takes 
$O\left(|\batchdeletes|\delworkalpha\right)$ work and 
$O\left(\deldepth\right)$ depth; and the update methods 
require $O(S)$ space in total.
Then, we show the following theorem about
our framework. 

\begin{theorem}\label{thm:framework}
    \cref{alg:framework} takes $$O\left(\batchsize \flipswork \log^2 n + \batchsize\insworkalpha + \batchsize\delworkalpha \right)$$
    amortized work and $$O\left(\log^2 n \log\log n + \flipsdepth + \insdepth + \deldepth\right)$$ depth \whp, in $O(n\log^2 n + m + S)$ space.
\end{theorem}

\ifspaa
\begin{proof}
\cref{thm:low-outdegree} states
that updating the \plds and getting the flipped edges require 
$O(\batchsize \log^2 n)$ amortized work, $O(\log^2 n\log\log n)$ depth, 
and $O(n\log^2 n + m)$
space. 
Since the calls to the procedures are independent and sequential, the total 
work, depth and space equal the sum of the work, depth, and space of our 
PLDS algorithm and $\batchedgeflip, \batchdelete$ and $\batchinsert$.

Then, the only additional information we need are the
sizes of $\batchinserts$ and $\batchdeletes$. By our algorithm,
$|\batchinserts|, |\batchdeletes| \leq \batchsize$ since 
$\batchinserts \cup \batchdeletes = \batch$. By~\cref{thm:low-outdegree}, $A$ has $O(\batchsize \log^2 n)$ amortized flips; thus, the amortized work of 
$\batchedgeflip$ is $O(\batchsize \flipswork\log^2 n)$. Finally,
the \plds uses $O(n\log^2 n + m)$ space;
thus, with the additional $O(S)$ space, 
the total space used is $O(n\log^2 n + m + S)$.
\end{proof}

In addition, we assume that the algorithms $\batchinsert$
and $\batchdelete$ correctly maintain the desired properties
required by each specific problem after processing $\batchinserts$
and $\batchdeletes$, respectively. Such an assumption ensures
the correctness of the solutions produced by our framework. 
We show in~\cref{sec:matching,sec:coloring,sec:clique} that
this is true for all of our procedures. Additionally, we 
can get rid of the $O(n\log^2 n)$ term in space at the expense
of an extra $O(\log^2 n)$ factor in depth by using our
space-efficient structures from~\cref{app:data-structure-impl}.

Using this framework (with the PLDS guarantees given in~\cref{thm:batch-dynamic}), 
we present batch-dynamic algorithms for a number
of problems in~\cref{sec:matching,sec:coloring,sec:clique} for 
maximal matching, $k$-clique counting, and vertex coloring.

\fi

\section{Maximal Matching}\label{sec:matching}
A maximal matching in a graph $G = (V, E)$ is a set of edges $M$ in the graph
such that no vertex is adjacent to two edges in $M$. Furthermore, no additional
edges can be added to $M$ without causing a vertex to be adjacent to two
edges in $M$.

We provide the following parallel batch-dynamic algorithm
for maximal matching using our framework given in~\cref{sec:framework}. 
We instantiate $\batchinsert$ and $\batchdelete$ for the maximal matching
problem in this section. We use the simple algorithm of Neiman and Solomon~\cite{neiman2015simple}
as a starting point, although we will see that the batch-dynamic setting introduces several non-trivial
challenges. 

\paragraph{Sequential Algorithm of Neiman and Solomon~\cite{neiman2015simple}} 
The sequential algorithm of Neiman and Solomon~\cite{neiman2015simple} uses the dynamic 
orientation algorithm of 
Brodal and Fagerberg~\cite{brodal99dynamicrepresentations}, which gives an $O(D)$ out-degree orientation for any
$D > 2\alpha_{max}$.
Given an edge insertion, they
check whether both endpoints are in the maximal matching. If not, they match the endpoints to each other.
For each vertex $u \in V$,
they maintain the set of unmatched \inn, $F(u)$, in a data structure consisting of an array augmented with a linked list.
On an edge deletion $(u, v)$ where $(u, v)$ is in the matching, they check $F(u)$
(resp.\ $F(v)$) to see if 
any \inn, $u'$ (resp.\ $v'$) are unmatched. If so, they match $u$ to $u'$ (resp.\ $v$ to $v'$). If no 
\inn are unmatched, they check whether any of their \outn, $u''$ (resp.\ $v''$) are unmatched. If so, $u$ (resp.\ $v$) matches 
with $u''$ (resp.\ $v''$). On an edge deletion $(u, v)$ (where the edge is oriented from $u$ to $v$),
if $u$ is unmatched, it removes itself from $F(v)$.
On an edge insertion, $(u, v)$, if $u$ is unmatched, it adds itself to $F(v)$, and if $u$ is matched we do not do anything. 
Finally, for an edge flip from $(u, v)$ to $(v, u)$, if $u$ is unmatched, it removes itself from $F(v)$;
if $v$ is free, it adds itself to $F(u)$. Again if $u$ is matched, we do not do anything.
Maintaining the maximal matching and updating all data structures
can be done in $O\left(\frac{\log n}{\log\left((\log n)/\alpha_{max}\right)} + \alpha_{max}\right)$ amortized time 
for $\alpha_{max} = o(\log n)$. For $\alpha_{max} = \Omega(\log n)$, they obtain $O(\alpha_{max})$ amortized time.

Unfortunately, the batch-dynamic setting introduces several challenges, the most important of which is: 
edge deletions may unmatch many different
vertices simultaneously, which need to be matched to potentially the same set of \inn. Thus, we can no longer
arbitrarily pick \inn to match unmatched vertices since many vertices may be matched to the same in-neighbor.
But we also cannot afford to look at \emph{all} of the \inn of an unmatched vertex since the in-degree is potentially
$\omega(\alpha)$.
Even for edge insertions, we cannot choose to add every edge insertion between two unmatched vertices to
the maximal matching since many edge insertions may occur on the \emph{same} unmatched vertex.

\paragraph{Batch-Dynamic Algorithm} Edge insertions are easier to handle;
for each edge insertion, in parallel, 
we check whether both endpoints adjacent
to the insertion are unmatched. If so, we run a static, parallel algorithm 
over all such vertices adjacent to an edge insertion but is unmatched; 
this finds a maximal matching among all vertices that
want to be matched due to edge insertions. If not, we do
nothing for these vertices. 

Deletions are trickier to handle. 
For each vertex incident to an edge deletion, 
we check whether it is still matched or if it can be matched with any of its 
neighbors. However, such an operation could be expensive because although a vertex has bounded
number of \outn, it may have many \inn.
To find a new matching for unmatched vertices due to edge deletions, we make use of the best-known 
low-depth, parallel, static maximal matching
algorithm which takes $O(m+n)$ work\footnote{The work of the parallel static matching algorithm given in~\cite{BFS12} can
be shown to be $O(m + n)$ \whp when using the high probability analysis of parallel bucket sort given by Bercea and Even~\cite{BerceaE21}.}
and $O(\log^2n)$ depth \whp~\cite{BFS12,FischerN18,Birn2013} combined with a scheme where we 
progressively double the number of \inn we attempt to match. Details about these procedures are provided in the next subsections.

\paragraph{Data Structure} We maintain the following data structures in our
algorithm. For each vertex $v$,
we maintain a parallel 
hash table, $\incoming_v$, of \inn which are
unmatched. Each time a vertex $v$ becomes unmatched, we inform all \outn
of $v$ that $v$ is unmatched. Similarly, when $v$ becomes matched, we inform
all \outn that it is matched.
Then, each vertex that has been
informed that $v$ has been unmatched adds $v$ to its 
hash table of
unmatched incoming neighbors, in parallel. We assume that the out-neighbors of every vertex $u$
are also maintained in a parallel hash table $\out_u$,
that is kept up to date by the edge orientation
algorithm. %
These data structures require $O(m)$ in total space usage.
Sequential versions of $\incoming_v$ and $\out_u$ are 
maintained by Neiman and Solomon~\cite{neiman2015simple}.

\subsection{Maximal Matching $\batchedgeflip$}

\begin{algorithm}[!t]\caption{$\mathsf{MaximalMatching}\batchedgeflip(A, \batchinserts, \batchdeletes)$}
    \label{alg:maximal-flips}
    \small
    \begin{algorithmic}[1]
    \Require{A set of edge flips $A$.}
    \Ensure{Updated data structures.}
    \ParFor{each flipped edge $(u, v) \in A$} \Comment{The edge is flipped from $(u, v)$ to $(v, u)$ and stored as $(u,v)$ in $A$.}\label{flip:outer-loop} 
        \If{$(u,v)$ is in the matching}
            \State Remove $u$ from $\incoming_v$.\label{flip:remove}
            \State Add $v$ to $\incoming_u$.\label{flip:add}
        \EndIf
    \EndParFor
    \end{algorithmic}
\end{algorithm}

The pseudocode for this procedure
is given in~\cref{alg:maximal-flips}.
To implement $\batchedgeflip$ for maximal matching, 
we update the data structures $\incoming_w$ to accurately account for unmatched \inn of vertices (which are 
stored in the $\incoming_v$ structures for each vertex $v$). To do this in parallel, for each 
flipped edge from $(u, v)$ to $(v, u)$ (\cref{flip:outer-loop}), we remove $u$ from $\incoming_v$ (\cref{flip:remove})
and add $v$ to $\incoming_u$ (\cref{flip:add}).

\subsection{Maximal Matching $\batchinsert$}
The pseudocode for this procedure is given in~\cref{alg:maximal-insert}.
To implement $\batchinsert$ for maximal matching, we need to check, in parallel, whether 
\emph{both} endpoints of the inserted edge
are unmatched (\cref{line:check-unmatched}). If so, we know that they can potentially be matched 
to each other. However, there could be multiple edge
insertions incident to the same unmatched vertex; thus, we cannot simply add every inserted edge 
between unmatched vertices
to the maximal matching. Instead, we keep track of all edge insertions between two unmatched 
vertices in a dynamic array $S$ (\cref{line:arr-edges-1}) %
and run a 
static, parallel maximal matching algorithm on the \emph{induced subgraph} given by $S$ 
(\cref{line:static-mm}).
We specifically use the work-efficient parallel, static maximal matching algorithm of Blelloch
\etal~\cite{BFS12} which was shown to have a better depth than previously stated
in the analysis provided by Fischer and Noever~\cite{Fischer20}. Finally, each newly matched
vertex from~\cref{line:static-mm} updates its \outn that it is now matched. For each such 
newly matched vertex $v$, each \outn $w$ of $v$ removes $v$ from $\incoming_w$.

The correctness of our procedure follows from the fact that
only new edge insertions may be added to the matching. 
Because our algorithm always maintains a maximal matching,
any previous edge that existed in the graph is either 
in the matching or is incident to a matched vertex. 
Thus, our procedure only needs to consider newly
inserted edges and such edges can be determined using
a parallel, static maximal matching algorithm~\cite{BFS12}.

\begin{algorithm}[!t]\caption{$\mathsf{MaximalMatching}\batchinsert(\batchinserts)$}
    \label{alg:maximal-insert}
    \small
    \begin{algorithmic}[1]
    \Require{A batch $\batchinserts$ of unique and valid 
    insertion updates.}
    \Ensure{A maximal matching.}
    \State $S \leftarrow \emptyset$. \Comment{$S$ contains matching candidate edges.}\label{line:arr-edges-1}
    \ParFor{each edge $\{u, v\} \in \batchinserts$}\label{line:find-unmatched-1}
        \If{$u$ \textbf{and} $v$ are unmatched}\label{line:check-unmatched}
            \State $S \leftarrow S \cup \{\{u, v\}\}$.\label{line:arr-edges-2}
        \EndIf
    \EndParFor
    \State Run $\static(G(S))$.\label{line:static-mm}
    \ParFor{each newly matched vertex $v$}\label{line:update-insert-1}
        \ParFor{each out-neighbor $w$ of $v$}\label{line:update-insert-2}
            \State Remove $v$ from $\incoming_w$.\label{line:update-insert-3}
        \EndParFor
    \EndParFor
    \end{algorithmic}
\end{algorithm}

\subsection{Maximal Matching $\batchdelete$}

\begin{algorithm}[!t]\caption{$\mathsf{MaximalMatching}\batchdelete(\batchdeletes)$}
    \label{alg:maximal-delete}
    \small
    \begin{algorithmic}[1]
    \Require{A batch $\batchdeletes$ of unique and valid 
    deletion updates.}
    \Ensure{A maximal matching.}
    \State $U \leftarrow \emptyset$. \Comment{$U$ contains newly unmatched vertices.}\label{line:unmatched-vertices}
    \State $T \leftarrow \emptyset$. \Comment{Contains the \outn of unmatched vertices.}\label{line:induced-out-1}
    \ParFor{each edge $\{u, v\} \in \batchdeletes$}\label{line:iterate-deletes}
        \If{$\{u, v\}$ is in the matching}\label{line:check-matched}
            \State $U \leftarrow U \cup \{u, v\}$.\label{mm-del:arr-edges-2}
            \State $T \leftarrow T \cup \out_u \cup \out_v$.\label{line:induced-out-2}
        \EndIf
    \EndParFor
    \State Run $\static(G(U \cup T))$.\label{line:outneighbor-match}
    \ParFor{each newly matched vertex $v$ in $G(U \cup T)$}\label{line:new-matched-outneighbor}
            \ParFor{each out-neighbor $w$ of $v$}\label{line:iterate-out}
                \State Remove $v$ from $\incoming_w$.\label{line:remove-from-out}
            \EndParFor
        \State $U \leftarrow U \setminus \{v\}$.\label{line:remove-from-u}
    \EndParFor
    \State $c\leftarrow 1$. \Comment{$c$ is the number of incoming unmatched neighbors picked to run
    the static maximal matching algorithm.}\label{line:c} \label{line:query-neighbors}
    \While{$U \neq \emptyset$}\label{unmatch:outer-while-loop}
        \ParFor{each vertex $u \in U$}\label{unmatch:inner-for-loop}
            \State Pick $c$ incoming unmatched neighbors arbitrarily.\label{line:pick-arbitrary}
            \If{$\incoming_u = \emptyset$}\label{unmatch:empty-set}
                \State $U \leftarrow U \setminus \{u\}$.
                \label{line:empty-incoming-free}
            \EndIf
        \EndParFor
        \State Let $G'$ be the induced subgraph consisting of all vertices in
        $U$ and the picked incoming unmatched neighbors.\label{unmatch:induced-subgraph}
        \State Run $\static(G')$.\label{line:ancestor-matchings}
        \ParFor{each newly matched vertex $v$ in $G'$}\label{unmatch:remove-from-out}
            \ParFor{each out-neighbor $w$ of $v$}\label{unmatch:out-neighbors}
                \State Remove $v$ from $\incoming_w$.\label{line:remove-outn}
            \EndParFor
            \State $U \leftarrow U \setminus \{v\}$.\label{line:remove-u}
        \EndParFor

        \State Set $c \leftarrow 2 \cdot c$.\label{line:double-c}
    \EndWhile
    \ParFor{each $v \in U$}
        \If{$v$ remains unmatched}
        \ParFor{each out-neighbor $w$ of $v$}
                \State Add $v$ to $\incoming_w$.
        \EndParFor
        \EndIf
    \EndParFor
    \end{algorithmic}
\end{algorithm}

The pseudocode for this algorithm is given in~\cref{alg:maximal-delete}.
For any edge $(u, v)$ that is part of the matching that has been removed by
an edge deletion, we create an induced subgraph consisting of the set of such
vertices and their out-neighbors (\cref{line:unmatched-vertices,mm-del:arr-edges-2,line:induced-out-1,line:induced-out-2}). %
Given $|\batchdeletes|$ such deletion updates, the induced
subgraph of each vertex $v$ affected by the deletions and its out-neighbors
has size $O(|\batchdeletes|\alpha)$. We use the parallel, static algorithm of Blelloch \etal~\cite{BFS12}
to find a matching in this induced graph (\cref{line:outneighbor-match}). 

For vertices that remain unmatched after the above procedure is run, we must
now attempt to match these vertices with the set of incoming unmatched neighbors. To do
this, we run the parallel, static maximal matching
algorithm on some induced
subgraphs of the remaining unmatched vertices and a subset of 
incoming vertices. 
Specifically, starting from $c = 1$ (\cref{line:c}), each vertex remaining in $U$ queries exactly
$c$ of its \inn (the \inn can be chosen arbitrarily) (\cref{line:pick-arbitrary}).
Suppose that $G'$ is the induced subgraph consisting of all vertices in $U$ and the picked incoming
unmatched neighbors of the vertices in $U$. We run \cite{BFS12} on $G'$ to obtain matchings (\cref{line:ancestor-matchings}).
The matched vertices consists of vertices in $U$ and (some) of their \inn. For each newly matched 
vertex $v$, we remove it from the $\incoming_w$ of each of its out-neighbors $w$ (\cref{line:remove-outn}).
Then, for each vertex in $U$ that becomes matched, we remove it from $U$ (\cref{line:remove-u}).
We double $c$ and proceed with this entire process again if there remains unmatched vertices $u$ in $U$ (\cref{unmatch:outer-while-loop}) 
\emph{and} $I_u$ is not empty (\cref{line:empty-incoming-free}).

The correctness of~\cref{alg:maximal-delete} follows immediately from our procedures.
Our algorithm always maintains a maximal matching after processing
a batch of updates. A vertex becomes unmatched (if it was previously
matched) if and only if it is incident to an edge deletion and the
edge deletion deletes a matched edge. An unmatched vertex can 
be matched to one of its \inn or \outn. We check
both sets of neighbors in our procedure in order to match all unmatched
vertices adjacent to edge updates.

\subsection{Work and Depth Analysis}

Here we show the work and depth analysis of our maximal matching algorithms (\cref{alg:maximal-flips,alg:maximal-insert,alg:maximal-delete}).

\begin{lemma}\label{lem:depth-analysis}
    The depth of~\cref{alg:maximal-flips,alg:maximal-insert,alg:maximal-delete} is $O(\log^2 n (\log \Delta + \log \log n))$ \whp
\end{lemma}

\begin{proof}
We first prove the depth of each algorithm separately and use~\cref{thm:framework} to find the 
total depth.

In~\cref{alg:maximal-flips}, we can process all flipped edges in parallel (\cref{flip:outer-loop}). 
Adding and removing vertices from the hash tables $\incoming_v$ requires $O(\log^* n)$ depth \whp to perform in parallel (\cref{flip:add,flip:remove}).

In~\cref{alg:maximal-insert}, finding all edges in $\batchinserts$ that are between two unmatched vertices can 
be done in parallel in $O(1)$ depth (\cref{line:find-unmatched-1,line:check-unmatched,line:arr-edges-2}).
Then, by the analysis in~\cite{BFS12,Fischer20}, the parallel, static algorithm we use in~\cref{line:static-mm} runs in $O(\log^2 n)$
depth \whp %
Finally, updating the $\incoming_w$ of each out-neighbor $w$ of a newly matched vertex $v$ can be done in parallel
in $O(\log^* n)$ depth \whp (\cref{line:update-insert-1,line:update-insert-2,line:update-insert-3}). %
 Thus,~\cref{alg:maximal-insert} can be done in $O(\log^2 n)$ depth \whp

    In~\cref{alg:maximal-delete}, finding all newly unmatched vertices and making the
    induced subgraph consisting of the these vertices and their \outn can be done in
    $O(\log n)$ depth (\cref{line:iterate-deletes,line:check-matched,line:arr-edges-2,line:induced-out-2}) using a parallel filter.
    As before, running the parallel, static algorithm takes $O(\log^2 n)$ depth \whp (\cref{line:outneighbor-match}).
    Then, removing each newly matched vertex $v$ from the $\incoming_w$ of each out-neighbor $w$ 
    of $v$ takes $O(\log^* n)$ depth (\cref{line:new-matched-outneighbor,line:iterate-out,line:remove-from-out}). 
    Removing the matched vertices $v$ from $U$ can also be done in parallel in $O(\log^* n)$ depth (\cref{line:remove-from-u})
    if $U$ is maintained as a parallel hash table.
    The depth of the outer while loop (\cref{unmatch:outer-while-loop}) is
    $O(\log \Delta)$ since the while loop iterates to a value of $c$ that is at
    most $c = O(\Delta)$. 
    When $c = \Delta$, all incoming neighbors of
    every vertex would be included in the induced subgraph $G'$
    and, hence, a maximal matching is guaranteed in this 
    final case. Because the value
    of $c$ is doubled each time, the total number of iterations 
    of the while loop is $O(\log \Delta)$. 
   The depth of the static
    matching procedure is $O(\log^2 n)$ \whp, and so the total depth
    of~\cref{line:c,line:query-neighbors,unmatch:outer-while-loop,unmatch:inner-for-loop,line:pick-arbitrary,unmatch:empty-set,line:empty-incoming-free,unmatch:induced-subgraph,line:ancestor-matchings,unmatch:remove-from-out,unmatch:out-neighbors,line:remove-outn,line:remove-u,line:double-c} is 
    $O(\log \Delta \log^2 n)$ \whp Finally, the last step of adding each remaining 
    unmatched vertex from $\incoming_w$ of
    each of its \outn $w$ takes $O(\log^* n)$ depth. Note that the matched vertices have already been removed from $\incoming_w$ in the 
    previous lines (\cref{line:remove-from-out,line:remove-outn}). Thus, \cref{alg:maximal-delete} takes $O(\log \Delta \log^2 n)$ depth \whp
    
    By~\cref{thm:framework}, the total depth of~\cref{alg:maximal-delete,alg:maximal-insert,alg:maximal-flips} is $O(\log^2 n (\log\Delta + \log\log n))$
    \whp
\end{proof}

\begin{lemma}\label{lem:work-bound}
    \cref{alg:maximal-flips,alg:maximal-insert,alg:maximal-delete} require 
    $O(\batchsize (\alpha + \log^2 n))$ amortized
    work \whp
\end{lemma}

\begin{proof}
As in the depth proof, we first prove
the work of each of the individual algorithms and then use~\cref{thm:framework} to show the final work bound.
    
    By~\cref{thm:low-outdegree}, the number of edge flips is $O(\batchsize \log^2 n)$ amortized.
    Thus, in~\cref{alg:maximal-flips}, the number of edge flips we process in total is $O(\batchsize\log^2 n)$
    (\cref{flip:outer-loop}). For each edge flip, we spend $O(1)$ work to add and remove, respectively,
    from $\incoming_u$ and $\incoming_v$ (\cref{flip:add,flip:remove}). Then, the total work of~\cref{alg:maximal-flips}
    is $O(\batchsize\log^2 n)$ amortized.

    In~\cref{alg:maximal-insert}, there are at most $\batchsize$ insertions and checking whether the endpoints
    of the edges are unmatched requires $O(\batchsize)$ work. This procedure produces at most $O(\batchsize)$ unmatched
    vertices in $S$ since each edge is incident to two vertices. Then, running the parallel,
    static work-efficient maximal matching algorithm of~\cite{BFS12} on the induced subgraph of the unmatched vertices 
    requires $O(\batchsize)$ work \whp 
    Then, removing the matched vertices from the $\incoming_w$ of each out-neighbor requires $O(\batchsize)$ work.
    Thus,~\cref{alg:insert} takes $O(\batchsize)$ work \whp.
    
    The remainder of the proof focuses on proving the work for~\cref{alg:maximal-delete}.
    First, we note that each vertex that becomes matched is either a 
    vertex in
    $u \in U$ (\cref{line:arr-edges-2}) or is an unmatched incoming neighbor of $u$. There can be at 
    most $4|U|$
    unmatched vertices that may become matched since each deletion update
    can unmatch at most two vertices. 
    Thus, the work of informing all out-neighbors throughout
    the procedure is $O(\batchsize \alpha)$ since by~\cref{thm:low-outdegree}, we 
    maintain an $O(\alpha)$ out-degree orientation and there are $O(\batchsize)$ deletion updates.
    Running the static algorithm (\cref{line:outneighbor-match}) on the induced 
    subgraph of $U$ and its \outn requires $O(\batchsize \alpha)$
    work \whp
    Then, the remaining work comes from the work of looking at $c$ incoming
    neighbors of each node $u \in U$ of the set of vertices $U$ that have not
    been matched and have at least one unmatched neighbor. 
    This requires $O(c|U|)$ work to find a maximal matching
    in the induced subgraph of $U$ and $c$ incoming vertices of each
    vertex in $U$. 
    We perform the following charging argument to calculate the work
    over all attempted $c$ values. 
    
    The key to the charging argument is that we charge the cost of
    \emph{attempted} matchings of a vertex to \emph{when it or its \inn are matched}. 
    More specifically, let a \emph{query}
    be an instance when an in-neighbor %
    of a vertex is chosen in~\cref{line:pick-arbitrary} and
    the static algorithm is run on the induced subgraph of the selected
    in-neighbors and the vertices remaining in $U$.
    When an in-neighbor becomes
    matched, we charge to it the cost of 
    \emph{each previous vertex that queried
    the matched in-neighbor}. %
    Since each matched in-neighbor has at most $O(\alpha)$ out-degree, each such matched in-neighbor will be 
    queried at most $O(\alpha)$ times. The static algorithm we run in~\cref{line:ancestor-matchings} takes work
    that is linear in the size of the induced subgraph $G'$; thus, this is $O(1)$ amortized per vertex in $U$ and
    its chosen \inn.
    Furthermore, as we stated before, there 
    can be at most $O(|U|)$ matched in-neighbors. Thus, the total charged cost to 
    each matched in-neighbor is $O(\alpha|U|)$. We now need to account for the cost of the 
    in-neighbors that were queried but not matched.
    
    To bound the number of such vertices, for each vertex
    $v \in U$, consider the last run of the static algorithm
    where $v$ remains unmatched after the run. During this previous run, 
    all queried \inn of $v$ were matched to some vertices in 
    $U$; if there exists an \inn that is unmatched, it 
    would have been matched to $v$.
    In the final run of the static algorithm for $v$, 
    we query at most two times the number of \inn 
    queried in the previous round. Thus, the remaining
    unmatched, queried \inn in the final run for $v$
    can be charged to 
    the \emph{previous run} where all queried \inn were matched.
    This results in an additional cost of $O(1)$ that is charged to each
    matched in-neighbor.
    
    In total, our charging argument shows that finding the  
    maximal matching in each subgraph takes $O(\alpha|U|)$ work. $|U| = O(\batchsize)$, and so
    the total work of~\cref{alg:maximal-delete} is $O(\alpha \batchsize)$.
    
    By~\cref{thm:framework}, the total work of our batch-dynamic maximal matching algorithm
    is $O(\batchsize(\log^2 n + \alpha))$ amortized, \whp
\end{proof}

\begin{theorem}\label{thm:main-mm}
    Our maximal matching algorithm takes $O(\batchsize (\alpha + \log^2 n))$ amortized work and
    $O(\log^2 n(\log \Delta + \log \log n))$ depth \whp, and uses $O(n\log^2 n + m)$ space.
\end{theorem}

\begin{proof}
The work and depth follow from~\cref{lem:depth-analysis} and~\cref{lem:work-bound}. 
The algorithm uses space equal to the space required by the \kc decomposition
algorithm since it only stores the additional $\incoming_v$ data structures 
which in total takes $O(m)$ additional space. Thus, the 
space required by our algorithm is $O(n\log^2 n + m)$.
\end{proof}

\section{Clique Counting}\label{sec:clique}

A $k$-clique is a set of $k$ vertices where
edges exist between all pairs of vertices in the set. 
Specifically, using our framework (\cref{alg:framework}) and 
our problem specific methods, 
we obtain a $k$-clique counting algorithm (for constant $k$) 
that runs in
$O(\outdeg^{k - 2}|\batch|\log^2 n)$ work and 
$O(\log^2 n)$ depth \whp, using
$O(m\outdeg^{k-2} + n\log n)$ space.

\subsection{Algorithm Overview}\label{sec:clique-overview}
Due to the complexity of our algorithm, we first provide some intuition behind the core ideas
before we give the specific details. First, we make the simple observation
that any clique in a directed acyclic graph has a vertex where all edges in the clique that are adjacent to the vertex are directed out from the vertex.
For a particular clique $C$, we call this vertex the \defn{source} of $C$. 

\begin{observation}\label{obs:simple-source}
Provided a directed acyclic orientation of a graph $G = (V, E)$, for any clique $C \in G$, there exists a \textit{unique} 
vertex $v \in C$ where all edges from $v$ to all other vertices $w \in C$ are directed from $v$ to $w$.  
\end{observation}

\begin{proof}
First, it is easy to see that the source is unique. This is because for any two vertices $u$ and $v$ in the clique, the edge $\{u, v\}$
must be directed either in the $(u, v)$ direction or in the $(v, u)$ direction, one of which makes $v$ (resp.\ $u$) no longer the source.

Then, a simple proof by contradiction shows that the source exists. Suppose for contradiction that all vertices in $C$ have at least one out-neighbor and 
one in-neighbor. We start with vertex $v$. Suppose that $v$'s out-neighbor is $w$ and $v$'s in-neighbor is $u$. By our assumption, $w$
must have at least one out-neighbor, $x$. $x \neq u$, otherwise, there exists a $3$-cycle in the graph. ($x \neq v$ also since we're only
considering simple graphs.) By the same argument, $x$ must have at least one out-neighbor, $y$. $y \not\in \{x, u, v, w\}$, otherwise, by the 
same argument, there would exist a cycle and we only consider simple graphs. Making the same argument for the $k$-th unique
out-neighbor, we require a $(k+1)$-st unique vertex in order to not create a cycle. This contradicts
the fact that $C$ is a $k$-clique.
\end{proof}

We begin our description with an explanation of how to find the newly created cliques resulting from edge insertions.
Using~\cref{obs:simple-source}, we make the second observation that 
for any edge update $(u, v)$, we can count the number of $k$-cliques (for constant $k$)
incident to $(u, v)$ \emph{and} where $u$ is the source vertex of the clique in $O(\alpha^{k-2})$ work and $O(1)$ depth, 
provided an $O(\alpha)$ acyclic low out-degree orientation.
This is because $u$ and $v$ must be in the clique and, thus, there are ${c\alpha \choose k-2} = O(\alpha^{k-2})$ 
additional vertices to choose from 
among $u$'s \outn (for some constant $c$ hidden in the $O(\alpha)$). This observation also means that we do not have to worry about finding
a clique \emph{until after all edges adjacent to its source vertices} are added. The clique will be found
by the last of these source edges when it is added.
Thus, the main challenge of our algorithm is how to find
the cliques resulting from edge updates to other vertices \emph{aside from those adjacent to the source vertex}. 

This leads to our third and final observation:
a $k$-clique can be formed from a $(k-1)$-clique by attaching a source vertex where all edges from the source vertex are directed into 
the vertices of the $(k-1)$-clique. The last crucial observation allows us to count $k$-cliques inductively by counting $(k-1)$-cliques, which are 
in turn counted using $(k-2)$-cliques, and so on. This means that for any $k$-clique $C$, by~\cref{obs:simple-source}, 
there exists a set of \emph{unique} source vertices responsible for the set of smaller cliques within $C$.
Specifically, for every clique $C_i \subseteq C$ of size $i \in [2, \dots, k]$, there exists a source for this clique.  
For every edge insertion, we first determine the possible sets of $k$ vertices which can be \emph{completed} by future 
edge insertions to form $k$-cliques. Potential cliques are determined using the above observation by assuming for each edge insertion $(u, v)$,
$u$ is the source of the clique. Suppose $C$ is one such set.
We assign the responsibility of counting the potential $k$-clique $C$ to the \emph{largest incomplete clique, $C_i \subset C$, without
a source}. This can occur when $C_i$ does not yet have enough edges to determine the source (see 
(1) in~\cref{fig:k-clique} where $\{a, b, d, e\}$ does not yet have a source). %
The base case, the smallest possible largest incomplete clique without a source,
is an edge; once this edge is inserted, the source of the edge counts the clique. The concept of the largest incomplete
clique without a source is fundamental to our algorithm.

Given a batch of edge insertions, if a set of edges completes the largest incomplete clique without a source, $C_i$, of $C$, then the 
new source of this clique is responsible for counting $C$ in the clique count.
Crucially, $C$ cannot be counted until $C_i$ is completed;
furthermore, for any set of $k$ vertices $C$ with a source, but is not a clique,
there exists a $C_i$ that can count $C$. 
If the batch does not complete the clique but a source has
been found for $C_i$, then, we determine the new largest incomplete clique without a source, $C_j$ (where $j<i$), that 
will be responsible for counting $C$. In~\cref{fig:k-clique}, (2) shows a set of insertions that determines that 
$e$ is the source of $\{a, b, d, e\}$. Then, the new largest incomplete clique without a source is $\{a, b, d\}$.

This naturally leads to an algorithm for 
counting $k$-cliques. We create $k-2$ parallel hash tables,
where for each potential $k$-clique $C$, 
we store the indices of the vertices comprising the largest incomplete cliques, $C_{i}$, without a source,
for $i \in [2, k-1]$, in table $I_i$. The values stored in these hash tables are the numbers of $k$-cliques $C$ that would be completed
if $C_{i}$ in table $I_i$ is completed. 
Storing the indices of all vertices allows us to determine the source when the appropriate edges
have been inserted. (More details are given in our detailed algorithm below.) Given
this set of structures, we increase the $k$-clique count when a clique from table $I_i$ is completed 
by a batch of insertions; to increase the $k$-clique count, we use the value stored for the 
clique. If there remains any incomplete cliques, we use 
each table $I_j$ for $j > 2$ to update tables $I_i$ for $i < j$ if any incomplete cliques in $j$ have found sources. 
We give an example illustration of this part of the algorithm in~\cref{fig:k-clique}.

\begin{figure}[!t]
    \centering
    \includegraphics[width=\columnwidth]{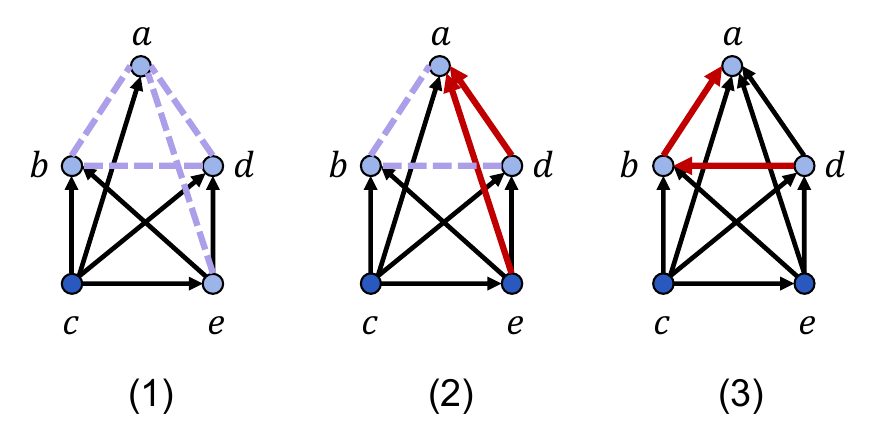}
    \caption{Example of incomplete cliques in our counting algorithm for counting $k = 5$ cliques.
    In (1), $c$ is the source of a potential $5$-clique. $\{a, b, c, d, e\}$ represents a potential $5$-clique.
    We do not yet know the source of the $4$-clique consisting of $\{a, b, d, e\}$ (the purple edges represent
    potential edges), and so we store $\{a, b, d, e\}$ in table $I_4$. (2) shows a set of edge insertions (indicated
    by the red edges) which determines a source ($e$) for the $4$-clique. Thus, we 
    insert $\{a, b, d\}$ in table $I_3$. Suppose that this is the only clique that would 
    be counted when edges are inserted between all pairs in $\{a, b, d\}$. Thus, we associate with 
    this key, a count of $1$ in table $I_3$. Finally, (3) shows two edge insertions which completes the 
    triangle; hence, we count the clique using the key $\{a, b, d\}$ and increment the $k$-clique count 
    using the count associated with it (in this example, the count is $1$) in table $I_3$.}
    \label{fig:k-clique}
\end{figure}

\paragraph{Data Structures} We maintain the following data structures in our algorithm. 
We maintain $k-2$ parallel hash tables, $I_i$ for $i \in [2, k-1]$. 
For each $I_i$, the keys are ordered sets of vertices of size $i$, and the values are counts. 
The counts represent the number of $k$-cliques that would form if all edges among the 
vertices in the keys exist. To prevent over-counting, one edge update 
incident to the new source of any newly completed clique
stored in $I_i$ is responsible for increasing the count by the stored value.

\subsection{$\batchedgeflip$ Implementation}
Our algorithm uses the framework given in~\cref{sec:framework}.
We first instantiate the algorithm for $\batchedgeflip(A, \batchinserts, \batchdeletes)$, which creates a set of edge insertions
and deletions from the flipped edges in $A$ and appends these edges to $\batchinserts$ and $\batchdeletes$. 
The pseudocode is given in~\cref{alg:clique-flips}. In parallel, for each edge that is flipped from
$(u, v)$ to $(v, u)$ (\cref{clique-flip:outer-loop}), we add $(u, v)$ to $\batchdeletes$ (\cref{clique-flip:delete})
and add $(v, u)$ to $\batchinserts$ (\cref{clique-flip:insert}).

\begin{algorithm}[!t]\caption{$\mathsf{CliqueCounting}\batchedgeflip(A, \batchinserts, \batchdeletes)$}
    \label{alg:clique-flips}
    \small
    \begin{algorithmic}[1]
    \Require{A set of edge flips $A$.}
    \Ensure{Updates $\batchinserts, \batchdeletes$.}
    \ParFor{each flipped edge $(u, v) \in A$} \Comment{The edge is flipped from $(u, v)$ to $(v, u)$ and stored as $(u,v)$ in $A$.}\label{clique-flip:outer-loop} 
        \State $\batchdeletes \leftarrow \batchdeletes \cup (u, v)$.\label{clique-flip:delete}
        \State $\batchinserts \leftarrow \batchinserts \cup (v, u)$.\label{clique-flip:insert}
    \EndParFor
    \end{algorithmic}
\end{algorithm}

\subsection{$\batchinsert$ Implementation}
We now instantiate $\batchinsert$ for $k$-clique counting.
The main basis of our $\batchinsert$ and $\batchdelete$ subroutines
is to maintain our parallel hash tables throughout edge insertions and deletions. 
We first describe $\mathsf{CliqueCounting}\batchinsert$, given in~\cref{alg:clique-insert}. 
$\mathsf{CliqueCounting}\batchdelete$ is symmetric and is discussed in~\cref{sec:batch-delete-clique}.

\begin{algorithm}[!t]\caption{$\mathsf{CliqueCounting}\batchinsert(\batchinserts)$}
    \label{alg:clique-insert}
    \small
    \begin{algorithmic}[1]
    \Require{A batch $\batchinserts$ of unique and valid 
    insertion updates.}
    \Ensure{An updated $k$-clique count and updated data structures.}
    \State Let $\kcount$ be the current count of the number of $k$-cliques in the graph.
    \State Insert the edges in $\batchinserts$ into the graph in the orientation specified
    by $\batchinserts$. Mark all edges in $\batchinserts$ in the graph.
    \State Let $R$ be the order of the edges in
    $\batchinserts$.\label{clique-insert:arbitrary-order}
    \ParFor{each edge $(u, v) \in \batchinserts$} \Comment{The edge is oriented from $u$ to $v$.}\label{clique-insert:outer-loop} 
        \For{$i \in [1, \dots, k-2]$}\label{clique-insert:inner-for-loop}
            \ParFor{each subset $T$ of $i$ \outn of $u$ (excluding $v$)}\label{clique-insert:all-subsets}
                \State Let $T'$ be the ordered set of $T \cup \{u, v\}$ sorted by vertex index.\label{clique-insert:orderedset}
                \If{$u$ is the source of $T'$ \textbf{and}
                $(u, v)$ is the earliest in $R$ out of all marked edges from $u$ to $w \in T'$}\label{clique-insert:check-complete}
                    \If{all edges between each pair $x, y \in T'$ exists}\label{clique-insert:check-if-pairs}
                        \State $j \leftarrow 2$.\label{clique-insert:smallest-index}
                \Else\label{clique-insert:not-pairs}
                        \State Find largest incomplete clique without a source, $C'$, in $T'$.\label{clique-insert:lic}
                        \State Let $j$ be the size of $C'$.\label{clique-insert:denote-lic}
                    \EndIf
                    \State $T_{sub} \leftarrow T'$.\label{clique-insert:tsub-clique}
                    \For{$l = i + 1$ to $j$}\label{clique-insert:update-table-loop}
                        \State Find $s$ the source of $T_{sub}$.\label{clique-insert:find-source}
                        \State $T_{sub} \leftarrow T_{sub} \setminus s$.\label{clique-insert:remove-source-clique}
                        \If{$|T_{sub}| = k - 1$}\label{clique-insert:new-clique}
                            \If{$l == i + 1$ and $T'$ is a $(i + 2)$-size clique}\label{clique-insert:check-if-clique} %
                                \State $\kcount \leftarrow \kcount + 1$.\label{clique-insert:increment-count-1}
                            \EndIf
                            \State $I_{l}[T_{sub}] \leftarrow I_{l}[T_{sub}] +1$.\label{clique-insert:assign-new-potential-clique}
                        \ElsIf{$T' \in I_{i+2}$}\label{clique-insert:previous-count}
                            \If{$l == i + 1$ and $T'$ is a $(i + 2)$-size clique}\label{clique-insert:check-if-clique-2}
                                \State $\kcount \leftarrow \kcount + I_{i+2}[T']$.\label{clique-insert:increment-count}
                            \EndIf
                            \State $I_{l}[T_{sub}] \leftarrow I_{l}[T_{sub}] + I_{i + 2}[T']$.\label{clique-insert:reassign-clique}
                        \EndIf
                    \EndFor
                \EndIf
            \EndParFor
        \EndFor
    \EndParFor
    \State Unmark all marked edges in the graph.
    \end{algorithmic}
\end{algorithm}

Before we dive into the details of the implementation of~\cref{alg:clique-insert}, we first provide some intuition for how our 
algorithm implements our intuitive approach in~\cref{sec:clique-overview}. The key piece of information that our algorithm maintains 
after all updates
is how many sets of $k$ vertices, $C$, can be a $k$-clique if a subset of $2 \leq i\leq k-1$ vertices, $C_i \subset C$, 
has edges between all pairs of vertices in the set. Table $i$ keeps all $C_i$ sets of vertices (as keys). 
In other words, the entry $I_{i}[C_i]$ precisely counts
the number of unique sets of $k$ distinct vertices, $C$, where the following two properties hold:

\begin{enumerate}
    \item $C_i \subset C$. 
    \item For every $v \in C \setminus C_i$, the directed edge from $v$ to $w$, $(v, w)$, exists in $G$ for every $w \in C_i$.
\end{enumerate}

The bulk of~\cref{alg:clique-insert} is concerned with updating all of the tables $I_i$ for $i \in [2, \dots, k-1]$ such that the 
above counts hold for every entry. Using these counts, we can find the number of new $k$-cliques created by inserting $\batchinserts$ by checking 
for each $C$, whether its largest incomplete clique without a source is completed. We can do this efficiently because (1) we do not 
need to check this individually for every $C$ since our tables $I_i$ already maintain these counts; and (2) if a largest incomplete clique without a source $C_i$ 
is completed, then it
must be incident to an edge update $(u, v) \in \batchinserts$, where $u$ is the new source of $C_i$ and $v \in C_i$. Knowing (2), we can 
afford to enumerate sets of \outn of $C_i$ to determine whether $C_i$ is a clique. If $C_i$ is a clique, then
we count all $C$ that has it as its largest incomplete clique without a source by adding $I_i[C_i]$ to the cumulative count.
The remaining parts of~\cref{alg:clique-insert} ensure
that we do not over-count newly formed cliques. 

We now describe~\cref{alg:clique-insert} in detail.
For each edge insertion $(u, v)$ where the edge is oriented from $u$ to $v$ (\cref{clique-insert:outer-loop}), we iterate
through \emph{all} possible subsets of $i$ \outn of $u$ (\emph{excluding} $v$, since we know $v$ must be included in the clique)
where $i \in [1, \dots, k-2]$ %
(\cref{clique-insert:inner-for-loop,clique-insert:all-subsets}).
We iterate through $i \in [1, \dots, k-2]$ because $u$ and $v$ necessarily need to be included in the clique.
This is to account for all possible largest incomplete cliques without a source that are currently in our hash tables. In order
to find whether $\batchinserts$ completes any of these cliques, we must find these cliques by performing this enumeration.
Let $T$ be the subset of \outn picked. We consider all cliques of size $i + 2$ consisting of the vertices in $T' = T \cup \{u, v\}$.
Then,~\cref{clique-insert:orderedset} provides a canonical order for the vertices in $T'$ so that we can search for $T'$ in
$I_i$. 
Note that we need to avoid duplicate counting. To avoid duplication, we use the order of the
edge insertions in $\batchinsert$ (\cref{clique-insert:arbitrary-order})
and assign the task of updating the clique count to the first insertion in this order, $(u, w)$, where $w \in T'$. Hence, the if statement in
\cref{clique-insert:check-complete} checks that all of these above conditions are satisfied.
The if statement in~\cref{clique-insert:check-if-pairs} checks whether the newly inserted edges create
a new clique, and if not (\cref{clique-insert:not-pairs}), 
the algorithm then finds the largest incomplete clique without a source, $C'$, that contains a subset of the vertices in $T'$ (\cref{clique-insert:lic}).
The algorithm then sets a parameter $j$ to be the size of $C'$ (\cref{clique-insert:denote-lic}).
If $T'$ is a completed
clique, it passes the check on~\cref{clique-insert:check-if-pairs} and we assign
$j = 2$ (\cref{clique-insert:not-pairs}). 
Now, we consider two possible scenarios. 

First, $(u, v)$ along with the other edge insertions 
in $\batchinserts$ could complete a largest incomplete clique without a source. In this case, we should increase the $k$-clique count if $u$ is also
the new source of the clique.  Furthermore,~\cref{clique-insert:check-if-clique,clique-insert:check-if-clique-2}
check if the clique is completed.
If $u$ is a new source, the clique for which it is the source is completed, and $T' \in I_{i+2}$, then we increment the clique count with the value $I_{i+2}[T']$  
(\cref{clique-insert:previous-count,clique-insert:check-if-clique-2}). The value stored in $I_{i+2}[T']$ is the number of new cliques
that are created if $T'$ is completed. Note that~\cref{clique-insert:increment-count} is only called if $|T'| < k$ since $l \leq k -1$ and
\cref{clique-insert:new-clique} handles the case when $|T'| = k$. %
This is because
we do not store size-$k$ sets of vertices in any of the tables; we do not need to store their values because we can enumerate them directly
by checking all $(k-1)$-size subsets of \outn of every $u$ in every edge insertion $(u, v)$.
We denote the ordered set of vertices that gives the key in table $I_l$ by $T_{sub}$ (initially setting $T_{sub} \leftarrow T'$ (\cref{clique-insert:tsub-clique})).
If $T' \not\in I_{i+2}$ and the size of $T_{sub}$ is $k - 1$ (implying the size of $T'$ is $k$), then we directly increment the clique count by $1$ (\cref{clique-insert:new-clique,clique-insert:check-if-clique,clique-insert:increment-count-1}). As before, in this case, we directly enumerate
the new clique for the edge insertion $(u, v)$ without needing to check the tables. 
This also means that $u$ is the source of the newly created $k$-clique consisting of the vertices in $T'$. 

After we update the $k$-clique count, we must then update
the $I_i[C_i]$ counts for each $C_i \subset T'$. We need
to update these counts because now $T'$ contains a vertex
$v \in T' \setminus C_i$ where there exists an edge
$(v, w)$ for every $w \in C_i$ (this vertex did not 
exist previously). Thus, the count for $C_i$ must be
incremented by $I_{i+2}[T']$ if $i+2 \leq k-1$, and $1$ if $i+2 = k$. 
This is because, as previously discussed, intuitively, $I_{i+2}[T']$ stores the number of  $k$-cliques that would be created if $T'$ were completed, so we must similarly maintain the number of $k$-cliques created now that each $C_i$ is completed.
We prove this more concretely in~\cref{sec:clique-correctness}.
When $T'$ is a clique, there exists a $C_i$ 
for every $2 \leq i < |T'|$
whose entry $I_{i}[C_i]$ needs to be updated.
So, \cref{clique-insert:update-table-loop} loops through
each of these possible sizes and the entries are
updated by~\cref{clique-insert:assign-new-potential-clique} 
or~\cref{clique-insert:reassign-clique} depending on 
whether $|T'| = k$.

The second scenario is that $(u, v)$ and the other edge insertions in $\batchinserts$ do not complete a clique but create a new source among
the vertices in $T'$.
The means that we need to find a new largest incomplete clique without a source within $T'$. Again, to avoid duplication, we assign the task to the 
earliest edge update that is incident to $u$. 
Similar to the case when $T'$ is a clique, the
algorithm also needs to update tables $I_{i + 1}$ to $I_j$ in this case 
(\cref{clique-insert:update-table-loop}). 
However, we do not update \emph{all} of the
tables since $T'$ still has a 
largest incomplete clique without a source (\cref{clique-insert:lic}). 
Let $l \in [j, i+1]$, the table $I_j$ is updated 
with the number of cliques that would be counted by it if it were the largest incomplete clique without a source. We need to update 
all of these tables (instead of just table $I_j$) in order to be able to 
handle deletions. This is due to the fact that when a smaller clique 
becomes incomplete due to a batch of deletions, it may cause a larger
$k$-clique to become incomplete. We cannot afford to find all 
such affected $k$-cliques; thus, we must store this information in the 
tables.

To compute the key for table $I_l$, we need to remove the source of $T_{sub}$ from $T_{sub}$ (\cref{clique-insert:find-source,clique-insert:remove-source-clique}).
Then, there are two cases we must consider (\cref{clique-insert:new-clique} and~\cref{clique-insert:previous-count}). 
In the first case, when $|T'| = k$, $u$ is a newly created
source for a new potential $k$-clique (\cref{clique-insert:new-clique}); thus, no entries in the tables
have counted $T'$ yet and we increment the count of $I_{k-1}[T_{sub}]$ by 
$1$ (\cref{clique-insert:new-clique,clique-insert:assign-new-potential-clique}) so that $T'$ will be counted when $T_{sub}$ is completed 
as the largest incomplete
clique without a source. 
In the second case (\cref{clique-insert:previous-count}), $T'$ is 
already an entry in table $I_{i+2}$; this means that it already counts
the a number of $k$-cliques that exist if $T'$ is a clique. In this case,
we increase the value for $I_{l}[T_{sub}]$ by $T_{i + 2}[T']$ 
since by definition of the values we store in $I_l[T_{sub}]$, if $T_{sub}$ is a clique, then all the $k$-cliques that are counted when $T'$ is a clique will now be counted when $T_{sub}$ is completed as the largest incomplete clique without a source
(\cref{clique-insert:reassign-clique}). 

\subsection{$\batchdelete$ Implementation}\label{sec:batch-delete-clique}

Our $\mathsf{CliqueCounting}\batchdelete$ algorithm is nearly identical to our $\mathsf{CliqueCounting}\batchinsert$ 
algorithm; in places where we assign clique counts
in the insertion algorithm,
we instead remove clique counts in the deletion
algorithm. Such changes are expected since
deletions of edges remove cliques from the count
and also remove assignments of cliques to largest
incomplete cliques without a source.
The pseudocode is provided in~\cref{alg:clique-delete}. The few changes
to the algorithm are highlighted in blue.

\begin{algorithm}[!t]\caption{$\mathsf{CliqueCounting}\batchdelete(\batchdeletes)$}
    \label{alg:clique-delete}
    \small
    \begin{algorithmic}[1]
    \Require{A batch $\batchdeletes$ of unique and valid 
    deletion updates.}
    \Ensure{An updated $k$-clique count and updated data structures.}
    \State Let $\kcount$ be the current count of the number of $k$-cliques in the graph.
    \State Insert all edges in {\color{blue} $\batchdeletes$} into the graph in the orientation specified by {\color{blue} $\batchdeletes$}. 
    Mark all edges in {\color{blue} $\batchdeletes$} in
    the graph.
    \State Let $R$ be the order of the edges in
    $\batchdeletes$.\label{clique-delete:arbitrary-order}
    \ParFor{each edge $(u, v) \in \batchdeletes$} \Comment{The edge is oriented from $u$ to $v$.}\label{clique-delete:outer-loop} 
        \For{{\color{blue} $i \in [k-2, \dots, 0]$}}\label{clique-delete:inner-for-loop}
            \ParFor{each subset $T$ of $i$ \outn of $u$}\label{clique-delete:all-subsets}
                \State Let $T'$ be the ordered set of $T \cup \{u, v\}$ sorted by vertex index.\label{clique-delete:orderedset}
                \If{$u$ is the source of $T'$ \textbf{and}
                $(u, v)$ is the earliest in $R$ out of all marked edges from $u$ to $w \in T'$}\label{clique-delete:check-complete}
                    \If{all edges between each pair $x, y \in T'$ exists}\label{clique-delete:check-if-pairs}
                        \State $j \leftarrow 2$.\label{clique-delete:smallest-index}
                \Else\label{clique-delete:not-pairs}
                        \State Find largest incomplete clique without a source, $C'$, in $T'$.\label{clique-delete:lic}
                        \State Let $j$ be the size of $C'$.\label{clique-delete:denote-lic}
                    \EndIf
                    \State $T_{sub} \leftarrow T'$.\label{clique-delete:tsub-clique}
                    \For{$l = i + 1$ to $j$}\label{clique-delete:deletion-loop}
                        \State Find $s$ the source of $T_{sub}$.\label{clique-delete:find-source}
                        \State $T_{sub} \leftarrow T_{sub} \setminus s$.\label{clique-delete:remove-source-clique}
                        \If{$|T_{sub}| = k - 1$}\label{clique-delete:new-clique}
                            \If{$l == i + 1$ and $T'$ is a clique}\label{clique-delete:check-if-clique}
                                \State $\kcount \leftarrow \kcount$ 
                                {\color{blue} $- 1$}.\label{clique-delete:increment-count-1}
                            \EndIf
                            \State $I_{l}[T_{sub}] \leftarrow I_{l}[T_{sub}]$ 
                            {\color{blue}$- 1$}.\label{clique-delete:assign-new-potential-clique}
                        \ElsIf{$T' \in I_{i+2}$}\label{clique-delete:previous-count}
                            \If{$l == i + 1$ and $T'$ is a clique}\label{clique-delete:check-if-clique-2}
                                \State $\kcount \leftarrow \kcount$ {\color{blue} $- I_{i+2}[T']$}.\label{clique-delete:increment-count}
                            \EndIf
                            \State $I_{l}[T_{sub}] \leftarrow I_{l}[T_{sub}]$ {\color{blue}$- I_{i + 2}[T']$}.\label{clique-delete:reassign-clique}
                        \EndIf
                    \EndFor
                \EndIf
            \EndParFor
        \EndFor
    \EndParFor
    \State {\color{blue} Delete} all marked edges in the graph.
    \end{algorithmic}
\end{algorithm}

\subsection{Correctness}\label{sec:clique-correctness}

To prove the correctness of our algorithm, we first show that~\cref{alg:clique-insert} and~\cref{alg:clique-delete} accurately store the
counts associated with the largest incomplete cliques without sources. For simplicity, we provide separate lemmas for~\cref{alg:clique-insert} and 
\cref{alg:clique-delete}, although fundamentally, the proof techniques are the same for both algorithms.

We use the following notation in~\cref{lem:correct-insert-table} and~\cref{lem:correct-delete-table}.
Let $c_{L}$ be the number of sets of $k$ vertices in the graph which do not form a clique, contains a source, and
whose largest incomplete cliques without sources is the set of vertices in $L$ \emph{after 
processing the current, input batch of updates}. We show that after running~\cref{alg:clique-insert} or~\cref{alg:clique-delete}, $I_{|L|}[L] = c_{L}$.
In fact, we show an even stronger lemma; suppose that $J$ is a set of vertices in the graph where $2 \leq |J| \leq k-1$. If we remove all edges between pairs of 
vertices in $J$, let $c_{J}$ be the number of sets of $k$ vertices that do not form a clique, contains a source, and
whose largest incomplete cliques without sources is the set of vertices in $J$. We show that $I_{|J|}[J] = c_{J}$. This stronger
form of the lemma is not necessary if we only consider insertion updates; however, under deletion updates, we require this 
stronger lemma in order to prove correctness. It is sufficient to assume that all data structures are maintained corrected at the 
beginning of~\cref{alg:clique-insert} and~\cref{alg:clique-delete} and they remain correct at the end of the algorithms (since by 
induction, this would prove that the data structures are always correctly maintained).

\begin{lemma}\label{lem:correct-insert-table}
After running~\cref{alg:clique-insert} on $\batchinserts$, $I_{|J|}[J] = c_{J}$ for every $c_J$ where $2 \leq |J| \leq k-1$.
\end{lemma}

\begin{proof}
We prove this lemma via induction on the table index $i$, starting with $i = k-1$. We first prove our base case for $i = k-1$. 
In~\cref{alg:clique-insert}, the value stored in $I_{k-1}[J]$ is only ever incremented in~\cref{clique-insert:assign-new-potential-clique}
since this is the only time when table $I_{k-1}$ can be modified (the condition in Line~\ref{clique-insert:previous-count} is never satisfied for 
any entries in table $I_{k-1}$). By the condition given in~\cref{clique-insert:check-complete}, $I_{k-1}[J]$ is only incremented when $T'$ has a source
$s$ where $J = T' \setminus \{s\}$. Furthermore, the condition that $I_{k-1}[J]$ is incremented by the earliest edge update incident to $s$
ensures that it is incremented at most once by each $T'$. The number of sets $c_J$ of vertices $T'$ where $J \in I_{k-1}$ is the largest incomplete
clique without a source (if all edges in $J$ are removed) is precisely the number of vertices $s$ in the graph with edges directed into 
all vertices in $J$. Our argument above shows that $I_{k-1}[J]$ is incremented exactly once for each such vertex $s$; furthermore,
it is incremented only if $s$ is adjacent to an edge update $(s, x) \in \batchinserts$ and $x \in J$. This proves our base case.  

We assume for our induction step that $I_{|J|}[J] = c_J$ for all tables $I_{j}$ for $j \in [k-1, \dots, k-l]$ and prove the lemma holds
for table $I_{k-l-1}$. The value $I_{k-l-1}[J]$ is increased in~\cref{clique-insert:reassign-clique}. Every $T'$ with $k$ vertices
increases the value of $I_{k-l-1}[J]$ by $1$ if its largest incomplete clique without a source has size $\leq |J|$. This is easy to see
since if all edges from $J$ are removed, then $J$ would be the largest incomplete clique without a source for $T'$. By our induction 
hypothesis, the counts of these $T'$s are correctly stored in tables $I_j$ for $j \in [k-1, \dots, k-l]$. \cref{clique-delete:check-complete}
ensures that only one edge update is responsible for incrementing $I_{k-l-1}[J]$ for each $T'$; furthermore, it ensures that $I_{k-l-1}[J]$
is incremented with the value stored in $I_{|C|}[C]$ where $C \subseteq T'$ is the previous largest incomplete clique without a source
before the current batch $\batchinserts$ of insertions. $J \subset C$ by definition, and $J$ is guaranteed to be the largest incomplete clique without 
a source (after the edge insertions in $\batchinserts$) by our argument above. In addition, each $T'$ is counted in at most one $C \subset T'$
in each table $I_j$ where $j \in [k-1, \dots, k-l]$. This is true by our induction hypothesis since each $T'$ has one unique $C$ where $|C| = j$
which is the largest incomplete clique without a source (if the edges in $C$ are removed).

The last step we need to prove in order to prove our induction hypothesis is that $I_{k-l-1}[J]$ is incremented by $1$ for $T'$ 
with exactly one $I_{|C|}[C]$ where $C \subset T'$. We prove this via contradiction. Suppose there are two subsets $C' \subset C\subset T'$
which are used to increment $I_{k-l-1}[J]$. Let $s' \in C'$ be the source of $C'$ and $s \in C$ be the source of $C$. 
This means that in order to satisfy~\cref{clique-delete:check-complete}, $s$ and $s'$ must be incident to some update $(s, x) \in \batchdeletes$ (resp.\ 
$(s', x') \in \batchdeletes$) where $x, x' \in T'$. This means that $C$ was the previous largest incomplete clique without a source for $T'$ and so 
$C'$ would not contain a count for $T'$ by our
induction hypothesis. Since, we process the tables in~\cref{clique-delete:inner-for-loop} starting with table
$I_{2}$ in increasing order of table index, $I_{k-l-1}[J]$ cannot be incremented with the count for $T'$ from $C'$, a contradiction.
Thus, $I_{|C|}[C]$ correctly counts all $T'$ and hence, 
$I_{k-l-1}[J]$ is incremented exactly once for each $T'$ and we have proven our inductive step.
\end{proof}

The proof of the property for~\cref{alg:clique-delete} is almost identical 
to~\cref{lem:correct-insert-table} except to account for the few
changes shown in blue in~\cref{alg:clique-delete}. For simplicity, we 
present only the parts of the proof that requires more
effort than replacing decrement for all mentions of increment in 
the proof of~\cref{lem:correct-insert-table}.

\begin{lemma}\label{lem:correct-delete-table}
After running~\cref{alg:clique-delete} on $\batchdeletes$, $I_{|J|}[J] = c_{J}$ for every $c_J$ where $2 \leq |J| \leq k-1$.
\end{lemma}

\begin{proof}
We prove this lemma via induction on the table index $i$, starting with $i = k-1$. The proof of our base case for $i = k-1$ directly follows from
the proof of the base case in~\cref{lem:correct-insert-table} when
we replace instances of increment with decrement.

We assume for our induction step that $I_{|J|}[J] = c_J$ for all tables $I_{j}$ for $j \in [k-1, \dots, k-l]$ and prove the lemma holds
for table $I_{k-l-1}$. The value $I_{k-l-1}[J]$ is increased in~\cref{clique-insert:reassign-clique}. The proof of the inductive
step follows from the proof of the inductive step in the proof
of~\cref{lem:correct-insert-table} by replacing instances of increase
by decrease, except for the last step which we prove below.

The last step we need to prove in order to prove our induction hypothesis is that $I_{k-l-1}[J]$ is decremented by $1$ for $T'$ 
with exactly one $I_{|C|}[C]$ where $C \subset T'$. We prove this via contradiction. The initial setup is the same as the setup in the proof
of~\cref{lem:correct-insert-table}. 
Suppose there are two subsets $C' \subset C\subset T'$
which are used to decrement $I_{k-l-1}[J]$. Let $s' \in C'$ be the source of $C'$ and $s \in C$ be the source of $C$. 
This means that in order to satisfy~\cref{clique-delete:check-complete}, $s$ and $s'$ must be incident to some update $(s, x) \in \batchinserts$ (resp.\ 
$(s', x') \in \batchinserts$) where $x, x' \in T'$. 

This means that $C$ \emph{is now}
the largest incomplete clique without a source for $T'$ after processing
the deletions in $\batchdeletes$. Thus, because we process the tables 
in \emph{decreasing order} by table index, starting with table $I_2$ 
(\cref{clique-delete:inner-for-loop}), $C$ satisfies the 
conditions in~\cref{clique-delete:check-complete} 
and by~\cref{clique-delete:reassign-clique}, $C$ would have deleted the 
count of $T'$ from $I_{|C'|}{C'}$. Thus,
$C'$ would not contain a count for $T'$ and $I_{k-l-1}[J]$ cannot be incremented with the count for $T'$ from $C'$, a contradiction.
$I_{|C|}[C]$ correctly counts all $T'$ by our induction hypothesis
and hence, $I_{k-l-1}[J]$ is decremented exactly once for each $T'$ and we have proven our inductive step.
\end{proof}

We are now ready to prove that our algorithms correctly return the 
$k$-clique count provided batches of updates.

\begin{theorem}\label{thm:clique-correct}
Our algorithms, \cref{alg:clique-insert} and~\cref{alg:clique-delete}, correctly returns the number of $k$-cliques in a 
given input graph, $G = (V, E)$, provided batches of updates
$\batchinserts$ and $\batchdeletes$, respectively.
\end{theorem}

\begin{proof}
Provided~\cref{lem:correct-insert-table} 
and~\cref{lem:correct-delete-table}, we only need to show the following:
given $\batchinserts$, each $k$-clique 
$C$ completed by $\batchinserts$ (i.e.\ 
contains at least one edge in $\batchinserts$), is counted exactly
once, and by exactly one update edge incident to the source of 
its largest incomplete clique without a source (\emph{prior} 
to the insertions);
given $\batchinserts$, each $k$-clique $C$ destroyed by $\batchdeletes$
(i.e.\ contains at least one edge in $\batchdeletes$), is subtracted
exactly once, and by exactly one update edge incident to the source
of its largest incomplete clique without a source (\emph{after} the 
deletions).

We first prove the above is true for insertions. The if statement in
\cref{clique-insert:check-complete} ensures at most one update 
edge for a set of vertices $C \subset T'$, where $T'$ is a newly 
formed clique, increments the clique count. Now, we prove 
that at most one subset of vertices increments the clique 
count for $T'$. Suppose for contradiction
two sets of vertices $C' \subset C \subset T'$ increments the total clique
count by $1$ for $T'$. Then, in order to pass the if statement in
\cref{clique-insert:check-complete}, the sources of both $C'$ and $C$
must be adjacent to updates in $\batchinserts$ that point to vertices
in $C'$ and $C$. Since $C' \subset C$, $C$ was the previous largest
incomplete clique without a source for $T'$. 
By~\cref{lem:correct-insert-table}, $C'$ does not contain the count for
$T'$ and thus, only $C$ increments the total clique count by $1$ for $T'$,
a contradiction. 

To prove that at least one subset of vertices increments the clique count for $T'$, suppose that $C$ was the previous largest incomplete clique
without a source for $T'$ but $C$ does not increment the clique count.
Since $T'$ is a new clique, it must be incident to at least one edge update
in $\batchinserts$. Since $C$ does not increment the clique count, it 
must not have found a source (and cannot satisfy~\cref{clique-insert:check-if-pairs}). (It must satisfy \cref{clique-insert:check-if-clique} or \cref{clique-insert:check-if-clique-2}
since $T'$ is a clique and we iterate through all possible $i$).
Since $C$ does not have a source, by~\cref{obs:simple-source}, 
it must be missing at least one edge. Then, $T'$ is not a clique, 
a contradiction.

The proof follows symmetrically for~\cref{lem:correct-delete-table}
except that instead of the previous largest incomplete clique without 
a source, we care about the largest incomplete clique without a
source \emph{after} processing $\batchdeletes$. Suppose for 
contradiction two sets of vertices $C' \subset C\subset T'$ decrement
the clique count. Then, their sources must both be incident to 
edge updates. Since, $C' \subset C$, $C$ is processed first by~\cref{clique-delete:inner-for-loop}
using~\cref{clique-delete:deletion-loop,clique-delete:assign-new-potential-clique,clique-delete:reassign-clique}. This means that the count of $T'$
would have been subtracted from $I_{|C'|}[C']$ and it cannot decrement
the clique count by $1$ for $T'$, a contradiction. Suppose
instead, that $C$ is the largest incomplete clique without a source
for $T'$ after processing $\batchdeletes$ and it does not decrement
the clique count. Either one of two scenarios can occur: 
either $I_{|C|}[C]$ no longer has the count for $T'$ or 
the source $s$ of $C$ is not incident to any updates.
No $C''$ where $C \subset C''$ can
decrement $I_{|C|}[C]$ for $T'$ since by our assumption, the source of
$C''$ is not incident to any updates directed into vertices
in $C''$. Thus the first scenario cannot occur and we consider the 
second scenario where the source $s$ 
must not be incident to an edge update directed into the vertices in
$C$. Then, $s$ still has all its directed edges to the vertices in
$C$ and so is the source of $C$. This means that $C$ has a source
and cannot be the largest incomplete clique \emph{without a source}, 
a contradiction. 
\end{proof}

Together with the proof of correctness of our framework,~\cref{sec:framework}, our algorithm correctly provides the 
$k$-clique count provided a batch of updates, $\batch$.

\subsection{Work and Depth Analysis}

We note for the following result that $\alpha$ is defined as $\max(\alpha_{b}, \alpha_{a})$
where $\alpha_{b}$ is the arboricity before the current batch of updates is processed and $\alpha_{a}$ is the 
arboricity after the current batch of updates is processed.

\begin{theorem}
We obtain a batch-dynamic $k$-clique counting algorithm that 
takes $O(\alpha \batchsize \log^2 n)$ amortized work and
$O(\log^2 n \log \log n)$ depth \whp, using
$O(m\alpha^{k-2} + n \log^2 n)$ space. %
\end{theorem}

\begin{proof}
We first show the work, depth, and space of our algorithms,~\cref{alg:clique-flips,alg:clique-delete,alg:clique-insert},
and then use~\cref{thm:framework} to show the bounds for our algorithm.
Note that the increments and decrements to the global $k$-clique count can be performed in $O(\log n)$ depth in parallel by writing each update to an array, and then using parallel reductions at the end to update the global $k$-clique count. 
We use the same strategy for updating the hash table counts.
Furthermore, our parallel hash table primitives allow
us to concurrently modify elements in parallel in $O(\log n)$
depth \whp

In~\cref{alg:clique-flips},
the batches $\batchinserts$ and $\batchdeletes$ can be obtained in 
$O(|\batch|\log^2n)$ work and $O(\log n)$ depth. Note that by construction, $|\batchinserts|, |\batchdeletes| = O(|\batch| + |A|) = O(\batchsize\log^2 n)$.
All edges can be checked in parallel (\cref{clique-flip:outer-loop}) and inserted into parallel dynamic arrays;
we can also use a simple parallel filter.
For the remainder of this proof,
we discuss the work and depth complexity for a batch $\batchinserts$ of 
edge insertions in~\cref{alg:clique-insert}; the deletion algorithm (\cref{alg:clique-delete}) has the same work, depth, and space complexity.

All edge insertions are processed in parallel using a parallel loop (\cref{clique-insert:outer-loop}). We then run a sequential for loop
of depth $O(k)$ (\cref{clique-insert:inner-for-loop}). Let $i$ be
the current index of the sequential for loop.
In order to process edge insertions $(u, v)$, where $u$ is a source, we 
iterate in parallel over all sets $T$ of $i + 1$ out-neighbors of $u$ 
including $v$. Since there are at most $O(\alpha)$ 
out-neighbors of $u$, and 
since $v$ is necessarily included, we have $\binom{c\alpha}{i} = 
O(\outdeg^{i})$ possible sets $T$ (assuming constant $k$). 
For constant $k$, we
perform a constant number of parallel hash table operations and checks 
for the existence of edges per set $T$ (\cref{clique-insert:check-if-pairs,clique-insert:smallest-index,clique-insert:lic,clique-insert:denote-lic,clique-insert:denote-lic,clique-insert:tsub-clique}). 
We make $O(k)$ iterations of the for loop in~\cref{clique-insert:update-table-loop};
updating the hash tables (\cref{clique-insert:assign-new-potential-clique,clique-insert:reassign-clique}) 
require $O(k)$ total work per edge update. Checking for the source of $T_{sub}$ over
all $T_{sub}$ requires $O(k^2)$ work per edge update (\cref{clique-insert:find-source}). 
Thus, per edge insertion $(u, v)$, for constant $k$, 
we incur $O(\outdeg^{i})$ work and $O(\log n)$ depth \whp 
Over all $i \in [0, \dots, k-2]$, this results in 
$\sum_{i=0}^{k-2} O(\alpha^{i})=O(\alpha^{k-2})$ total work over 
all $i$, \whp The depth is $O(\log^* n)$ \whp due to the hash table 
operations and updating the table values by writing to an array
and using a parallel reduction for each entry results in $O(\log n)$ depth.

Lastly, we update the global $k$-clique count by 
writing each update to an array 
and using a parallel reduction at the end, which maintains the same work 
and depth bound.

Processing the entire batch of insertions in parallel, we have $O(\alpha^{k-2} |\batch|\log^2 n)$ amortized work and $O(\log n)$ depth \whp Thus, in total, our $k$-clique counting algorithm takes $O(\outdeg^{k - 2}|\batch|\log^2 n)$ amortized work and $O(\log^2 n)$ depth \whp by~\cref{thm:framework}.

Our space usage is proportional to the space required to store the contents of the parallel hash tables $I_i$ for $i \in [2, \dots, k-1]$. By construction, for each edge insertion $(u, v)$, we create at most $\sum_{j=0}^{k-2} O(\outdeg^{j-2}) = O(\outdeg^{k-2})$ hash table entries across all $I_i$. This follows directly from our work analysis. Thus, in total, we use space proportional to $O(m\outdeg^{k-2})$.
\end{proof}

\subsection{Comparison with Previous Work}
The best-known batch-dynamic algorithm
for $k$-clique counting for graphs with low arboricity
is given by Dhulipala \etal~\cite{DLSY21}. They give a 
$O(\batchsize m \alpha^{k-4})$ expected work and 
$O(\log^{k-2}n)$ depth \whp algorithm using $O(m + \batchsize)$ space. 
Our algorithm improves upon the work of this previous 
result when $m = \omega(\alpha^2\log^2n)$. 
Note
that $\alpha \leq \sqrt{m}$ \cite{ChNi85}. 
Furthermore, in real-world graphs, often $\alpha \ll \sqrt{m}$,
since real-world graphs tend to have small arboricity. 

Our algorithm achieves better depth for all $k > 4$. For $4$-cliques, 
our depth matches the previous algorithm while for larger cliques, we achieve
a better depth. Finally, we obtain these gains with an increase in space
of $O(\alpha^{k-2} + \log^2 n)$ multiplicative factor, 
but for bounded arboricity graphs, this increase in space
is small.
\section{Coloring Algorithms}\label{sec:coloring}
The vertex coloring problem looks to assign colors to assign colors to vertices 
such that no two adjacent vertices are assigned the same color. A $c$-vertex coloring
uses at most $c$ colors to color all vertices in the graph. In this section, we present 
two batch-dynamic algorithms. Although our algorithms are based heavily on the sequential algorithms by Henzinger et al.~\cite{HNW20}, we present them as an example of using our framework.
One maintains an (\defn{explicit}) coloring over the 
vertices and one maintains an \defn{implicit} coloring. In the explicit setting, a valid
coloring is always maintained in the graph among all vertices. In the implicit setting, the algorithm 
maintains a set of data structures and on \defn{queries} of one or more vertices, returns
a coloring that is valid on the induced subgraph of the queried vertices. Thus, in the implicit 
setting, both updates and queries could take $\Omega(1)$ work to process. Below, we give our vertex
coloring algorithms.

\subsection{Explicit $O(\alpha \log n)$-Coloring}\label{sec:explicit}
In this section, we present 
a parallel batch-dynamic, randomized $O(\alpha \log 
n)$-coloring algorithm that is robust against an oblivious adversary 
and uses $O(\log^2 n)$ amortized work, matching the amortized running
time in the sequential setting. Notably, $\alpha$ is the current 
arboricity of the graph, after processing the current batch of updates.
This algorithm is inspired by the coloring algorithm 
of Henzinger \etal~\cite{HNW20}, and directly uses the PLDS. 

\paragraph{Sequential Explicit Coloring Algorithm of Henzinger \etal~\cite{HNW20}}
The explicit vertex coloring algorithm of Henzinger \etal~\cite{HNW20} uses a separate 
palette of colors for each level in the LDS. 
When a vertex moves to a new level, it chooses a color uniformly at random from among the 
free colors in the palette at its level; specifically, the free colors are colors that are
not occupied by its neighbors. If an edge insertion occurs between two vertices with the same color,
then an arbitrary endpoint chooses a new color uniformly at random from the free colors in its palette.

Since the sequential algorithm processes one vertex at a time, it
does not have to deal with color conflicts when more than one vertex chooses a free color from the same palette. 
However, in the batch-dynamic case, this is an issue since more than one vertex on the same level may need
to choose a free color. We show that allowing such vertices to keep choosing colors is sufficient to ensure
both work-efficiency and low depth, \whp, provided we give a large enough palette.

\paragraph{Batch-Dynamic $O(\alpha \log n)$-Vertex Coloring}
As in the previous sections, we use our framework given in~\cref{sec:framework} for our coloring algorithms.
The pseudocode for our implementations of the methods are given in~\cref{alg:explicit-coloring-flips,alg:explicit-coloring-insert,alg:explicit-coloring-delete}.
Given that~\cref{alg:explicit-coloring-insert} and~\cref{alg:explicit-coloring-delete} are very similar to each other,
we explain all three algorithms here.

\begin{algorithm}[!t]\caption{$\mathsf{ExplicitColoring}\batchedgeflip(A, \batchinserts, \batchdeletes)$}
    \label{alg:explicit-coloring-flips}
    \small
    \begin{algorithmic}[1]
    \Require{A set of edge flips $A$.}
    \Ensure{A list of vertices which changed levels.}
    \State $S \leftarrow \emptyset$.
    \ParFor{each edge $(u, v) \in A \cup \batchinserts \cup \batchdeletes$} \label{explicit-coloring-flip:outer} 
        \If{$u$ (resp.\ $v$) changed levels}
            \State $S \leftarrow S \cup \{u\}$ (resp.\ $\{v\}$).\label{ecf:s}
        \EndIf
    \EndParFor
    \State \Return $S$.
    \end{algorithmic}
\end{algorithm}

\begin{algorithm}[!t]\caption{$\mathsf{ExplicitColoring}\batchdelete(\batchdeletes)$}
    \label{alg:explicit-coloring-delete}
    \small
    \begin{algorithmic}[1]
    \Require{A batch $\batchdeletes$ of unique and valid 
    insertion updates.}
    \Ensure{A valid $O(\alpha \log n)$-coloring.}
    \ParFor{each edge $\{u, v\} \in \batchdeletes$}\label{ec-del:outer}
        \If{$u \in S$ (resp.\ $v \in S$)}\label{ec-del:check-change-level}
            \State $S \leftarrow S \setminus \{u\}$ (resp.\ $\{v\}$).\label{ec-del:delete-s}
            \While{$u$ (resp.\ $v$) has a neighbor with the same color}
                \State Recolor $u$ (resp.\ $v$) with a free color from $P_{\level(u)}$ (resp.\ $P_{\level(v)}$) picked 
                uniformly at random.\label{ec-del:recolor}
            \EndWhile
        \EndIf
    \EndParFor
    \end{algorithmic}
\end{algorithm}

\begin{algorithm}[!t]\caption{$\mathsf{ExplicitColoring}\batchinsert(\batchinserts)$}
    \label{alg:explicit-coloring-insert}
    \small
    \begin{algorithmic}[1]
    \Require{A batch $\batchinserts$ of unique and valid 
    insertion updates.}
    \Ensure{A valid $O(\alpha \log n)$-coloring.}
    \ParFor{each edge $\{u, v\} \in \batchinserts$}\label{ec-ins:outer}
        \If{$c(u) = c(v)$}\label{ec-ins:check-conflict}
            \While{$u$ has a neighbor with the same color}\label{ec-ins:while-conflicting}
                \State Recolor $u$ with a free color from $P_{\level(u)}$ picked 
                uniformly at random.\label{ec-ins:recolor}
            \EndWhile
            \State $S \leftarrow S \setminus \{u\}$.
        \EndIf
    \EndParFor
    \ParFor{each remaining vertex $v \in S$}\label{ec-ins:leftover}
        \While{$v$ has a neighbor with the same color}\label{ec-ins:leftover-conflict}
            \State Recolor $v$ with a free color from $P_{\level(v)}$ picked 
                uniformly at random.\label{ec-ins:left-recolor}
        \EndWhile
    \EndParFor
    \end{algorithmic}
\end{algorithm}

First,~\cref{alg:explicit-coloring-flips} determines the set of vertices which changed levels
after processing the batches of insertions and deletions. We can find these vertices in parallel (\cref{explicit-coloring-flip:outer}).
The vertices which changed edges are added to the set $S$ (\cref{ecf:s} which can be accessed by 
\cref{alg:explicit-coloring-insert} and \cref{alg:explicit-coloring-delete}.

Each level $\ell \in \group_i$ is initialized with a unique palette with
$2 \cdot (2 + 3/\lambda)(1+\eps)^{i}$ colors.
Vertices on level $\ell$ will be colored
only with colors from the palette on level $\ell$. $P_{\ell}$ denotes the palette for level $\ell$.
A free color for a vertex $v$ is a color from $P_{\level(v)}$ that is not occupied by any neighbors of $v$.

 When a vertex moves up or down one or more levels, it recolors itself
    using the palette of the new level $\ell \in g_i$.\footnote{This 
    step is necessary to maintain our bound in terms of the current
    arboricity, $\alpha$, for the number of colors 
    used in the coloring.} In~\cref{alg:explicit-coloring-delete}, because deletions cannot cause
    two neighboring vertices to have the same color, we only need to recolor the vertices which changed levels (\cref{ec-del:check-change-level}).
    To pick a free color (\cref{alg:explicit-coloring-insert} \cref{ec-ins:recolor}, \cref{alg:explicit-coloring-delete} \cref{ec-del:recolor}),
    the vertex $v$ looks at the colors occupied by all
    its up-neighbors and picks a color 
    that does not collide with the colors of
    any of its up-neighbors. We only need to check the up-neighbors because 
    the palettes are distinct across levels. In fact, a vertex can only 
    conflict with the neighbors in its own level, but since we keep all 
    up-neighbors in a single data structure, we check all up-neighbors.
  
   In addition to checking the vertices which changed levels, 
   given a batch of insertions, $\batchinserts$, in~\cref{alg:explicit-coloring-insert}, 
   we iterate over all insertions in parallel (\cref{ec-ins:outer}) and check whether
   any insertions are between two vertices with the same color (\cref{ec-ins:check-conflict}). 
   Then, for each edge insertion $\{u, v\}$ between 
   two vertices with the same color, we arbitrarily select one vertex, $u$, to recolor itself (\cref{ec-ins:recolor}).
    $v$ selects a free color uniformly at random from its palette (\cref{ec-ins:recolor}). If $v$
    still conflicts with any of its up-neighbors (\cref{ec-ins:while-conflicting}), it
    recomputes its palette of free colors again by looking at the colors of 
    its up-neighbors and picks a color uniformly at random. 
    This process repeats until no vertices conflict with their
    up-neighbors in color.
    Finally, the remaining vertices which changed levels choose colors from their respective palettes 
    (\cref{ec-ins:leftover,ec-ins:leftover-conflict,ec-ins:left-recolor}).

\paragraph{Analysis}
Since each level has a unique palette and at most $(2 + 3/\lambda) (1+\eps)^{i}$
neighbors of a vertex $v$ can be on the same level as $v$ if $\ell(v) \in
\group_i$, $v$ has at least $(2+3/\lambda)(1+\eps)$
free colors that it can choose from its palette of
size $2 \cdot (2+3/\lambda)(1+\eps)^i$.
We first show that this strategy only requires $O(\alpha \log n)$ colors.
Part of the proofs of~\cref{lem:num-colors,lem:coloring-work} 
follow the analysis provided in~\cite{HNW20} but we present it for completeness.

\begin{lemma}\label{lem:num-colors}
At most $O(\alpha \log n)$ colors are required in our algorithm.
\end{lemma}

\begin{proof}
Each level $\ell \in \group_i$ has $O((1+\eps)^i)$ colors (assuming $\lambda$ is
constant). We showed in~\cref{lem:decomp-bound} that our coreness estimate 
is upper bounded by $(2+\eps)\core(v)$. This means that the largest
group index where a vertex is on a level in the group is 
$\log_{(1+\eps)}\left((2+\eps)\core(v)\right) + 1$. Hence,
$i = \max_{v}\left(\log_{(1+\eps)}(\core(v)) + \log_{(1+\eps)}(2+\eps) + 1\right)$.
Since 
each group contains $O(\log n)$
levels, the number of colors used is

\begin{align*}
&\sum_{i = 0}^{\log_{(1+\eps)}\left(\max_v(\core(v))\right) + \log_{(1+\eps)}(2+\eps) + 1} 2(2+3/\lambda)(1+\eps)^i \log n \\ 
&= O(\max_v(\core(v))\log n) = O(\alpha \log n).
\end{align*}
\end{proof}

We now show that our procedure requires $O(\batch\log^2 m)$ amortized work
using \plds.

\begin{lemma}\label{lem:coloring-work}
For a batch $\batch$, under the oblivious adversary assumption,\footnote{An oblivious adversary cannot see our algorithms outputs (i.e., they cannot see our coloring)
before determine the set of updates.} our coloring algorithm requires $O(|\batch| \log^2 n)$ 
amortized work, in expectation.\footnote{We can show the work to be $O(\batchsize \log^3 n)$ amortized
\whp somewhat tediously using the Chernoff bound. However, the bound requires an additional $O(\log n)$ factor of work
compared to the $O(\batchsize \log^2 n)$ amortized work bound in expectation.}
\end{lemma}

\begin{proof}
For a vertex 
that moves to a different level, the work to recolor it can be charged 
to the work of the PLDS update procedure. The bulk of this proof is devoted to proving
this fact. 

First, we show that the oblivious adversary cannot cause recolorings too often via
adversarial edge insertions between two vertices with the same color. The proof for this part is similar to the proof of Lemma 8 in~\cite{HNW20}. To show this,
we crucially rely on the fact that the adversary \emph{cannot} see the colors
of the vertices before they pick the updates.
For an edge insertion 
$(u,v)$ between two vertices, $u$ and $v$, that 
are on the same level $\ell \in g_i$ and causes a conflict, $u$ and $v$ have
at most $(2+3/\lambda)(1+\eps)^{i}$ neighbors (at the same level, using the same palette) but $2 \cdot 
(2+3/\lambda)(1+\eps)^i$ total colors in its palette. The algorithm 
arbitrarily picks one of the two endpoints, without loss of generality $u$, 
to recolor itself. 
$u$ has at least $(2+3/\lambda)(1+\eps)^i$ free colors
and it picks each color uniformly at random from these $(2+3/\lambda)(1+\eps)^i$ free colors.\footnote{$u$ may pick a random color multiple times, and we consider the 
palette that is last used by $u$ to pick its final color.} This means that
$u$ picks any particular color $c$ in its palette with 
probability at most $\frac{1}{(2 + 3/\lambda) (1+\eps)^{i}}$. 
The same argument holds for a vertex that moved to a new level and 
needs to be recolored. Since the adversary is oblivious, they have to guess
which color $u$ picked. Even assuming the much stronger assumption that the
adversary knows the colors of all vertices except 
$u$ (it does not in actuality), 
the adversary still only has at most a $\frac{1}{(2 + 3/\lambda) (1+\eps)^{i}}$
chance of picking $c$ ($u$'s color) and creating an edge insertion between $u$ and 
a vertex with color $c$. Thus in, 
expectation, the adversary must create
$(2 + 3/\lambda) (1+\eps)^{i}$ edge insertions incident to $u$ before
they pick one that conflicts with $u$'s color. The $O((2 + 3/\lambda) (1+\eps)^{i})$ 
cost of finding a color for $u$ can be amortized over these edge insertions.

For the vertices
that were recolored due to level movements, we can charge their cost to the
cost of moving levels. In expectation, the vertex tries at most two times (by what we showed above)
before it is successfully recolored, resulting in $O(1)$ total cost, in expectation.
Thus, the amortized update time is equal to the number of conflicts
and vertices that moved to a different level. 
This takes $O(\batchsize \log^2 n)$ work in expectation since there are $O(\batchsize\log^2n)$ edge flips and updates
and is the same as the \plds update time.

In addition to the cost of recoloring due to adversarial insertions,
recall that our batch-dynamic algorithm also requires multiple vertices
to keep picking colors uniformly at random until they pick unique colors
not occupied by their neighbors (\cref{alg:explicit-coloring-delete}-\cref{ec-del:recolor},
\cref{alg:explicit-coloring-insert}-\cref{ec-ins:recolor}). 
We need to show that this procedure
does not add too much additional cost to the cost of recoloring 
due to adversarial updates.
In fact, next we show that additionally picking random colors until all vertices pick
a non-conflicting color does not add additional work, asymptotically, \whp
It is easy to show that this running time also holds with high probability.
Given a set $X$
of vertices in level $\ell \in g_i$ that randomly picked the same color or
that moved to a different level, 
in expectation, after a round of recoloring, 
the number of vertices in $X$ that again result in conflicts is
$\frac{|X|}{(2 + 3/\lambda) (1+\eps)^{i}} \leq \frac{|X|}{(2 + 3/\lambda)} \leq \frac{|X|}{2}$. %
Since each vertex is independently picking a color, we can show via the 
Chernoff bound that with probability at most $\exp(\eps^2 |X|/6)$, more than $(1+\eps)\cdot \frac{|X|}{2}$ vertices need to pick colors again. Thus, when 
$|X| \geq c \log n$, for large enough constant $c > 0$, with high probability, the number of vertices that need to re-pick their colors
decreases by a factor of $2$. 
Then, each time we recolor, the number of vertices that obtain their final color decreases by a constant fraction \whp, and we can charge the cost 
of all subsequent recolorings to the first time we recolor the vertices.

When $|X| < c \log n$, 
the probability that any $c \log n$ consecutive
trials results in a vertex $v \in X$ picking
the same color is $(\frac{1}{2})^{c \log n} \leq \frac{1}{n^c}$. By the union bound
over the $\batchsize$ updates, the total probability that any two vertices 
conflict is $\frac{\batchsize}{n^c}$. We can pick $c \geq 3$ to obtain 
with high probability that each insertion results in $O(\log n)$ 
conflicts. Thus, the amortized update time is $O(\log n)$, \whp, 
since each edge gets charged $O(\log n)$ for up to two endpoints.
The total amortized work is $O(\batchsize\log^2 n)$ which can be charged
to the time necessary to perform the orientation algorithm.
By~\cref{thm:framework}, we obtain that the total work of our algorithm is $O(\batchsize \log^2 n)$, in expectation, $O(\batchsize \log^3 n)$ \whp
\end{proof}

\begin{lemma}\label{lem:coloring-depth}
Our coloring algorithm requires $O(\log^2 n \log \log n)$ depth per 
batch, with high probability.
\end{lemma}
\begin{proof}
All vertices that need to change colors pick their new colors 
independently in parallel, and picking a new color takes $O(1)$ depth. 
All vertices, first, in parallel, record 
the colors of their up-neighbors. Then, vertices, in parallel, pick
colors not occupied by their up-neighbors. 

In the remainder of the proof, 
we prove that we need $O(\log n)$ depth \whp
to resolve the conflicts resulting from multiple neighboring 
vertices picking the same color. 
The probability that a vertex conflicts with a neighbor 
is at most $\frac{1}{(2 + 3/\lambda) (1+\eps)^{i}} \leq \frac{1}{(2 + 
3/\lambda)}\leq \frac{1}{2}$. The probability that we still have conflicts after 
$c \log n$ tries, 
for some constant $c > 0$, is then at
most $(\frac{1}{2})^{c\log n} \leq 
\frac{1}{n^{c}} $. Picking $c \geq 2$ and applying a union bound over all vertices gives a polynomially small
probability that conflicts occur after $c \log n$ tries.
Thus, we need to randomly choose colors $O(\log n)$ 
times, \whp. 

Since the depth of picking a color for each 
vertex is $O(1)$, the total depth
for picking colors is $O(\log n)$ \whp This depth is additive to the depth of our orientation algorithm 
because we first use our \plds to move vertices 
to their final levels 
and then recolor the vertices. By~\cref{thm:framework},
the overall depth is 
$O(\log^2 n \log\log n)$ \whp. 
\end{proof}

\begin{theorem}
Our batch-dynamic $O(\alpha\log n)$-vertex coloring algorithm requires $O(|\batch| \log^2 n)$ amortized work,
$O(\log^2 n\log\log n)$ depth, \whp, using $O(m+n\log^2 n + \alpha\log n)$ space.
\end{theorem}

\begin{proof}
Our work and depth bounds hold by~\cref{lem:coloring-work} and~\cref{lem:coloring-depth}, respectively.
Finally, our coloring algorithm uses an additional space equal to the total number of colors used by the algorithm
to maintain the palettes. Thus, the algorithm requires an additional $O(\alpha\log n)$ space; by~\cref{thm:framework}, this
results in $O(m+n\log^2 n + \alpha\log n)$ space.
\end{proof}

\subsection{Implicit $O(2^{\alpha})$-Coloring}\label{sec:implicit}

In this section, we present an implicit $O(2^{\alpha})$-coloring algorithm,
where $\alpha$ is the current arboricity of the graph, after processing the most recent batch of updates.
As defined previously, an \emph{implicit} coloring is maintained by a set of data structures
whereby on \emph{queries} of one or more vertices, the data structures return a valid coloring
in the induced subgraph of the queried vertices.

To do this, we maintain a batch-dynamic version of the arboricity forest
decomposition structure of Henzinger \etal~\cite{HNW20}. We construct these forests
using the batch-dynamic Euler tour data structure of Tseng 
\etal~\cite{TsengDB19}. As in previous sections, we use the framework provided in~\cref{sec:framework}.
The sequential, dynamic algorithm of Henzinger \etal~\cite{HNW20} turns out to be somewhat 
tricky to parallelize. Specifically, Henzinger \etal~\cite{HNW20} prevents cycles in the
forests they create by sequentially inserting edges into one of two possible trees. In the batch-dynamic 
setting, since we are inserting multiple edges simultaneously, we need to run a cycle-breaking algorithm to 
split the cycles among the trees; such an algorithm is somewhat cumbersome to implement. Instead, 
we present a simpler version of their algorithm below that is much easier to parallelize, provided
an acyclic low out-degree orientation; our simpler algorithm provides the same guarantees.

\paragraph{Our Simplified 
Dynamic Arboricity Forest Decomposition Structure}
We first provide a simplified version of the arboricity forest decomposition
structure of~\cite{HNW20} here. We are able to simplify the structure since we assume an acyclic orientation algorithm, while
the arboricity decomposition structure in~\cite{HNW20} can use any orientation algorithm, not necessarily only
acyclic ones.
We also present some new proofs that our simplified structure still solves
the $O(2^\arb)$-coloring problem in the same work bounds as the structure 
presented in~\cite{HNW20}.
Then, we build on this simplified structure to design
our parallel batch-dynamic algorithm.

Provided an $\sigma$ out-degree orientation,
the key idea behind the arboricity forest decomposition structure 
of~\cite{HNW20} is to create $2\sigma$ \emph{undirected}
forests. However, we show here that $\sigma$ undirected forests is 
sufficient for this problem if the out-degree orientation is also acyclic.
Below, we present a simpler version of the algorithm using only $\sigma$ 
undirected forests via a simple lemma (\cref{lem:arboricity-decomposition-correctness}) we prove.

They use two different types of data structures to maintain the forests: 
the top tree data structures of~\cite{AHLT05} 
and an array for each node maintaining
which trees contain an outgoing edge of that node. We denote the array 
for node $v$ by $A_v$. Furthermore, we denote the $i^{\text{th}}$ forest by $F_i$.
The forests maintain the following invariants:

\begin{enumerate}
    \item There exists a unique root for each tree in each forest.
    \item For each $l \in \{0, \dots, \sigma - 1\}$ and each $v \in V$, no 
    forest $F_l$ contains two or more outgoing edges of $v$.
    \item No forest where $j \geq |\outnbr(v)|$ contains an outgoing edge of 
    $v$.
    \item $A_v[i] = 1$ (for $i \in [\sigma]$) if and only if $F_{i}$ contains 
    an outgoing edge of $v$. Otherwise, $A_v[i] = 0$.
\end{enumerate}

The forests support the following two operations:

\begin{enumerate}
    \item \textbf{Insert oriented edge:} A new 
    directed edge $(u, v)$ is inserted
    into the structure in $O(\log n)$ time. Let $d(u) = |\outnbr(u)|$ (where
    $\outnbr(u)$ is the out-degree of $u$ before the new edge insertion). 
    This is done by inserting the edge into $F_{d(v)}$ and setting $A_u[i] = 
    1$. The top tree allows this operation to be done in $O(\log n)$ time.
    The out-degree of $v$ is now $d(v) + 1$ and all invariants remain 
    satisfied.
    \item \textbf{Delete oriented edge:} A directed edge $(u, v)$ is 
    deleted from the structure in $O(\log n)$ time. We first find the 
    location of the edge in the forests. This can be done by maintaining
    pointers from edges to their respective locations in the forests. Let 
    $F_i$ be the forest that contains $(u, v)$. Delete $(u, v)$ from $F_i$.
    The top tree allows this operation to be done in $O(\log n)$ time. 
    Then, we find the tree with the largest index that contains an outgoing
    edge of $u$. (By our invariants, this should be $F_{d(u)-1}$ where 
    $d(u)$ is the out-degree of $u$ before the edge deletion.)
    Let $e$ be the outgoing edge of $u$ in $F_{d(u) - 1}$. Then, delete
    $e$ from $F_{d(u)-1}$ and insert $e$ into $F_i$. The top tree allows
    us to perform both operations in $O(\log n)$ time.
\end{enumerate}

Any edge flips from $(u,v)$ to $(v, u)$ 
in our orientation algorithm can be handled by first performing
an edge deletion of $(u, v)$ followed by an edge insertion of $(v, u)$ 
using the algorithms above.

We now prove the correctness of the simplified structure for acyclic
orientation algorithms.

\begin{lemma}\label{lem:arboricity-decomposition-correctness}
Let $\sigma$ be the maximum out-degree of our acyclic orientation algorithm.
Then, $F_0, \dots, F_{\sigma-1}$ provides an arboricity decomposition of 
the graph. 
\end{lemma}

\begin{proof}
This proof relies on proving two parts of the simplified
algorithm. First, we need to show
that the union of all forests in $F_0, \dots, F_{\sigma -1}$ gives all of the 
edges in the input graph. To show this first property, 
we need only show that no outgoing edge of 
any vertex $v \in V$ is in any forest $F_{j}$ where $j \geq d(v)$ and
$d(v)$ is the out-degree of vertex $v$. This directly follows from the
invariants. Furthermore, each inserted edge is inserted into at least 
one of the forests. 

Then, we need to show that no cycles exists in 
any of the forests (i.e., each $F_i$ is properly a forest). To do this, it 
is sufficient to show that no cycles exist in any $F_i$. We prove this via
contradiction. Suppose that a cycle exists in $F_i$. By our invariant, 
this means that no vertex $v$ incident to two edges, $(v, w)$ and $(v, u)$, 
in the cycle has both $(v, w)$ and $(v, u)$ oriented outwards from $v$. Otherwise, this
would violate the invariant that at most one outgoing edge of $v$ is in 
any forest. Thus, for every vertex in the cycle, one of the incident 
edges must be directed out and one directed in, in the orientation produced
by the orientation algorithm. This is a contradiction to the acyclicity of 
the orientation algorithm.
\end{proof}

\paragraph{Batch-Dynamic Algorithm} We implement the batch-dynamic 
algorithm in the following way. We implement the trees using 
the batch-dynamic Euler tour trees of Tseng \etal~\cite{TsengDB19}. 
These trees allow inserts/deletes and distance-to-root 
operations in $O(|\batch| \log n)$ work and $O(\log n)$ 
depth \whp \footnote{The high probability bound directly follows
from the high probability bounds of parallel skip-lists.} 
Furthermore, we maintain a parallel hash table, $T$, which 
contains the edges as keys and pointers to the tree containing each edge
as the values.

In this algorithm,
let our batch of edges to insert into our forests be the
set of edge insertions, deletions, and edge flips. To implement $\batchedgeflip$,
we create two sets of updates per flipped edge: an edge deletion
and an edge insertion. Then, following the framework in~\cref{sec:framework},
we first process the deletions and 
then the insertions. 

For each oriented deleted edge, we check, in parallel, in 
the hash table $T$ for the location of each edge (which tree 
each edge is in). %
Then, we perform, in parallel,
deletions of the edges in the respective parallel Euler tour
trees. All of this can be done in $O(|\batch| \log^3 n)$
amortized work
and $O(\log n)$ depth \whp 
The work results from performing Euler tree operations on the set of 
edges in $\batchinserts$ and $\batchdeletes$. There are $O(\batchsize\log^2 n)$ updates in these batches and each Euler tree operation requires
$O(\log n)$ work; thus, our total work is $O(\batchsize \log^3 n)$.
We can perform all updates to our Euler trees in parallel, hence, the total depth
is the depth of performing these updates, $O(\log n)$ \whp

For each vertex, we maintain 
the number of edges deleted from its trees as well as the 
trees from which edges are deleted. This can be done in $O(|\batch|
\log^2 n)$ amortized work and $O(\log n)$ depth \whp
Finally, we find, in parallel, the $X_i$ outgoing edges of $v_i$
in the last $X_i$ trees that contain an outgoing edge of $v_i$,
where $X_i$ is the number of edges that were deleted from $v_i$'s
trees. In parallel, we arbitrarily pick a unique slot for each edge and assign
it to to its respective empty slots in the trees. This last step can 
also be done in $O(|\batch| \log^3 n)$ amortized work and $O(\log n)$
depth \whp

For the insertion edges, we first sort the edges by their outgoing
endpoint. Then, we determine how many edges we are inserting in each
$v_i$'s trees by doing a parallel count. Then, finally, in
parallel, we insert each $v_i$'s edge into the next $X_i$ empty
trees where $X_i$ is the number of edge inserts that are oriented
out from $v_i$. All of this requires $O(|\batch| \log^3 n)$ amortized
work and $O(\log n)$ depth. 

The correctness of our procedure follows 
from~\cref{lem:arboricity-decomposition-correctness}. 
Altogether, we obtain the following theorem of our batch-dynamic algorithm
implicit $O\left(2^{\alpha}\right)$-coloring, using~\cref{thm:framework}.

\begin{theorem}\label{thm:implicit-coloring}
For a batch $\batch$, our batch-dynamic implicit coloring algorithm provides a 
$O\left(2^{\alpha}\right)$-coloring in $O(|\batch| \log^3 n)$
amortized work and $O(\log^2 n)$ depth \whp for updates,
and $O(Q\alpha \log n)$ work and $O(\log n)$ depth, \whp, for $Q$ 
queries, using $O(n\log^2 n + m)$ space.%
\end{theorem} 

\begin{proof}
The work and depth follow from our above arguments and~\cref{thm:framework}.
For queries, we parallelize the algorithm of~\cite{HNW20}. For a set of $Q$ 
vertices, for each vertex, we find the set of forests $[1, \dots, d(v)]$ where 
$d(v) = O(\alpha)$, containing each of the outgoing edges of $v$. 
As in~\cite{HNW20}, we let $p_v(i)$ be the parity of the distance for the
$i$-th Euler tree. Then, in parallel, 
we determine the distance of $v$ from the root of the Euler
tree in each of these forests.
If the distance is odd, we assign $p_v(i) = 1$ and $p_v(i) = 0$ otherwise. 
The color given to $v$ is then $(p_v(1), \dots, p_v(d(v))) \in \{0,1\}^{O(\alpha)}$.
Querying the Euler trees require $O(\log n)$ work per tree query. We have
$O(Q\alpha)$ total queries, resulting in $O(Q\alpha\log n)$ total work. 
Then, processing all queries simultaneously requires $O(\log n)$ depth \whp

Finally, the extra space required is the space to store the extra Euler trees
and the hash table $T$. $T$ uses $O(m)$ space. All of the Euler trees store $O(m)$
edges; thus, the total additional space used is $O(m)$.
\end{proof}

\section{Conclusion}

We design the first shared-memory, multi-core parallel batch-dynamic level data
structure that returns a $(2+\eps)$-approximation
for the \kc decomposition problem, drawing inspiration from the sequential level data structures
of Bhattacharya \etal~\cite{BHNT15} and Henzinger \etal~\cite{HNW20} which were used for dynamic
densest subgraphs and dynamic low out-degree orientation, respectively. 
Our algorithm achieves $O(\log^2 m)$ amortized work
and has $O(\log^2 m \log\log m)$ depth. We also present a proof of the
$(2+\eps)$-factor of approximation for our data structure, a new proof that is
also applicable (with a simple change) to the original sequential level data structures
of Bhattacharya \etal~\cite{BHNT15} and Henzinger \etal~\cite{HNW20}.

In addition to our batch-dynamic \kc decomposition results, we also 
give a batch-dynamic algorithm for maintaining an $O(\alpha)$ out-degree 
orientation, where $\alpha$ is the \emph{current} arboricity 
of the graph. We demonstrate the usefulness of our low out-degree orientation
algorithm by presenting a new framework to formally study 
batch-dynamic  algorithms in bounded-arboricity graphs.
Our framework obtains new provably-efficient
parallel batch-dynamic algorithms for 
maximal matching, clique counting, and vertex coloring.

We perform extensive experimentation of our parallel batch-dynamic 
\kc decomposition algorithm on large real-world data sets that 
show that our PLDS is not only theoretically
but also practically efficient. Our experiments tested error vs.\ runtime,
batch size vs.\ runtime, number of hyper-threads vs.\ runtime, and space vs.\
error. We also tested the sensitivity of our implementation to the various tunable
parameters of our algorithm. Finally, we tested our algorithm against six other
benchmarks on $11$ real-world graphs, including graphs orders of magnitude
larger than previously studied by other dynamic algorithms. We see an
improvement in performance against all other benchmarks in our experiments.
Specifically, we achieve speedups of up to
$114.52\times$ against the best parallel implementation, up to $544.22\times$
against the best approximate sequential algorithm, and up to $723.72\times$
against the best exact sequential algorithm. Such speed-ups exceed the expected
speed-up gained from parallelism alone (since we only use $60$ hyper-threads)
and are also due to the theoretical
improvements of our algorithm as well as our choice of heuristic optimizations.

An interesting open problem is to design a parallel batch-dynamic algorithm that
is space-efficient (uses linear space), without incurring additional costs in depth.

\section*{Acknowledgements}
This research is supported by NSF GRFP
\#1122374,
DOE Early Career Award \#DE-SC0018947,
NSF CAREER Award \#CCF-1845763, Google Faculty Research Award, Google Research Scholar Award, FinTech@CSAIL Initiative, DARPA
SDH Award \#HR0011-18-3-0007, and Applications Driving Architectures
(ADA) Research Center, a JUMP Center co-sponsored by SRC and DARPA.

\bibliographystyle{ACM-Reference-Format}
\bibliography{ref}

\end{document}